\documentclass[a4paper,11pt,dvipsnames]{article}
\usepackage[utf8]{inputenc}
\usepackage[margin=1in]{geometry}

\usepackage{quantikz}
\usetikzlibrary{external}
\usepackage{soul}
\usepackage{graphicx}
\usepackage{amsmath,amssymb,amsthm,mathrsfs,amsfonts,dsfont}
\usepackage{braket}
\usepackage{xcolor} %
\usepackage[colorlinks=true,citecolor=MidnightBlue,linkcolor=BrickRed,breaklinks=true]{hyperref}
\usepackage[capitalise]{cleveref}
\crefname{figure}{Figure}{Figures}

\usepackage{relsize}
\usepackage[title]{appendix}
\usepackage{caption}
\usepackage{subcaption}
\usepackage{multirow}

\usepackage{rotating}
\usepackage{array}
\newcolumntype{P}[1]{>{\centering\arraybackslash}p{#1}}

\usepackage[intoc]{nomencl}
\makenomenclature
\usepackage{etoolbox}
\renewcommand\nomgroup[1]{%
  \item[\bfseries
  \ifstrequal{#1}{A}{Function Spaces}{%
  \ifstrequal{#1}{B}{Finite Element Arrays, Domains, and Spaces}{%
  \ifstrequal{#1}{C}{Quantum Computing}{%
  \ifstrequal{#1}{D}{Other Symbols}{
  }}}}%
]}

\newtheorem{thm}{Theorem}[section]
\newtheorem{defn}[thm]{Definition}
\newtheorem{prop}[thm]{Proposition}
\newtheorem{cor}[thm]{Corollary}
\newtheorem{lem}[thm]{Lemma}
\newtheorem{rem}[thm]{Remark}

\usepackage[many]{tcolorbox}
\usepackage{capt-of}
\newtcolorbox[auto counter, number within=section, use counter=thm]{examplebox}[2][]{%
  colback=black!5!white,      %
  colframe=blue!25!black,         %
  coltitle=white,         %
  fonttitle=\bfseries,    %
  coltext=black,          %
  boxrule=0.5mm,          %
  width=\linewidth,       %
  before=\par\medskip\noindent,  %
  after=\par\medskip\noindent,            %
  title={Example~\thethm: #2}, %
  breakable,
  enhanced jigsaw,
  before upper={\addtocounter{thm}{-1}\refstepcounter{thm}},
  #1,
}

\usepackage{bm}
\usepackage{mathrsfs}
\usepackage{sectsty}

\usepackage{stmaryrd}
\usepackage{braket}

\DeclarePairedDelimiter\abs{\lvert}{\rvert}%
\DeclarePairedDelimiter\norm{\lVert}{\rVert}%
\newcommand{\ketbra}[2]{\ket{#1}\!\bra{#2}}

\newcommand{\NOT}{\mathrm{NOT}}
\newcommand{\PREP}{\textsc{prep}}

\newcommand{\numel}{\textnormal{\texttt{numel}}} %
\newcommand{\numnp}{\textnormal{\texttt{numnp}}} %
\newcommand{\nen}{\textnormal{\texttt{nen}}} %

\newcommand{\numbp}{\textnormal{\texttt{numbp}}} %

\newcommand{\el}{\textnormal{\texttt{el}}} %
\newcommand{\nn}{\textnormal{\texttt{nn}}} %
\newcommand{\np}{\textnormal{\texttt{np}}} %

\newcommand{\mesh}{\mathfrak{mesh}} %
\newcommand{\IX}{\mathrm{IX}}

\newcommand{\zpad}{\operatorname{\mathbf{zpad}}}
\newcommand{\LCU}{\operatorname{LCU}}

\newcommand{\uoi}{\hat{A}^\parallel}
\newcommand{\uoip}{\hat{A}^{\parallel,p}}
\newcommand{\uoim}[1]{\hat{A}^{\parallel,#1}}

\DeclareMathOperator*{\bA}{\pmb{\mathbf{A}}}

\newcommand{\bracfrac}[2]{\left(\frac{#1}{#2}\right)}

\newcommand{\bbR}{\mathbb{R}}

\newcommand{\bbZ}{\mathbb{Z}}

\newcommand{\bbC}{\mathbb{C}}
\newcommand{\bbP}{\mathbb{P}}
\newcommand{\bbD}{\mathbb{D}}

\newcommand{\ellb}{{\bm{\ell}}}

\newcommand{\betab}{\bm{\beta}}

\newcommand{\bb}{\bm{b}}

\newcommand{\eb}{\bm{e}}
\newcommand{\fb}{\bm{f}}

\newcommand{\jb}{\bm{j}}

\newcommand{\nb}{\bm{n}}

\newcommand{\sbm}{\bm{s}}

\newcommand{\ub}{\bm{u}}

\newcommand{\xb}{\bm{x}}

\newcommand{\Ab}{\bm{A}}

\newcommand{\Jb}{\bm{J}}

\newcommand{\Xb}{\bm{X}}

\newcommand{\Dc}{{\mathcal D}}

\newcommand{\Fc}{{\mathcal F}}

\newcommand{\Hc}{{\mathcal H}}

\newcommand{\Kc}{{\mathcal K}}
\newcommand{\Lc}{{\mathcal L}}
\newcommand{\Mc}{{\mathcal M}}
\newcommand{\Nc}{{\mathcal N}}
\newcommand{\Oc}{{\mathcal O}}

\newcommand{\Qc}{{\mathcal Q}}

\newcommand{\Uc}{{\mathcal U}}

\newcommand{\zerob}{{\boldsymbol{0}}}

\newcommand{\xib}{\boldsymbol{\xi}}

\newcommand{\Phib}{\bm{\Phi}}

\newcommand{\lambdab}{\bm{\lambda}}

\newcommand{\Psib}{\bm{\Psi}}

\usepackage[
	backend=biber,
	style=numeric,  %
    giveninits=true,      %
	sorting=none,
	isbn=false,
	doi=true,
	url=true,
	maxbibnames=99,   %
]{biblatex}
\begin{document}

\begin{center}
    {\textbf{\Large{A Quantum Algorithm for the Finite Element Method}}}\\
    \vspace{0.21in}

    {\large Ahmad M. Alkadri}$^{\hyperlink{email-alkadri}{*}}$\\
    {\small 
    {\emph{Department of Chemical \& Biomolecular Engineering, University of California, Berkeley,}\\
    \emph{Berkeley, California 94720, USA}
    }
    }\\[1em]

    {\large Tyler D. Kharazi}$^{\hyperlink{email-kharazi}{\dag}}$\\
    {\small 
    \emph{Department of Chemistry, University of California, Berkeley,}\\
    \emph{Berkeley, California 94720, USA}
    }\\[1em]

    {\large K. Birgitta Whaley}$^{\hyperlink{email-whaley}{\ddag}}$\\
    {\small
    \emph{Department of Chemistry, University of California, Berkeley,}\\
    \emph{Berkeley, California 94720, USA}\\
    \emph{and Berkeley Center for Quantum Information and Computation, University of California, Berkeley,}\\
    \emph{Berkeley, California 94720, USA}
    }\\[1em]

    {\large Kranthi K. Mandadapu}$^{\hyperlink{email-mandadapu}{\S}}$\\
    {\small
    \emph{Department of Chemical \& Biomolecular Engineering, University of California, Berkeley,}\\
    \emph{Berkeley, California 94720, USA}\\
    \emph{and Chemical Sciences Division, Lawrence Berkeley National Laboratory,}\\
    \emph{Berkeley, California 94720, USA}
    }\\[1em]
\end{center}

\vspace{13pt}

\begin{abstract}
    The finite element method (FEM) is a cornerstone numerical technique for solving partial differential equations (PDEs).
    Here, we present \textbf{Qu-FEM}, a fault-tolerant era quantum algorithm for the finite element method. 
    In contrast to other quantum PDE solvers, Qu-FEM preserves the geometric flexibility of FEM by introducing two new primitives, the \emph{unit of interaction} and the \emph{local-to-global indicator matrix}, which enable the assembly of global finite element arrays with a constant-size linear combination of unitaries. 
    We study the modified Poisson equation as an elliptic problem of interest, and provide explicit circuits for Qu-FEM in Cartesian domains. 
    For problems with constant coefficients, our algorithm admits block-encodings of global arrays that require only $\tilde{\Oc}\left(d^2 p^2 n\right)$ Clifford+$T$ gates for $d$-dimensional, order-$p$ tensor product elements on grids with $2^n$ degrees of freedom in each dimension, where $n$ is the number of qubits representing the $N=2^n$ discrete grid points.
    For problems with spatially varying coefficients, we perform numerical integration directly on the quantum computer to assemble global arrays and force vectors. Dirichlet boundary conditions are enforced via the method of Lagrange multipliers, eliminating the need to modify the block-encodings that emerge from the assembly procedure. 
    This work presents a framework for extending the geometric flexibility of quantum PDE solvers while preserving the possibility of a quantum advantage.
\end{abstract}
\vspace{15pt}

\noindent\rule{4.6cm}{0.4pt}

\noindent\small{%
	\hypertarget{email-alkadri}{$*$\,\href{mailto:ahmadalkadri@berkeley.edu}{ahmadalkadri@berkeley.edu}}\\
	\hypertarget{email-kharazi}{$\dag$\,\href{mailto:kharazitd@berkeley.edu}{kharazitd@berkeley.edu}}\\
	\hypertarget{email-whaley}{$\ddag$\,\href{mailto:whaley@berkeley.edu}{whaley@berkeley.edu}}\\
	\hypertarget{email-mandadapu}{$\S$\,\href{mailto:kranthi@berkeley.edu}{kranthi@berkeley.edu}}
}

\vspace{25pt}

 \newpage
\small \tableofcontents
\vspace{08pt}

\tikzexternaldisable %

\newpage
\section{Introduction}
\label{sec:intro}

Many problems of interest to engineers and scientists are governed by partial differential equations (PDEs). These include the Navier-Cauchy equations that govern the deformation of linear elastic solids in mechanical structures, the Navier-Stokes equations that govern the flow of Newtonian fluids, the reaction-diffusion equation for reactor design, Maxwell's equations for electromagnetism, the Fokker-Planck equation (FPE) for the evolution of probability density functions in statistical mechanics, and even the Schr\"{o}dinger equation in quantum mechanics~\cite{Muskhelishvili1977,thorne2017modern,leal2007advanced,bird2006transport,jackson2021classical,nielsen2010quantum}. 
A great deal of classical computing resources have been dedicated to the numerical solution of PDEs, since their predictive capability can be leveraged by engineers for design, and by scientists to gain physical insights into systems modelled by the PDEs. Even with access to classical high-performance computing (HPC) resources, however, the direct solution of many PDEs remains intractable. 

Problems where direct numerical solution (DNS) of PDEs using classical resources may be infeasible can be broadly characterized into two types: high-dimensional problems, and low-dimensional problems where one requires high-precision calculations. High-dimensional problems include the FPE and the backwards Kolmogorov equation (BKE), which describe the evolution (forward and backwards in time, respectively) of a probability density for a system with $\eta$ particles in $d$-dimensions~\cite{risken1996fokker,E2006,hasyim2022supervised}.
Discretizing space into $N$ grid points in each dimension to solve these PDEs requires a number of grid points that scales exponentially with the number of particles as $N^{d\eta}$, quickly rendering direct simulation of the FPE and BKE infeasible. 

One example of a low-dimensional problem that requires high-precision computing is simulation of the Navier-Stokes (NS) equations in the turbulent regime, where it is crucial to adequately resolve the smallest scales of turbulence. In turbulent flows, these smallest scales are known as the \textit{Kolmogorov microscales}, representing the scales at which viscous dissipation occurs~\cite[Chapter~5]{pope2000turbulent}. The Kolmogorov microscales $L_\eta$ scale with the Reynolds number $(Re)$ of the flow as $L_\eta \sim Re^{-3/4}$~\cite[Chapter~6]{frisch1995turbulence}.
This means that to run a Finite Element Method (FEM) simulation of the (three-dimensional) NS equations with a mesh that is fine enough to capture the Kolmogorov microscales requires a total number of elements that scales as $(Re^{3/4})^3 = Re^{9/4}$.
The least memory-intensive finite element for solving the NS equations is the ``MINI element'', which contains a total of $19$ degrees of freedom (DOF)~\cite{arnold1984stable,Fortin1993}. Additionally, a typical Reynolds number for fluid flow over the wing of a commercial aircraft is $Re \sim 10^7$.
Assuming double precision (i.e., $8$ bytes) for each DOF, to simply store the solution at a single timestep of a turbulent simulation over an airplane wing requires memory on the order of 
\begin{align*}
    (10^7)^{9/4} \,\text{elements} \cdot 19 \,\frac{\text{DOFs}}{\text{element}} \cdot 8 \,\frac{\text{bytes}}{\text{DOF}} \approx 10^3 \,\text{petabytes}
    \,.
\end{align*}
The Frontier supercomputer has a total system memory of $9.2$ petabytes---one of the highest among all known capacities for HPC~\cite{ornl_frontier_user_guide,hpcwire_frontier}---still orders of magnitude smaller than what would be required to simply store the solution, rendering DNS practically infeasible. For this reason, current classical methods circumvent DNS by using alternate methods such as Reynolds-averaged Navier-Stokes (RANS), Large Eddy Simulation (LES), or hybrid RANS-LES methods (such as Detached Eddy Simulation), which significantly reduce computational costs by using effective models for the smallest scales instead of explicitly resolving them~\cite{pope2000turbulent}.
An amplitude encoding of the solution on a quantum computer, however, utilizes a logarithmic number of resources, reopening the potential for DNS of the NS equations in the turbulent regime.

The potential for quantum computers to efficiently carry out the DNS of PDEs that would otherwise require an astronomical (and financially prohibitive) amount of classical resources has spurred the development of several quantum algorithms for PDEs beyond the Schr\"{o}dinger equation of quantum mechanics.
These include quantum adaptations of classical numerical methods such as finite difference methods~\cite{childs2021high} and spectral methods~\cite{childsQuantumSpectralMethods2020}, as well as novel quantum approaches that use ``Schr\"{o}dingerization''~\cite{Jin2024Schrodingerization,jin2025quantumpreconditioningmethodlinear} and Koopman-von Neumann~\cite{lloyd2020quantumalgorithmnonlineardifferential,An2022} theory to map classical PDEs to a Hamiltonian simulation problem. To the authors' knowledge, however, every quantum algorithm for solving classical PDEs in the current literature is limited to rectangular geometries (see Table~\ref{tab:pde-methods}). 
The FEM is one of the most natural paths to solving problems of practical interest---especially low-dimensional ones---in arbitrary geometries.
The FEM can easily handle complex geometries, irregular shapes, and boundaries, making it ideal for real-world problems.
However, explorations of a quantum adaptation of the FEM have been limited so far in the literature.

\begin{table}
	\centering
	\renewcommand{\arraystretch}{1.1} %
	\footnotesize
	\begin{tabular}{|P{0.8cm}|P{2cm}|P{2.7cm}|P{1.85cm}|P{2cm}|P{2cm}|P{1.5cm}|}
		\hline
		& \textbf{Numerical Method} & \textbf{PDE Type} & \textbf{Boundary Conditions} & \textbf{Geometric Flexibility} & \textbf{Time Complexity $\tilde{\Oc}(d,p,N)$} & \textbf{Error Scaling $\epsilon(d,p,N)$} \\
		\hline
		\multirow{4}{0.8cm}{\centering \rotatebox{90}{\parbox{1.3cm}{\centering Classical\\$u$}}} 
		& Spectral \cite{boyd2001chebyshev} & \multirow{4}{1.2cm}{\centering Any} & \multirow{4}{1.85cm}{\centering Any} & Limited & $d N^{d}\log N$ & $\exp(-N)$ \\[0em]
		\cline{2-2}\cline{5-7}
		& FDM \cite{strikwerda2004finite} & & & Limited & \multirow{3}{2cm}{\centering $ p^{d+1} N^{d+1}$} & \multirow{3}{1.5cm}{\centering $ N^{-p/d}$} \\[0em]
		\cline{2-2}\cline{5-5}
		& FVM \cite{eymard2000finite} & & & Moderate/high & & \\[0em]
		\cline{2-2}\cline{5-5}
		& FEM \cite{ciarlet2002finite,brenner2008mathematical} & & & Very high & & \\[0em]
		\hline
		\multirow{5}{0.8cm}{\centering \rotatebox{90}{\parbox{3.5cm}{\centering Quantum\\$\ket{u}$}}}
		& Spectral \cite{childs2021high} & Elliptic & \multirow{5}{1.85cm}{\centering \vfill Periodic, Dirichlet, Neumann} & & $d N^5$ & $\exp(-N)$ \\[0em]
		\cline{2-3}\cline{6-7}
		& FDM \cite{kharazi2024explicitblockencodingsboundary} & Elliptic & & Rectangular & $d^2 p N^2$ & \multirow{4}{1.5cm}{\centering \vfill$N^{-p/d}$} \\[0em]
		\cline{2-3}\cline{6-6}
		& FVM \cite{fillion2019simple} & Hyperbolic & & & \parbox{1.5cm}{\centering $d N^2$ \\ ($p \equiv 1$)} & \\[0em]
		\cline{2-3}\cline{5-6}
		& \multirow{2}{2cm}{\centering Qu-FEM [This work]} & \parbox{2.5cm}{\centering \vspace{2pt} Elliptic\\ (const. coeffs.)} & & \multirow{2}{2cm}{\centering edge-to-edge rectangular tilings} & $d^3 p^2 N^2$ & \\[0em]
		\cline{3-3}\cline{6-6}
		& & \parbox{2.75cm}{\centering \vspace{2pt} Elliptic\\ (black-box\footnotemark{} coeffs.)} & & & $d^3 p^{2d} N^2 $ & \\[0em]
		\hline
	\end{tabular}
	\caption{%
		Comparison of various classical and quantum numerical methods for solving PDEs. Here, $N = 2^n$ is the number of grid points per spatial direction, $d$ is the physical dimension, and $p$ is the order of the element/stencil.
		The scalings are generated assuming: a condition number $\kappa$ of order $\Oc(N^2)$, a classical conjugate gradient linear systems solver with time complexity $\Oc(N^d \sqrt{\kappa})$, and a QLSP solver with a time complexity of $\tilde{\Oc}(d \kappa)$ queries to the block-encoding of the differential operator. The $\tilde{\Oc}$ notation hides any $\operatorname{polylog}$ factors (see Theorems~\ref{thm:query-complexity-Assembly-dD-degree-p} and~\ref{thm:assembly-via-numerical-integration-d-dims} for a full accounting of the complexity scalings).
		Note that these scalings do not account for the effects of preconditioning, adaptive meshing, or the dependence of the scalings on derivatives of the solution.
	}
	\label{tab:pde-methods}
\end{table}

\footnotetext{The exponential scaling in dimension $p^{2d}$ is due to the cost of performing arbitrary state preparation in $d$-dimensions. If additional structure is present, this scaling could be made polynomial (see Ref.~\cite{kharazi2024explicitblockencodingsboundary} for an example).}

Following Harrow, Hassidim and Lloyd's~\cite{harrow2009quantum} (HHL) algorithm for solving the quantum linear system's problem (QLSP)\footnote{Given a system of linear equations $A\xb = \bb$, a QLSP solver prepares a normalized state $\ket{x}$ proportional to the solution $\xb$ using the normalized vector $\ket{b}$ and some access model for the matrix $A$ as inputs. See Sections~\ref{sec:background} and~\ref{sec:block-encoding-framework} for a brief summary of quantum computing notation.} 
in logarithmic time (though with several caveats; see~\cite{aaronson2015read}), Clader et al.~\cite{clader2013preconditioned} explored the FEM as a potential application of the HHL algorithm that could yield a quantum advantage. The authors of~\cite{clader2013preconditioned} present a methodology that uses a one-sparse access oracle to build the FEM arrays, a quantum implementation of the sparse approximate inverse (SPAI) preconditioner, and the HHL algorithm to calculate the radar cross section of a metallic scattering region (by solving Maxwell's equations). 
Claims of a quantum advantage in~\cite{clader2013preconditioned}, however, were challenged by 
Montanaro and Pallister~\cite{montanaro2016quantum}, who used estimates for the condition numbers and convergence rates that emerge from FEM theory to analyze the performance of a QLSP solver applied to the FEM. Montanaro and Pallister find that quantum speedups that arise from such an application can become diminished when performance is analyzed with respect to the desired solution accuracy $\epsilon$ rather than the system size $N$. Nevertheless, they conclude that:\\
\indent ``\textit{{[...]} there are still two types of problem where quantum algorithms for the FEM could achieve a significant advantage over classical algorithms: those where the solution has large higher-order derivatives, and those where the spatial dimension is large.}''\\
We note that in addition to quantum advantage in terms of time complexity, problems where storage of the solution requires an exponential amount of resources---such as DNS of the NS equations at the Kolmogorov length-scale outlined above---constitute another class of ``quantum advantage'' through space complexity.

More recently, Arora et al.~\cite{arora2024implementation} utilized classical FEM procedures to obtain a unitary decomposition of the stiffness matrix and employed the variational quantum linear solver (VQLS) to introduce a method for noisy intermediate-scale quantum (NISQ) computers. 
Deiml and Peterseim~\cite{deiml2024quantum} present a quantum algorithm for solving second-order elliptic PDEs on Cartesian grids that makes use of the BPX preconditioner~\cite{bramble1990parallel} to transform the linear system into a well-conditioned one suitable for quantum computation. The use of a preconditioner is particularly notable, as it is important to maintain a bounded condition number to achieve quantum advantage in low dimensions.
Quantum algorithms for preconditioning have remained relatively underexplored, however, as most classical preconditioning methods rely on explicit operations that are challenging to implement efficiently on quantum computers. In addition to block-encoding finite element arrays in more general geometries, generalizing preconditioners to also fit these geometries presents a challenge.

In this article, we present a framework for implementing the Finite Element Method (FEM) on a fault-tolerant quantum computer---a quantum algorithm that we abbreviate as ``Qu-FEM''. 
Classically, the FEM forms global arrays via an \textit{assembly process} from local elemental contributions~\cite{ciarlet2023finite,papadopoulos2015280a,donea2003finite}. Since the elemental contributions depend only on information that is local to the element, this assembly process can be efficiently parallelized, which has made FEM well-suited for HPC~\cite{brenner2008mathematical,bramble1990parallel}. Here, we seek to emulate this paradigm of parallel assembly on a quantum computer. 

To that end, the Qu-FEM framework utilizes a new device that we call ``the unit of interaction'' (see Definition~\ref{defn:unit-of-interaction}) to obtain efficient block-encodings of finite element arrays for high-dimensional problems with constant elemental contributions. For low-dimensional problems---where the elemental stiffness matrix can be non-constant---we utilize an operator that is closely related to the unit of interaction that we name ``the local-to-global node number indicator matrix'' (see Definition~\ref{defn:local-to-global-node-number-indicator-matrix}). By performing numerical integration directly on the quantum computer, we provide an efficient solution to the general assembly problem 
(Definition \ref{def:general-assembly-problem}) in low dimensions. 
As we introduce each of these concepts, we demonstrate that each component of Qu-FEM can be efficiently implemented for domains that are edge-to-edge rectangular tilings (i.e., rasterized and Cartesian domains).
This matches the geometric flexibility of the quantum algorithm for solving classical PDEs presented by~\cite{kharazi2024explicitblockencodingsboundary}. 
The Qu-FEM framework, however, makes use of oracles that are \textit{a priori} agnostic to the particular geometry of the problem, rendering it a good point of departure for implementing a quantum FEM on more arbitrary geometries.

The remainder of this article is structured as follows. In Section~\ref{sec:background}, we summarize the essential notation, definitions, and results from both the quantum computing and finite element literature that are used in this article. Additional background is presented in Appendices~\ref{sec:analysis-of-FEM} and~\ref{sec:block-encoding-framework} for convenience. In Section~\ref{sec:assembly-of-global-arrays-linear}, we introduce the main tools used to assemble block-encodings of finite element arrays: the unit of interaction (Definition~\ref{defn:unit-of-interaction}) and the local-to-global node number indicator matrix (Definition~\ref{defn:local-to-global-node-number-indicator-matrix}). We also present explicit block-encodings for finite element arrays with constant elemental contributions, and example boxes that carry out sample computations. 
In Section~\ref{sec:numerical-integration}, we give an algorithm for the quantum assembly of finite element arrays with non-constant elemental contributions by performing numerical integration directly on the quantum computer. 
In Section~\ref{sec:constraints}, we demonstrate how to enforce Dirichlet boundary conditions through a quantum algorithm for the method of Lagrange multipliers. 
Additionally, we also assemble the boundary force vector for enforcing non-homogeneous Neumann boundary conditions, as well as the body force vector for non-homogeneous PDEs.
We end with a numerical demonstration of the Qu-FEM algorithm in Section~\ref{sec:demo}.

\section{Background}
\label{sec:background}

Throughout the article, we use notation that is prominent in both the quantum computing and the classical finite element literature. In this section, we summarize the pertinent definitions, notation, and results from both literatures for convenience.
We also provide additional background on the theoretical analysis of FEM in Appendix~\ref{sec:analysis-of-FEM}, and additional background on the quantum block-encoding framework in Appendix~\ref{sec:block-encoding-framework}.

\subsection{Quantum Computing Notation}
\label{sec:notation}

We will denote quantum states and operators using Dirac notation (also called `bra-ket' notation)~\cite{nielsen2010quantum}. In particular, a state $\ket{\psi}_n$ will refer to a (usually normalized) vector in the space $\bbC^N$, where $N := 2^n$ for some number of qubits $n$. Tensor products of states and operators are denoted using the Kronecker product ``$\otimes$'' (e.g., $\ket{\psi}_n \otimes \ket{\phi}_m$ or $\ket{\psi}_n \otimes U_A$ for some operator $U_A$), though we sometimes omit the tensor product symbol between states for convenience (e.g., $\ket{\psi}_n \ket{\phi}_m$). 
For any integer $k$, we denote the set of non-negative integers less than $k$ by $[k] := \{0, 1, \dots, k-1\}$. 
For a matrix $A$, $\|A\| := \sup_{\braket{\psi | \psi} = 1} \|A\ket{\psi}\|$ is the operator norm induced by the vector $2$-norm $\|\ket{\psi}\| := \sqrt{\braket{\psi | \psi}}$ on the Hilbert space that $\ket{\psi}$ originates from (this is also called the ``spectral norm''). We denote the max norm as $\|A\|_\mathrm{max} := \max_{ij} |A_{ij}|$. Additionally, we write $\norm{A}_1 := \sum_{ij} \abs{A_{ij}}$ for the entrywise $\ell_1$-norm (not to be confused with the induced $1$-norm).

Our algorithm is designed to run on a fault-tolerant quantum computer, for which the Clifford $+$ $T$ gates form a universal gate set for quantum computing~\cite{nielsen2010quantum} (see Figure~\ref{fig:Clifford+T-gate-set}). With stabilizer quantum error correction codes (such as the surface code) in mind, we anticipate that circuits designed with these gates will be most compatible with error correction protocols in the era of fault-tolerant quantum computing \cite{wang2024comprehensive}. 
In particular, the Toffoli gate, and other multi-control gates, will be utilized extensively in our constructions. Ref.~\cite{nielsen2010quantum} gives a circuit for the Toffoli gate in terms of Clifford $+$ $T$ gates, which we summarize in Figure~\ref{fig:Toffoli-gate}. 
The circuits in this article use mostly standard quantum circuit notation~\cite{nielsen2010quantum}. Additionally, with $b$ as a binary string of size $k-j+1$, we will use the notation $C_{b}^{j:k}(U_m)$ to denote the controlled application of unitary $U$ on qubit number $m$, conditioned on qubits $j$ through $k$ being in the $\ket{b}$ state.
The $\oslash$ symbol is used to indicate a \textit{multiplexer}, which controls on all states in an ancilla register (see Eq.~\eqref{eq:LCU-circuit} for an example).

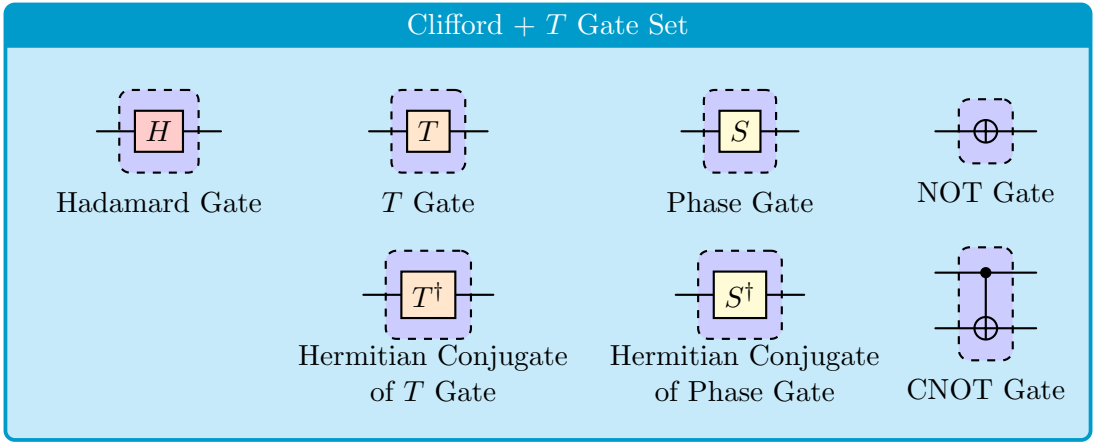
\begin{figure}
	\centering
	\begin{minipage}[t]{0.9\linewidth}
		\begin{tcolorbox}[colback=cyan!20, colframe=cyan!80!black, title=Clifford $+$ $T$ Gate Set, center title]
			\setlength{\tabcolsep}{2pt}
			\renewcommand{\arraystretch}{1.5}
			\begin{tabular}{cccc}
				\begin{tabular}{c}
					\begin{quantikz}        
						& \gate[style={fill=red!20}]{H} 
						\gategroup[1,steps=1,style={dashed,rounded
							corners,fill=blue!20, inner
							xsep=2pt},background]{}
						&
					\end{quantikz}\\ Hadamard Gate
				\end{tabular}
				& 
				\begin{tabular}{c}
					\begin{quantikz}        
						& \gate[style={fill=orange!20}]{T} 
						\gategroup[1,steps=1,style={dashed,rounded
							corners,fill=blue!20, inner
							xsep=2pt},background]{}
						&
					\end{quantikz}\\ $T$ Gate
				\end{tabular}
				& 
				\begin{tabular}{c}
					\begin{quantikz}        
						& \gate[style={fill=yellow!20}]{S} 
						\gategroup[1,steps=1,style={dashed,rounded
							corners,fill=blue!20, inner
							xsep=2pt},background]{}
						&
					\end{quantikz}\\ Phase Gate
				\end{tabular}
				& 
				\begin{tabular}{c}
					\begin{quantikz}        
						& \targ{}
						\gategroup[1,steps=1,style={dashed,rounded
							corners,fill=blue!20, inner
							xsep=2pt},background]{}
						&
					\end{quantikz}\\ NOT Gate
				\end{tabular}
				\\ & 
				\begin{tabular}{c}
					\begin{quantikz}        
						& \gate[style={fill=orange!20}]{T^\dag} 
						\gategroup[1,steps=1,style={dashed,rounded
							corners,fill=blue!20, inner
							xsep=2pt},background]{}
						&
					\end{quantikz}\\
					\renewcommand{\arraystretch}{1}
					\begin{tabular}{c} Hermitian Conjugate \\ of $T$ Gate \end{tabular}
				\end{tabular}
				& 
				\begin{tabular}{c}
					\begin{quantikz}        
						& \gate[style={fill=yellow!20}]{S^\dag} 
						\gategroup[1,steps=1,style={dashed,rounded
							corners,fill=blue!20, inner
							xsep=2pt},background]{}
						&
					\end{quantikz}\\
					\renewcommand{\arraystretch}{1}
					\begin{tabular}{c} Hermitian Conjugate \\ of Phase Gate \end{tabular}
				\end{tabular}
				& 
				\begin{tabular}{c}
					\begin{quantikz}        
						& \ctrl{1}
						\gategroup[2,steps=1,style={dashed,rounded
							corners,fill=blue!20, inner
							xsep=2pt},background]{}
						& \\
						& \targ{} & 
					\end{quantikz}\\ CNOT Gate
				\end{tabular}
			\end{tabular}
		\end{tcolorbox}
	\end{minipage}
	\caption{Circuit symbols for the Clifford $+$ $T$ gate set. This gate set is used as a universal gate set for fault-tolerant quantum computing.}
	\label{fig:Clifford+T-gate-set}
\end{figure}

\subsection{Quantum Primitives: Block-Encoding and Linear Combination of Unitaries}
\label{sec:BE-and-LCU}

In general, the matrices that arise in FEM are not unitary. The block-encoding access model presents a method of accessing a general $n$-qubit matrix $A$ on the quantum computer by embedding it into a larger, unitary matrix $U_A$, and then post-selecting the subspace that the desired matrix lies in.

\begin{defn}[$(\alpha,m,\epsilon)$-Block-Encoding \cite{gilyenQuantumSingularValue2019a,lin2022lecture}]
    For an $n$-qubit matrix $A$, we say that the $(n+m)$-qubit unitary matrix $U_A$ is an $(\alpha,m,\epsilon)$-block-encoding of $A$ if
    \begin{equation}
        \|A - \alpha (\bra{0}^{\otimes m}\otimes I_n) U_A (\ket{0}^{\otimes m}\otimes I_n)\| \leq \epsilon \,,
    \end{equation}
    in which case we write $U_A \in (\alpha,m,\epsilon)\mathrm{-BE}(A)$. 
    The parameter $\alpha > 0$ is called the \textbf{subnormalization factor}, while the integer $m \ge 0$ signifies the number of ancilla needed to implement the block-encoding.
    If we can implement the above unitary $U_A$ exactly, then we call it an $(\alpha,m)$-block-encoding of $A$, and write $U_A \in (\alpha,m)\mathrm{-BE}(A)$.
    \label{def:block-encoding}
\end{defn}
Henceforth in this work, we assume that block-encodings are implemented exactly, and only analyze the gate cost, subnormalization, and number of ancilla qubits needed to implement the finite element method. For a block-encoding of a matrix to exist, we require that its spectral norm satisfy $\|A\| \le 1$. We represent the block-encoding as a $2^m \times 2^m$ block matrix
\begin{equation}
    U_A = \begin{bmatrix}
        A/\alpha & * & \cdots & *\\
        * & * & \cdots & *\\
        \vdots & \vdots & \ddots & \vdots\\
        * & * & \cdots & *
    \end{bmatrix}\,,
\end{equation}
where the $*$ represent $n$-qubit blocks whose entries are inconsequential. The block-encoding is said to succeed when the ancilla are all measured to be in the $\ket{0}$ state after application of $U_A$. As a circuit, a successful application of the block-encoding is denoted by
\tikzexternalenable
\begin{align}
\begin{quantikz}
    \lstick{$\ket{0}^{\otimes m}$} & \gate[2]{U_A} & \meter{} \rstick{$\bra{0}^{\otimes m}$}\\
    \lstick{$\ket{\psi}_n$} & & \rstick{$\frac{1}{\alpha}A\ket{\psi}_n$}
\end{quantikz}
\,.
\end{align}
\tikzexternaldisable
The subnormalization is related to the success probability in that the probability of measuring the ancilla in the all-zero (i.e., $0^m$) state is 
\begin{align}
    p(0^m) &= \frac{1}{\alpha^2} \| A \ket{\psi}_n \|^2 \,.
\end{align}

\definecolor{lavenderpurple}{rgb}{0.59, 0.48, 0.71}
\begin{figure}
    \centering
    \begin{minipage}[t]{0.9\linewidth}
    \begin{tcolorbox}[colback=lavenderpurple!20, colframe=lavenderpurple!80!black, title=Toffoli Gate, center title]
    \setlength{\tabcolsep}{2pt}
    \renewcommand{\arraystretch}{1.5}
    \begin{equation*}
    \begin{tabular}{c}
        \begin{quantikz}        
            & \ctrl{2}
            \gategroup[3,steps=1,style={dashed,rounded
            corners,fill=blue!20, inner
            xsep=2pt},background]{}
            & \\
            & \ctrl{1} & \\
            & \targ{} & 
        \end{quantikz}\\ Toffoli Gate
    \end{tabular}
    = 
    \begin{quantikz}[column sep=6pt]
        & & & & \ctrl{2} & & & & \ctrl{2} & & \ctrl{1} & & \ctrl{1} & \gate{T} &\\
        & & \ctrl{1} & & & & \ctrl{1} & & & \gate{T^\dag} & \targ{} & \gate{T^\dag} & \targ{} & \gate{S} & \\
        & \gate{H} & \targ{} & \gate{T^\dag} & \targ{} & \gate{T} & \targ{} & \gate{T^\dag} & \targ{} & \gate{T} & \gate{H} & & & &
    \end{quantikz}
    \end{equation*}
    \end{tcolorbox}
    \end{minipage}
    \caption{The Toffoli gate expressed in terms of Clifford $+$ $T$ gates, as given by \cite{nielsen2010quantum}.}
    \label{fig:Toffoli-gate}
\end{figure}
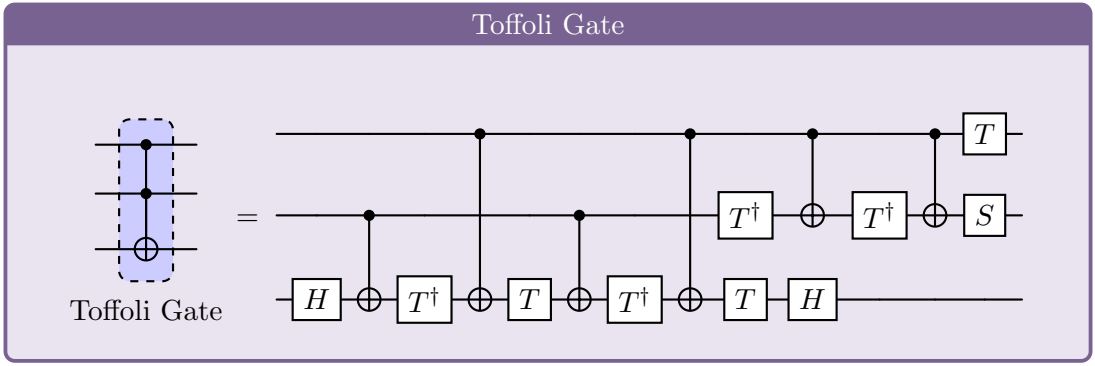

The assembly of global finite element arrays typically involves 
summing elemental contributions that couple the degrees of freedom (or nodes) within each element. The technique of the Linear Combination of Unitaries (LCU)~\cite{childsHamiltonianSimulationUsing2012} is a powerful quantum primitive that will enable us to perform many tasks such as assembly and numerical integration on the quantum computer. In what follows, we summarize LCU briefly, with more details regarding its proof given in Appendix~\ref{sec:block-encoding-framework}.

Let $L := \sum_{j \in [J]} \beta_j U_j$ be a linear combination of unitaries. For simplicity, suppose that $J = 2^p$, and for each $\beta_j \in \bbC$ write $\beta_j = r_j e^{i\theta_j}$ with $r_j > 0$ and $\theta_j \in [0,2\pi)$. We define the ``prepare oracle'' as any unitary \textsc{prep} that satisfies
\begin{align}
    \textsc{prep} \ket{0}^{\otimes p} &= \sum_{j \in [J]} \sqrt{\beta_j} \ket{j} \,,
    \label{eq:prep-R}
\end{align}
where we take the principal square root $\sqrt{\beta_j} = \sqrt{r_j} e^{i\theta_j/2}$. That is, $\textsc{prep}$ satisfies 
\begin{align}
    \textsc{prep} &= \frac{1}{\sqrt{\|\beta\|_1}} \begin{bmatrix}
        \sqrt{\beta_{0}} & * & \cdots & *\\
        \sqrt{\beta_{1}} & * & \cdots & *\\
        \vdots & \vdots & \ddots & \vdots\\
        \sqrt{\beta_{J-1}} & * & \cdots & *
    \end{bmatrix}
    \,,
\end{align}
where $\|\beta\|_1 := \sum_{j \in [J]} |\beta_j|$. 
Similarly, we require another ``prepare'' oracle, $\widetilde{\textsc{prep}}$, that prepares the LCU coefficients as a superposition of row vectors:
\begin{align}
    \bra{0}^{\otimes p} \widetilde{\textsc{prep}} &= \sum_{j \in [J]} \sqrt{\beta_j} \bra{j} \,,
    \label{eq:prep-L-dag}
\end{align}
or in matrix form 
\begin{align}
    \widetilde{\textsc{prep}} &= \frac{1}{\sqrt{\|\beta\|_1}} \begin{bmatrix}
        \sqrt{\beta_{0}} & \sqrt{\beta_{1}} & \cdots & \sqrt{\beta_{J-1}}\\
        * & * & \cdots & *\\
        \vdots & \vdots & \ddots & \vdots\\
        * & * & \cdots & *
    \end{bmatrix}
    \,.
\end{align}
Note that when $\theta_j = 0$ for all $j \in [J]$, then we can take $\widetilde{\textsc{prep}} := \textsc{prep}^\dag$. In addition, we assume access to the ``select'' oracle 
\begin{equation}
    \textsc{sel}:= \sum_{j \in [J]}\ket{j}\bra{j}\otimes U_j
    \,.
    \label{eq:select-oracle-defn}
\end{equation}
The LCU block-encoding $U_L \in (\|\beta\|_1,p)\mathrm{-BE}(L)$ can then be expressed as the $(n+q)$-qubit unitary
\begin{equation}    
    U_L := \left(\widetilde{\textsc{prep}} \otimes I_n\right)\cdot \textsc{sel} \cdot \left(\textsc{prep} \otimes I_n\right).
\end{equation}
In circuit form, we have 
\tikzexternalenable
\begin{align}
\begin{quantikz}
        \lstick{$\ket{0}^{\otimes p}$}& \gate{\textsc{prep}}
        \gategroup[2,steps=3,style={dashed,rounded
        corners,fill=blue!20, inner
        xsep=2pt},background,label style={label
        position=below,anchor=north,yshift=-0.2cm}]{{$(\|\beta\|_1, p)\mathrm{-BE}(L)$}}
        &\qw \mathlarger{\mathlarger{\mathlarger{\mathlarger{\oslash}}}} \vqw{1} &\gate{\widetilde{\textsc{prep}}} &\qw\rstick{$\bra{0}^{\otimes p}$}\\
        \lstick{$\ket{\psi}$}& \qw  &\gate{U_j} &\qw&\qw\rstick{$\frac{1}{\|\beta\|_1}L\ket{\psi}$} 
    \end{quantikz}
    \label{eq:LCU-circuit}
\end{align}
\tikzexternaldisable
where the $\oslash$ notation indicates control on all states in the ancilla register (i.e., this symbol represents the $\ketbra{j}{j}$ factor of the select oracle defined in Eq.~\eqref{eq:select-oracle-defn}).
As a shorthand for the LCU circuit (Eq.~\eqref{eq:LCU-circuit}), we will write
\begin{align}
    U_L &= \LCU\left((U_j)_{j \in [J]}, \beta \right) \,.
\end{align}
If instead we have that $\textsc{prep}$ and $\widetilde{\textsc{prep}}$ are implemented so that $U_L \in (\|\beta\|_1,p,\epsilon)\mathrm{-BE}(L)$, we shall write 
\begin{align}
    U_L &= \LCU_\epsilon\left((U_j)_{j \in [J]}, \beta \right) \,.
\end{align}

\subsection{Position Operators and Cartesian Coordinates}
\label{sec:position-operators-and-Cartesian-coordinates}

Given a bounded $d$-dimensional domain $\Omega \subset \bbR^d$, we want to encode a set of coordinates $x = (x^1,\dots,x^d) \in \bbR^d$ as amplitudes of some quantum state. Assume that $\Omega \subseteq [0,1]^d$. Since we eventually evaluate functions on a discrete set of nodal points and quadrature points of a mesh of $\Omega$ (in Section~\ref{sec:numerical-integration}), it is sufficient to be able to describe the coordinates on some lattice of $[0,1]^d$.

Suppose that we have $N = 2^n$ nodal points in each dimension. We will describe the positions on this lattice using a similar convention to Ref.~\cite{kharazi2024explicitblockencodingsboundary}. 
In one dimension, the domain $[0,1]$ becomes discretized by the lattice $\{i/(N-1) \mid i \in [N]\}$. This corresponds to some set of basis states $\ket{0},\dots,\ket{N-1}$. We want access to the \textit{position operator}\footnote{The position operator $X$ should not be confused with the NOT gate. In circuit diagrams, we will use cross-hairs for the NOT gate to avoid ambiguity (see Figure~\ref{fig:Clifford+T-gate-set}).} $X$, which satisfies 
\begin{align}
    X\ket{i} &= i \ket{i} \,,
    \label{eq:position-operator}
\end{align}
where $i \in [N]$.
We can block-encode the position operator using the following lemma.
\begin{lem}[LCU of diagonal position matrix, adapted from {\cite[Lemma~15]{mukhopadhyay2024PauliFierz}}]
\label{lem:position-operator}
Let the ``position operator" $X = \operatorname{diag}(0, 1, \ldots, N-1)$ be an $N\times N$ matrix with $N = 2^n$. Then the position operator can be written as the summation 
\begin{equation}
    X = \frac{N-1}{2}I - \frac{1}{2}\sum_{i=0}^{n-1}2^i Z^{(i)}
\end{equation}
where 
\begin{equation}
    Z^{(i)} = 
    \underbrace{ I \otimes\cdots\otimes I }_{(n-i-1)\text{-times}}
    \otimes Z \otimes 
    \underbrace{ I \otimes \cdots \otimes I }_{i\text{-times}}
    \,.
\end{equation}
Implementing this sum as an LCU (see \cref{lem:linear-combination-of-block-encodings}) defines a block-encoding of the discrete position operator $X$ with subnormalization $N-1$, and requiring $\log(n)$ ancilla qubits (i.e.,  $(N-1,\log(n))\mathrm{-BE}(X)$).
\end{lem}
For $d$-dimensions, we represent the computational basis on $dn$ qubits using the tensor product 
\begin{align}
    \ket{\jb} &:= \ket{j_{d-1}}\cdots\ket{j_1}\ket{j_{0}} \,,
\end{align}
where $\jb = (j_0,\dots,j_{d-1})$ is a multi-index with $j_i \in [N]$ for all $i \in [d]$. The position operator for the $i$th coordinate is then denoted and defined by 
\begin{align}
    X^{(i)} &:= \underbrace{ I \otimes\cdots\otimes I }_{(d-i-1)\text{-times}}
    \otimes X \otimes 
    \underbrace{ I \otimes \cdots \otimes I }_{i\text{-times}}
    \,.
\end{align}
Notice that all of the position operators $X^{(i)}$ commute with one another. Additionally, $X^{(i)}$ plays the role of picking out the $i$th coordinate in the basis vector $\ket{\jb}$, i.e., 
\begin{align}
    X^{(i)}\ket{\jb} &= j_i \ket{\jb} \,.
\end{align}

Given a function $f \colon [0,1]^d \to \bbR$, we represent the function on the lattice using the position basis as the vector
\begin{align}
    \ket{f} &:= \frac{1}{\Nc_f} \sum_{\jb \in [N]^d} f\left(\jb/(N-1)\right) \ket{\jb} \,,
    \label{eq:f-vector-representation}
\end{align}
where $\Nc_{f} := \sqrt{\sum_{\jb \in [N]^d} |f\left(\jb/(N-1)\right)|^2}$ is a normalizing factor.
Alternatively, we can represent $f$ as an operator $\hat{f}$ that is diagonalized by the position basis as 
\begin{align}
    \hat{f}\ket{\jb} &= \frac{1}{\| f \|_{L^\infty([0,1]^d)}} f\left(\jb/(N-1)\right) \ket{\jb}
    \iff 
    \hat{f} = \frac{1}{\| f \|_{L^\infty([0,1]^d)}} \sum_{\jb \in [N]^d} f\left(\jb/(N-1)\right) \ket{\jb}\bra{\jb} 
    \,.
    \label{eq:f-operator-representation}
\end{align}
Given access to the operator form in Eq.~\eqref{eq:f-operator-representation}, we can always prepare the vector form in Eq.~\eqref{eq:f-vector-representation} by applying the operator to the uniform superposition state on $nd$ qubits:
\begin{align}
    \hat{f} \left( \ket{+}^{\otimes nd} \right) =
    \hat{f} \left( \frac{1}{N^{d/2}} \sum_{\jb \in [N]^d} \ket{\jb} \right) &= \frac{1}{N^{d/2}} \cdot \frac{1}{\| f \|_{L^\infty([0,1]^d)}} \sum_{\jb \in [N]^d} f\left(\jb/(N-1)\right) \ket{\jb} \,,
    \label{eq:f-hat-to-f}
\end{align}
where $\ket{+} := \frac{1}{\sqrt{2}} \left( \ket{0} + \ket{1} \right)$.
The success probability for a block-encoding $U_{\hat{f}}$ that implements $\hat{f}$ is given by the square of the norm of Eq.~\eqref{eq:f-hat-to-f}, which for large $N$ satisfies~\cite{mcardle2022quantum}
\begin{align}
    \| \hat{f} \ket{+}^{\otimes nd} \| = 
    \Fc_{f}^{[N]} &:= \frac{1}{\| f \|_{L^\infty([0,1]^d)}} \sqrt{ \frac{1}{N^d} \sum_{\jb \in [N]^d} |f\left(\jb/(N-1)\right)|^2 }\\
    &\approx \frac{\| f \|_{L^2([0,1]^d)}}{\| f \|_{L^\infty([0,1]^d)}} =: \Fc_{f}^{[\infty]}
    \,,
\end{align}
where $\Fc_{f}^{[N]}$ is the filling fraction of the function $f$. One may then do $\Oc\left(1/\Fc_{f}^{[N]}\right)$ rounds of exact amplitude amplification (see~\cite[Appendix A]{mcardle2022quantum}) to boost the success probability of the block-encoding $U_{\hat{f}}$ to $1$. 

The filling fraction represents the (square root of the) volume that $|f|^2/\| f \|_{L^\infty([0,1]^d)}^2$ fills in the hypercube $[0,1]^d$. Since the function $f$ is known, the filling fraction can be estimated (or at least lower-bounded) classically. Whenever we use the operator $\hat{f}$ to prepare the state $\ket{f}$, we will require that the filling fraction does not scale exponentially with the problem parameters (in particular, $N$ and $d$). 
Given a polynomial approximation for the function $f$, we can block-encode the operator $\hat{f}$ in one dimension by applying Quantum Signal Processing (reviewed in Section~\ref{sec:QSP}) on the position operator $X$, and more generally in $d$-dimensions using the Multivariate Quantum Eigenvalue Transformation (reviewed in Section~\ref{sec:MQET}) on the position operators $\Xb := \{X^{(i)}\}_{i \in [d]}$.

\subsection{Weak Formulation of the Modified Poisson's Equation as a Quantum Linear Systems Problem}
\label{subsec:weak-formulation-modified-Poisson}

Consider the modified Poisson equation
\begin{subequations}
\label{eq:modified-Poisson-system}
\begin{align}
    -\nabla \cdot (D \nabla u) + k u &= f 
    \quad \text{on } \Omega
    \,, \label{eq:modified-Poisson-eqn} \\
    u &= g 
    \quad \text{on } \Gamma_D \,, \label{eq:modified-Poisson-eqn-Dirichlet-bc}\\
    D\nabla u \cdot \nb  &= h
    \quad \text{on } \Gamma_N \,, \label{eq:modified-Poisson-eqn-Neumann-bc}
\end{align}
\end{subequations}
where $\Gamma_D$ and $\Gamma_N$ satisfy $\partial\Omega = \overline{\Gamma}_D \cup \overline{\Gamma}_N$ and are, respectively, the portions of the boundary where Dirichlet and Neumann boundary conditions are prescribed. Additionally, $D$, $\rho$, and $f$ are scalar fields over~$\Omega$, and $\nb$ is the outward unit normal to $\Omega$. If $D$ is constant and $f \equiv 0$, this reduces to the Helmholtz equation, so this equation may also be labelled as an inhomogeneous Helmholtz equation. Physically, we can also view this as a steady-state reaction-diffusion equation with a diffusive flux $\Jb := -D \nabla u$, a source $f$, and a first-order (consumption) reaction with the rate law $ku$. Other examples of physical phenomena modelled by this equation include steady-state heat diffusion, elastic deformation of beams and plates, electrostatic equilibrium, and the propagation of acoustic waves~\cite{ciarlet2023finite}.

The weak formulation of the modified Poisson equation can be obtained by multiplying Eq.~\eqref{eq:modified-Poisson-eqn} by a test function $\delta u$ and integrating by parts, whereupon we arrive at
\begin{align}
    \int_\Omega D\nabla u \cdot \nabla \delta u \,d\Omega + \int_{\Omega} k u \delta u \,d\Omega &= \int_\Omega \delta u f \,d\Omega + \int_{\Gamma_N} \delta u h \,d\Gamma
    \,.
    \label{eq:modified-Poisson-weak-form}
\end{align}
Note that the test function is assumed to satisfy $\delta u = g$ on $\Gamma_D$, which corresponds to the Dirichlet boundary conditions being enforced strongly in the system. This means that instead of lying in the function space $H^1(\Omega)$, the solution $u$ is in the space of \textit{admissible solutions} $\Uc := \left\{ u \in H^1(\Omega) \mid u = g \text{ on } \Gamma_D \right\}$, and the test function $\delta u$ belongs to the space of \textit{admissible variations} $\Uc_0 := \left\{ \delta u \in H^1(\Omega) \mid \delta u = 0 \text{ on } \Gamma_D \right\}$ (see Appendix~\ref{sec:analysis-of-FEM} and Ref.~\cite{papadopoulos2015280a} for more information). 
We will modify the system of equations in Section~\ref{sec:constraints} to ensure that these constraints on the function space are satisfied.

We can obtain a finite-dimensional Bubnov-Galerkin approximation of \cref{eq:modified-Poisson-weak-form} if we approximate the solution $u$ and its variation $\delta u$ as linear combinations
\begin{align}
    u &= \sum_{j \in [\numnp{}]} u_j N_j 
    \quad \text{and} \quad
    \delta u = \sum_{j \in [\numnp{}]} \delta u_j N_j 
    \,,
    \label{eq:u-PG-discretization}
\end{align}
where $u_j, \delta u_j \in \bbC$ are the values of the degrees of freedom, $\numnp{}$ is the total number of degrees of freedom, and $\{N_0, \dots, N_{\numnp{}-1}\}$ are basis (or interpolation) functions. Equations~\eqref{eq:modified-Poisson-weak-form} and~\eqref{eq:u-PG-discretization} lead to the linear system 
\begin{align}
    \Lc \ub &= \fb \,,
    \label{eq:classical-FE-linear-systems-problem}
\end{align}
where $\Lc := K +M$, and
\begin{align}
    K &:= \sum_{j,k \in [\numnp{}]}\left( \int_\Omega D \nabla N_i \cdot \nabla N_j \,d\Omega \right) \ketbra{i}{j} \,, \label{eq:K-global-integral-defn} \\
    M &:= \sum_{j,k \in [\numnp{}]}\left( \int_\Omega k N_i N_j \,d\Omega \right) \ketbra{i}{j} \,, \label{eq:M-global-integral-defn} \\
    \fb &:= \sum_{j \in [\numnp{}]} \left( \int_\Omega N_j f \,d\Omega + \int_{\Gamma_N} N_j h \,d\Gamma \right) \ket{j} \,, \label{eq:f-integral-defn}\\
    \ub &= \sum_{j \in [\numnp{}]} u_j \ket{j} \,.
\end{align}

Presently, $\ub$ and $\fb$ are unnormalized vectors, and thus cannot be prepared as quantum states. 
To turn this into a Quantum Linear Systems Problem (QLSP), we need to normalize the right-hand-side of \cref{eq:classical-FE-linear-systems-problem} and instead solve for 
\begin{equation}
    \left(\frac{\|\ub\|}{\|\fb\|}\right) \Lc \ket{u} = \ket{f} \,,
    \label{eq:FE-LSP-normalized}
\end{equation}
where without loss of generality we have assumed that $\fb \ne \zerob$, and
\begin{align}
    \ket{u} := \ub / \|\ub\|
    \quad \text{and} \quad 
    \ket{f} := \fb / \|\fb\| \,.
\end{align}

\begin{defn}[Quantum Linear Systems Problem (QLSP)]
\label{def:QLSP}
    Let $\Lc \in \bbC^{2^n \times 2^n}$ be an $n$-qubit matrix, and $\ket{f} \in \bbC^{2^n}$ an $n$-qubit quantum state. The QLSP is to prepare an $n$-qubit quantum state $\ket{u}$ such that 
    \begin{align*}
        \Lc \ket{u} \text{ is equal to } \ket{f} \text{ up to normalization.}
    \end{align*}
    We will compactly represent the QLSP as 
    \begin{align}
        \Lc \ket{u} = \ket{f}
        \,,
        \label{eq:QLSP-general-system}
    \end{align}
    with the understanding that since quantum states appear on both sides of the equation, this equality is up to normalization.
\end{defn}
\noindent
Comparing Eqs.~\eqref{eq:FE-LSP-normalized} and~\eqref{eq:QLSP-general-system}, we see that Eq.~\eqref{eq:FE-LSP-normalized} defines a QLSP for the FEM formulation of the Modified Poisson's equation. Upon preparing the ``solution'' $\ket{u}$, however, how does one recover the unnormalized vector $\ub$ (i.e., the \textit{solution})?

In this article, we will access $\Lc$ using a block-encoding $U_\Lc \in (\alpha,m)\mathrm{-BE}(\Lc)$, where $\alpha \ge \| \Lc \|$ (see Appendix~\ref{sec:block-encoding-framework} for a review of the block-encoding framework). Choose a lower bound for the singular values $\beta \le \sigma_\text{min}(\Lc)$. Using the QLSP solver in Ref.~\cite{An_2022}, we can transform $U_\Lc$ to a block-encoding of $U_{\Lc^{-1}} \in (\frac{\alpha}{\beta}, m+1)\mathrm{-BE}(\Lc)$ with time complexity $\Oc(\kappa n / \epsilon)$. Here, $\kappa \le \frac{\alpha}{\beta}$ is the condition number of $\Lc$, and $\epsilon$ is the desired precision. The success probability of applying this block-encoding to the state $\ket{f}$ is precisely 
\begin{align}
    p_\text{QLSP} &:= \left(\frac{\beta}{\alpha}\right)^2 \| \Lc^{-1} \ket{f} \|^2
    \,.
\end{align}
Recognizing that $\Lc^{-1} \ket{f} = \frac{\ub}{\| \fb \|}$, we can rearrange this for 
\begin{align}
    \| u \| &= \left(\frac{\alpha}{\beta}\right) \|\fb\| \sqrt{p_\text{QLSP}}
    \,.
    \label{eq:u-solution-norm-recovery}
\end{align}
As the quantum algorithm for the QLSP is run, the success probability $p_\text{QLSP}$ can be simultaneously estimated to precision $\epsilon$ using $\Oc(1/\epsilon^2)$ iterations of the quantum circuit.
The remaining quantities on the right-hand side of Eq.~\eqref{eq:u-solution-norm-recovery} are known, so that the norm of the solution may be recovered\footnote{The assumption that $\|\fb\|$ is known or that it can be computed efficiently will be necessary several times throughout this work. Here, since this vector corresponds to the force applied on the system and the boundary conditions prescribed on the problem, this is known data, and so this assumption is valid in this case.}.

\subsection{Tensor Products of Lagrange Elements}
\label{subsec:tensor-products-of-Lagrange-elements}

In this section, we outline the construction of $d$-rectangular elements from $1$-rectangles (i.e., line elements). The basis functions on these elements will be Lagrange interpolation polynomials. Consider a point $x \in \bbR^d$. We will write $x = (x^0, \dots, x^{d-1})$ for the Cartesian coordinates of this point. We first construct the one-dimensional element, then build higher-dimensional elements in the next section by making use of the tensor product structure of rectangular Lagrange elements.

\begin{defn}[Standard One-Dimensional Lagrange Finite Element of Order $p$]
    Consider the one-dimensional interval $\Omega^e := [0,1]$. Let $p \ge 1$ be an integer. Then the standard Lagrange basis functions of order $p$ are given by 
    \begin{align}
        N_j^e(x) &:= \frac{(x - x_0)}{(x_j - x_0)} \dots \frac{(x - x_{j-1})}{(x_j - x_{j-1})} \frac{(x - x_{j+1})}{(x_j - x_{j+1})} \dots \frac{(x - x_p)}{(x_j - x_p)} = \prod_{\substack{m \in [p+1] \\ m \ne j}} \frac{(x - x_m)}{(x_j - x_m)} \,,
        \label{eq:Lagrange-basis-function-defn}
    \end{align}
    where $x_m := \frac{m}{p}$ for all $m \in [p + 1]$. The set $\{x_m\}_{m \in [p+1]}$ are referred to as the local nodal points of the element. For convenience, we denote the (unnormalized) vector of these basis functions as 
    \begin{align}
        \ket{N^e(x)} &:= \begin{bmatrix}
            N_1^e(x) & \hdots & N_p^e (x)
        \end{bmatrix}^\dag \,.
    \end{align}
    We denote the function space spanned by these basis functions by $\Qc_p(\Omega^e) := \operatorname{span}_\bbR\{N_j^e(x)\}_{j \in [p+1]}$. The triple $(\Omega^e, \{x_m\}_{m \in [p+1]}, \ket{N^e(x)})$ consisting of the domain $\Omega^e$, the nodal points $\{x_m\}_{m \in [p+1]}$ (also called the degrees of freedom), and the basis functions $\ket{N^e(x)}$ is referred to as a Lagrange finite element of order $p$.
\end{defn}

For convenience, we will identify a finite element with the symbol for its function space (e.g., we will use $\Qc_p(\Omega^e)$ in lieu of $(\Omega^e, \{x_m\}_{m \in [p+1]}, \ket{N^e(x)})$). The Lagrange basis functions satisfy the fundamental property 
\begin{align}
    N_j^e(x_k) = \delta_{jk} \,,
    \label{eq:fundamental-property}
\end{align}
which renders the basis functions interpolatory~\cite{papadopoulos2015280a}. Interpolatory basis functions are desirable because the value of their corresponding degree of freedom is exactly equal to the value of the solution at that node. This simplifies the extraction of observables of interest (see \cref{sec:demo} for an example).

\begin{lem}[Partition of unity]
\label{lem:partition-of-unity}
    Consider a Lagrange element $\Qc_p(\Omega^e)$ with basis functions $\ket{N^e(x)}$. Then the basis functions form a partition of unity at each point of the element, i.e., 
    \begin{align}
        \sum_{j \in [p+1]} N_j^e(x) &= 1 \ \forall x \in \Omega^e \,.
        \label{eq:partition-of-unity}
    \end{align}
\end{lem}

\begin{proof}
    Let $f(x) := \sum_{j \in [p+1]} N_j^e(x) - 1 $. Since each $N_j^e$ is a polynomial of degree $p$, $f$ must be a polynomial with $\deg f \le p$. By the fundamental property (Eq.~\eqref{eq:fundamental-property}), $f$ has $p+1$ distinct zeros at the nodes $\{x_m\}_{m \in [p+1]}$ of the Lagrange element. It follows that $f \equiv 0$, which yields Eq.~\eqref{eq:partition-of-unity}.
\end{proof}

We can form higher-dimensional ($d$-rectangular) elements through the Cartesian product
\begin{align}
    \Omega^{e,d} := \underbrace{\Omega^e \times \dots \times \Omega^e}_{d\text{-times}} \,.
\end{align}
The corresponding higher-dimensional Lagrange basis functions can be obtained by taking tensor products of $\ket{N^e(x)}$.

\begin{defn}[Tensor Product Element]
    Let $(\Omega^e, \{x_m\}_{m \in [p+1]})$ be an element with basis functions $\ket{N^e(x)}$ satisfying the fundamental property $N_j^e(x_k) = \delta_{jk}$. The tensor product element is the domain $\Omega^{e,d}$ together with the nodal points 
    \begin{align}
        \{x_{m}\}_{m \in [p+1]^d} := 
        \{x_{m_0}\}_{m_0 \in [p+1]} \times \dots \times \{x_{m_{d-1}}\}_{m_{d-1} \in [p+1]} \,,
        \label{eq:tensor-product-nodes}
    \end{align}
    where $m := (m_0, \dots, m_{d-1})$ is a $d$-dimensional multi-index, and the basis functions
    \begin{align}
        \ket{N^{e,d}(x^0,\dots,x^{d-1})} &:= \ket{N^e(x^{d-1})} \otimes \dots \otimes \ket{N^e(x^{0})} \,.
    \end{align}
\end{defn}

More general element types may consist of arbitrary convex polygons in $d$-dimensions; however, this work focuses specifically on quadrilateral elements. Tensor product elements naturally inherit their node numbering convention from the Kronecker product structure, as described above and illustrated in Example~\ref{ex:2d-basis}.

\begin{examplebox}{Two-dimensional Lagrange tensor product element $Q_1(\Omega^{e,2})$}\label{ex:2d-basis}
	\begin{minipage}[t]{1.\linewidth}
		\vspace*{0pt}
		\begin{center}
			\includegraphics[width=0.4\linewidth]{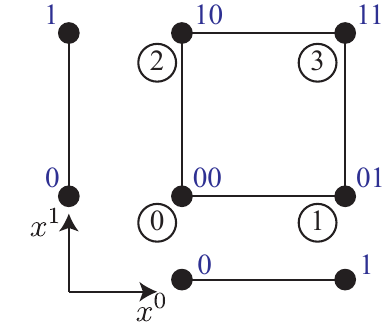}
			\captionof{figure}{Two-dimensional tensor product element.}\label{fig:2d-element}
		\end{center}
	\end{minipage}
\begin{minipage}[t]{1.\linewidth}
	First, observe that for one-dimensional elements $Q_1(\Omega^e)$, the basis functions that linearly interpolate between the two nodes at $x=0$ and $x=1$ are (see the definition in \cref{eq:Lagrange-basis-function-defn} and the fundamental property \cref{eq:fundamental-property})
    \begin{align}
        \ket{N^e(x)} &= \begin{bmatrix}
            1 - x \\ x
        \end{bmatrix} \,.
        \label{eq:Ne-1d}
    \end{align}
    Then the two-dimensional tensor product element has the basis functions 
    \begin{align}
        \ket{N^{e,2}(x^0,x^1)} &:= \ket{N^e(x^1)} \otimes \ket{N^e(x^0)}\\
        &= \begin{bmatrix}
            (1 - x^0)(1 - x^1)\\
            x^0 (1 - x^1)\\
            (1 - x^0) x^1\\
            x^0 x^1
        \end{bmatrix} \,.
        \label{eq:2d-Lagrange-basis-functions}
    \end{align}
    For the nodal points, taking the Cartesian product (\cref{eq:tensor-product-nodes}) of the set of one-dimensional nodal points $\{x_i\}_{i \in [2]} = \{0,1\}$ yields 
    \begin{align}
        \{x_{i,j} \}_{i,j \in [2]} 
        &= \{x_i\}_{i \in [2]} \times \{x_j\}_{j \in [2]}\\
        &= \{0,1\} \times \{0,1\}\\
        &= \left((0,0), (0,1), (1,0), (1,1)\right) \,.
        \label{eq:2d-Lagrange-nodes}
    \end{align}
    See Figure~\ref{fig:2d-element} for a visual representation of this Cartesian product. By evaluating the basis functions (\cref{eq:2d-Lagrange-basis-functions}) at these nodal points (\cref{eq:2d-Lagrange-nodes}), we can verify that the basis functions automatically satisfy the fundamental property,  
    \begin{align}
    	N_{i,j}^{e,2}(x_{k,l}) 
    	&= N_i^e(x_k) N_j^e(x_l)
    	= \delta_{ik}\delta_{jk}
    	\,.
    \end{align}
    
    \end{minipage}\hfill%

\end{examplebox}
Recalling Eqs.~\eqref{eq:K-global-integral-defn} and~\eqref{eq:M-global-integral-defn}, the elemental stiffness matrix for a linear PDE for a one-dimensional element is given by 
\begin{align}
    K_{ij}^e &= \int_{\Omega^e} dx\, \nabla N_i^{e} \cdot \nabla N_j^{e} \,,
    \label{eq:Kij-e}
\end{align}
while the elemental mass matrix for a linear PDE for a one-dimensional element is given by 
\begin{align}
    M_{ij}^e &= \int_{\Omega^e} dx\,  N_i^{e} N_j^{e} \,.
    \label{eq:Mij-e}
\end{align}
For higher-dimensional elements built out of a tensor product structure, the integrals that define the stiffness and mass matrices can be rewritten as iterated integrals. We illustrate this in the following example.

\begin{examplebox}{Two-dimensional Stiffness and Mass Matrices}\label{ex:2d-Ke}
    Recall the example of the tensor product element $Q_1(\Omega^{e,2})$ in Example~\ref{ex:2d-basis}. In two-dimensions, the elemental stiffness matrix is 
    \begin{align}
        K_{ij}^{e,2} &= \int_{\Omega^{e,2}} d^2 x\, \nabla N_i^{e,2} \cdot \nabla N_j^{e,2}\\
        &= \int_{\Omega^{e,2}} dx^0 dx^1\, 
        \frac{\partial}{\partial x^k} \left( N_{i_1}^e(x^1) N_{i_0}^e(x^0)  \right) \cdot \frac{\partial}{\partial x^k} \left( N_{j_1}^e(x^1) N_{j_0}^e(x^0)  \right) \,,
    \end{align}
    where $i = (i_1,i_0)$ and $j = (j_1,j_0)$ have been written as multi-indices. Rewriting the double integral as an iterated integral, 
    \begin{align}
    \begin{split}
        K_{ij}^{e,2} &= \int_{\Omega^{e}} \int_{\Omega^{e}} dx^0 dx^1\, 
        \big( {N_{i_1}^e}'(x^1) N_{i_0}^e(x^0) {N_{j_1}^e}'(x^1) N_{j_0}^e(x^0)\\ 
        &\hspace{1.5in} + N_{i_1}^e(x^1) {N_{i_0}^e}'(x^0) N_{j_1}^e(x^1) {N_{j_0}^e}'(x^0) \big)
    \end{split}\\
    \begin{split}
        &=
        \left(\int_{\Omega^{e}} dx^1\, {N_{i_1}^e}'(x^1) {N_{j_1}^e}'(x^1) \right)
        \left(\int_{\Omega^{e}} dx^0\, N_{i_0}^e(x^0) N_{j_0}^e(x^0) \right)\\
        &\hspace{0.5in} + 
        \left(\int_{\Omega^{e}} dx^1\, {N_{i_1}^e}(x^1) {N_{j_1}^e}(x^1) \right)
        \left(\int_{\Omega^{e}} dx^0\, 
        {N_{i_0}^e}'(x^0) {N_{j_0}^e}'(x^0) \right) \,,
        \label{eq:Kij-e-2d-last-line}
    \end{split}\\
    &= K_{i_1 j_1} M_{i_0 j_0} + M_{i_1 j_1} K_{i_0 j_0} 
    \,,
    \end{align}
    where we have used \cref{eq:Kij-e,eq:Mij-e} in the last line. \Cref{eq:Kij-e-2d-last-line} is exactly the $(i,j)$ component of the tensor 
    \begin{align}
        K^{e,2} &= K^e \otimes M^e + M^e \otimes K^e \,.
    \end{align}
    Similarly, the two-dimensional mass matrix is given by 
    \begin{align}
        M^{e,2} &= M^e \otimes M^e \,.
    \end{align}
\end{examplebox}

The procedure in Example~\ref{ex:2d-Ke} generalizes to $d$-dimensions as:
\begin{gather}
    K^{e,d} = K^e \otimes \underbrace{M^e \otimes \dots \otimes M^e}_{(d-1)\text{-times}}
    + M^e \otimes K^e \otimes  \underbrace{M^e \otimes \dots \otimes M^e}_{(d-2)\text{-times}} + \dots + M^e \otimes \dots \otimes M^e \otimes K^e
    \,,
    \label{eq:Ke-tensor-product-structure}
    \\
    M^{e,d} = \underbrace{M^e \otimes \dots \otimes M^e}_{d\text{-times}}
    \,.
    \label{eq:Me-tensor-product-structure}
\end{gather}
In general, the elemental stiffness matrix $K^{e,d}$ and mass matrix $M^{e,d}$ are dense $(p+1)^d \times (p+1)^d$ matrices for $\Qc_p(\Omega^{e,d})$ elements. Using the constant-size $(p+1) \times (p+1)$ $K^e$ and $M^e$ matrices, however, in the next section we will efficiently construct the higher-dimensional global arrays using the same tensor product structure as above. That is, by exploiting the structure of Cartesian mesh types, we will show that \cref{eq:Ke-tensor-product-structure,eq:Me-tensor-product-structure} hold not just for the local arrays, but for the global arrays as well.

In cases where the elemental arrays are constant throughout the domain, we can compute the entries just once using a classical computer, and then construct oracles to access their entries on the quantum computer.
Henceforth, we assume we have access to an oracle that prepares the entries of both $K^e$ and $M^e$. That is, we assume that we have access to amplitude oracles (see \cref{defn:amplitude-oracle}) $O_{K^e}$ and $O_{M^e}$ such that 
\begin{align}
    O_{K^e}\ket{0}\ket{i}\ket{j} &= \left(\sqrt{K_{ij}^e}\ket{0} + \sqrt{1-K_{ij}^e}\ket{1}\right) \ket{i}\ket{j}
    \,,
    \label{eq:Ke-oracle}\\
    O_{M^e}\ket{i}\ket{j} &= \left(\sqrt{M_{ij}^e}\ket{0} + \sqrt{1-M_{ij}^e}\ket{1}\right) \ket{i}\ket{j}
    \,.
    \label{eq:Me-oracle}
\end{align}
Note that for the standard $\Qc_p(\Omega^e)$ element, we will always have $\norm{K^e}_\mathrm{max} < 1$ and $\norm{M^e}_\mathrm{max} < 1$, so that the above oracles make sense as quantum operators.
Alternatively, we assume access to prepare oracles satisfying 
\begin{align}
    \textsc{prep}_{K^e} \ket{0}_{\log\nen{}}\ket{0}_{\log\nen{}} &= \frac{1}{\sqrt{\|K^e\|_1}} \sum_{i,j \in [\nen{}]} \sqrt{K_{ij}^e} \ket{i}\ket{j}
    \,,
    \label{eq:prep-Ke}\\
    \textsc{prep}_{M^e} \ket{0}_{\log\nen{}}\ket{0}_{\log\nen{}} &= \frac{1}{\sqrt{\|M^e\|_1}} \sum_{i,j \in [\nen{}]} \sqrt{M_{ij}^e} \ket{i}\ket{j}
    \,,
    \label{eq:prep-Me}
\end{align}
with analogous definitions for $\widetilde{\textsc{prep}}_{K^e}$ and $\widetilde{\textsc{prep}}_{M^e}$ (see Section~\ref{sec:BE-and-LCU} for definitions and the convention for calculating $\sqrt{K_{ij}^e}$ when $K_{ij}^e < 0$). Since the entries of the mass matrix $M^e$ are all non-negative, we can take $\widetilde{\textsc{prep}}_{M^e} = \textsc{prep}_{M^e}^\dag$. Additionally, the mass matrix will have unit subnormalization, as we see below.

\begin{prop}
\label{prop:mass-matrix-subnormalization}
    The subnormalization of the mass matrix is identically equal to one, i.e., 
    \begin{align}
        \alpha_M := \|M^{e,d}\|_1 &= 1
        \,.
    \end{align}
\end{prop}

\begin{proof}
    Noting that $N_i^{e,d} \ge 0$ and $\operatorname{vol}\left(\Omega^{e,d}\right) = \operatorname{vol}\left( [0,1]^d \right) = 1$, we have from Lemma~\ref{lem:partition-of-unity} that 
    \begin{align}
        \alpha_M &= \sum_{j,k \in [p+1]^d} |M_{jk}|
        = \sum_{j,k \in [p+1]^d} \int_{\Omega^{e,d}} d\Omega^{e,d} \,  N_j^{e}(x) N_k^{e}(x)\\
        &= \int_{\Omega^{e,d}} d\Omega^{e,d} \left( \sum_{j \in [p+1]^d} N_j^e(x) \right) \left( \sum_{k \in [p+1]^d} N_k^e(x) \right)\\
        &= \int_{\Omega^{e,d}} d\Omega^{e,d} = 1 \,,
    \end{align}
    as desired.
\end{proof}

In the next section, we will see that this tensor product structure carries over to the global stiffness and mass arrays. Given access to oracles that prepare the entries of $K^e$ and $M^e$, and utilizing a block-encoding of what we call the `unit of interaction', we will demonstrate an efficient block-encoding of the global mass and stiffness matrices defined in Eqs.~\eqref{eq:K-global-integral-defn} and~\eqref{eq:M-global-integral-defn}.

\section{Quantum Assembly of Global Element Arrays}
\label{sec:assembly-of-global-arrays-linear}

In this section, we introduce our framework for assembling global finite element arrays on a quantum computer. 
The assembly of global arrays in the classical finite element method (FEM) requires knowledge of the mesh connectivity ``$\IX{}$'', which maps the local node numbering convention of each element to the global node numbering convention for the degrees of freedom in the mesh. Here, we show how a quantum oracle $\Oc_{\IX{}}$ that encodes the mesh connectivity can be used to perform the same assembly procedure on a quantum computer. 
The key step in this procedure is to reverse the order of summation during assembly so that the summation over the elements is the ``inner'' summation, while the summation over the local element nodes is the ``outer'' summation. This renders the number of terms in the LCU that carries out this summation constant with respect to the system size, leading to an efficient assembly procedure on the quantum computer.
To more effectively highlight this step, we consider only finite element arrays with constant elemental contributions in this section, and address the case of non-constant elemental contributions in Section~\ref{sec:numerical-integration}, as the latter case requires additional tools from numerical integration.
We conclude this section by giving an explicit implementation for the assembly of first-order quadrilateral Lagrange elements in multiple dimensions, and generalize this implementation to higher-order elements. We also provide a complexity analysis for the number of Toffoli or simpler gates required to block-encode the assembled finite element arrays.

\begin{figure}
    \centering
    \includegraphics[width=0.7\linewidth]{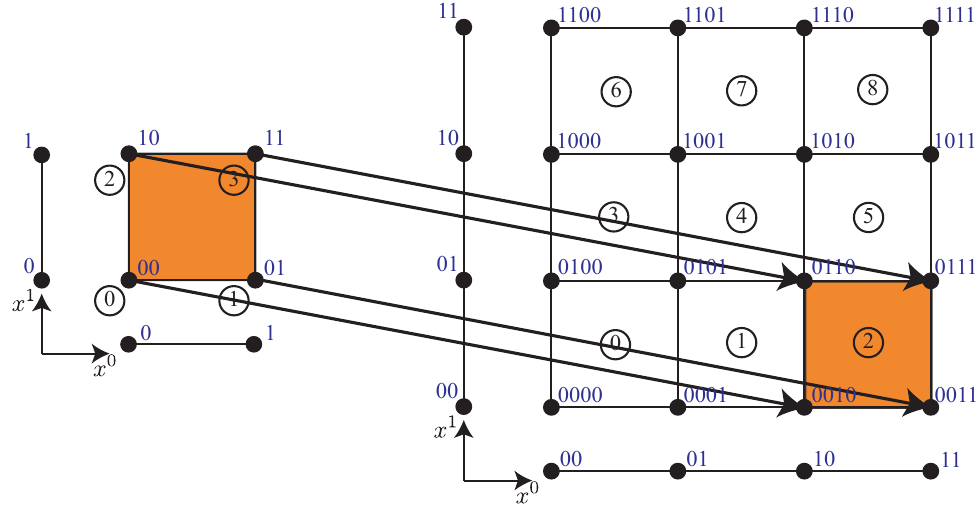}
    \caption{A two-dimensional mesh of first-order quadrilateral Lagrange elements with $\numel{} = 9$ elements, $\nen{} = 4$ elemental nodes, and $\numnp{} = 16$ total nodal points. A representative element is highlighted in orange. The local and global node numbers are written in binary, and the numbering convention is induced from the tensor product of one-dimensional elements.}
    \label{fig:element-insertion}
\end{figure}

\subsection{Connectivity Arrays}

To describe a mesh in FEM, we require at least two data structures: one that encodes the spatial positions of each node in the mesh, and one that stores the local and global node numbering conventions of the nodes.
In this work, we will follow the notation of Ref.~\cite{papadopoulos2015280a} for the computational parameters and arrays that are used in FEM assembly.
For parameters that describe the mesh size (see Figure~\ref{fig:element-insertion} for an example), we denote: 
\begin{itemize}
    \item \numel{} as the total number of elements in the mesh,
    \item \numnp{} as the total number of nodal points in the mesh, and
    \item \nen{} as the number of nodal points per element.
\end{itemize}
Note that once $\numel{}$ and $\nen{}$ are specified, the total number of nodal points $\numnp{}$ is also specified (i.e., only two out of three quantities in the set $\{\numel{},\nen{},\numnp{}\}$ are independent).

Classical access to a mesh can be quantumly abstracted as access to an oracle $\mesh{}$, which given an element index $\el{} \in [\numel{}]$ and a local node index $\nn{} \in [\nen{}]$, returns the coordinates of the corresponding node. For example, rescaling a $d$-dimensional domain $\Omega \subset \bbR^d$ so that it lies within the unit box $[0,1]^d$ (i.e., $\Omega \subseteq [0,1]^d$), the mesh oracle satisfies 
\begin{equation}
    \mesh{}_h \ket{\el{}} \ket{\nn{}} \underbrace{\ket{0} \cdots \ket{0}}_{d\mathrm{-times}} = \ket{\el} \ket{\nn{}} \ket{\bar{x}^0} \cdots \ket{\bar{x}^{d-1}}
    \,,
    \label{eq:mesh-oracle-defn}
\end{equation}
where $\ket{\bar{x}^j}$ is a basis state with $j \in [d]$ and $\bar{x}^j = 2^r [0.b_1 \dots b_r]$ an $r$-bit binary string encoding the $j$-th coordinate of the nodal point. 
In Eq.~\eqref{eq:mesh-oracle-defn}, we use the subscript ``$h$'' to denote dependence on a mesh-refinement parameter, which is typically the side-length or diameter of an element.
This oracle enables the evaluation of several important mesh-related functions, such as isoparametric mappings and Jacobian evaluations that arise from pull-backs to a reference element.
For low-dimensional problems in complex geometries, mesh generation can be performed classically, and subsequently, the $\mesh{}$ oracle can be implemented as a quantum read-only memory (QROM)~\cite{phalak2022optimization} circuit. Achieving a quantum advantage for solving PDEs in complex geometries with high-precision computation will likely require a combination of QROM access to a (classically generated) coarse version of the mesh, followed by a quantum implementation of a mesh refinement scheme to transform the oracle into a mesh fine enough to simulate at the desired precision. In this work, however, we consider only meshes whose $\mesh{}$ oracles we can give explicit quantum circuits for.
Additionally, we will restrict our attention to structured meshes of the same element type that carry one degree of freedom for each node.

We now turn our attention to the mapping between the local and global degree of freedom numbers.
The connectivity array $\mathrm{IX}$ is an $\nen{} \times \numel{}$ matrix of integers whose columns correspond to element numbers, and rows correspond to local node numbers. The entry $J := \IX{}(j,\el{})$ of the connectivity matrix is the global node number $J \in [\numnp{}]$ of local degree of freedom $j \in [\nen{}]$ of element $\el{} \in [\numel{}]$ (see Figure~\ref{fig:IX-injective} for an example). Quantumly, we implement the connectivity matrix through an oracle $\Oc_\IX{}$ that satisfies
\begin{equation}
    \Oc_\mathrm{IX} \ket{j} \ket{\el{}} \ket{0} = \ket{j} \ket{\el{}} \ket{J} 
    \,.
    \label{eq:IX-quantum-oracle}
\end{equation}
This oracle encodes the map from the local node numbers to the global node numbers, and is completely reversible. In the remaining subsections, however, we work to construct a version of this map that maps the element number $\ket{e}$ directly to the global node number $\ket{J}$.

\subsection{Tools for Parallel Assembly: The Unit of Interaction and the Local-to-Global Node Number Indicator Matrix}
\label{subsec:unit-of-interaction-and-indicator-matrix}

To outline our strategy for assembling global finite element arrays on a quantum computer, consider the stiffness matrix $K$ (see Eq.~\eqref{eq:K-global-integral-defn}) assembled from elemental contributions $K^e$ (given by Eq.~\eqref{eq:Kij-e}). 
One strategy for assembly is to embed the dense $\nen{} \times \nen{}$ array $K^e$ into an $\numnp{} \times \numnp{}$ array, shift the contributions in the correct locations using the connectivity oracle $\Oc_\mathrm{IX}$, and then LCU over all elements $\el{} \in [\numel{}]$. This follows the classical assembly procedure, which can be summarized as:
\begin{align}
    \xrightarrow{\text{Compute elemental contribution}}&\, 
    K^{\el{}} = \sum_{j,k \in [\nen{}]} K_{jk}^\el{} \ketbra{j}{k} \,,\\
    \xrightarrow{\text{Map contribution to its position in the global array}}&\,
    \sum_{j,k \in [\nen{}]} K_{jk}^\el{} \ketbra{\mathrm{IX}(j,\el{})}{\mathrm{IX}(k,\el{})} \,,\\
    \xrightarrow{\text{Sum over all elements}}&\,
    K = \sum_{\el{} \in [\numel{}]} \sum_{j,k \in [\nen{}]} K_{jk}^\el{} \ketbra{\mathrm{IX}(j,\el{})}{\mathrm{IX}(k,\el{})} \,.
    \label{eq:K-classical-summation-order}
\end{align}
On a quantum computer, however, the outer summation in \cref{eq:K-classical-summation-order} is an LCU of size $\numel{} = \Oc(\numnp{}) = \Oc(2^n)$, and will carry an exponentially large subnormalization factor for an $n$-qubit grid. Quantum computers will thus require an alternate strategy for assembly, which we introduce in this section. A visual summary of the difference between the classical and quantum paradigms for the parallel assembly of finite element arrays is given in Figure~\ref{fig:quantum-vs-classical}.

For many linear PDEs of interest---such as the modified Poisson equation given in Section~\ref{subsec:weak-formulation-modified-Poisson}---the elemental stiffness matrix is constant throughout the computational domain. Since the elemental stiffness matrix contains a constant number of entries (that don't vary with the system size), they can be computed classically, and accessed through a $(\textsc{prep}_{K^e},\widetilde{\textsc{prep}}_{K^e})$ state-preparation pair (see Eq.~\eqref{eq:prep-Ke}). 
Switching the order of the summation in Eq.~\eqref{eq:K-classical-summation-order}, we compute
\begin{align}
    K = \sum_{j,k \in [\nen{}]} K_{jk}^\el{} \left( \sum_{\el{} \in [\numel{}]} \ketbra{\mathrm{IX}(j,\el{})}{\mathrm{IX}(k,\el{})} \right) \,.
    \label{eq:K-quantum-summation-order}
\end{align}
The matrix that appears in the inner sum of Eq.~\eqref{eq:K-quantum-summation-order} is the opening wedge for performing efficient parallel assembly on a quantum computer. We formalize this in the following definition.

\begin{figure}
    \centering
    \begin{subfigure}[b]{1.\textwidth}
         \centering
         \includegraphics[scale=0.7]{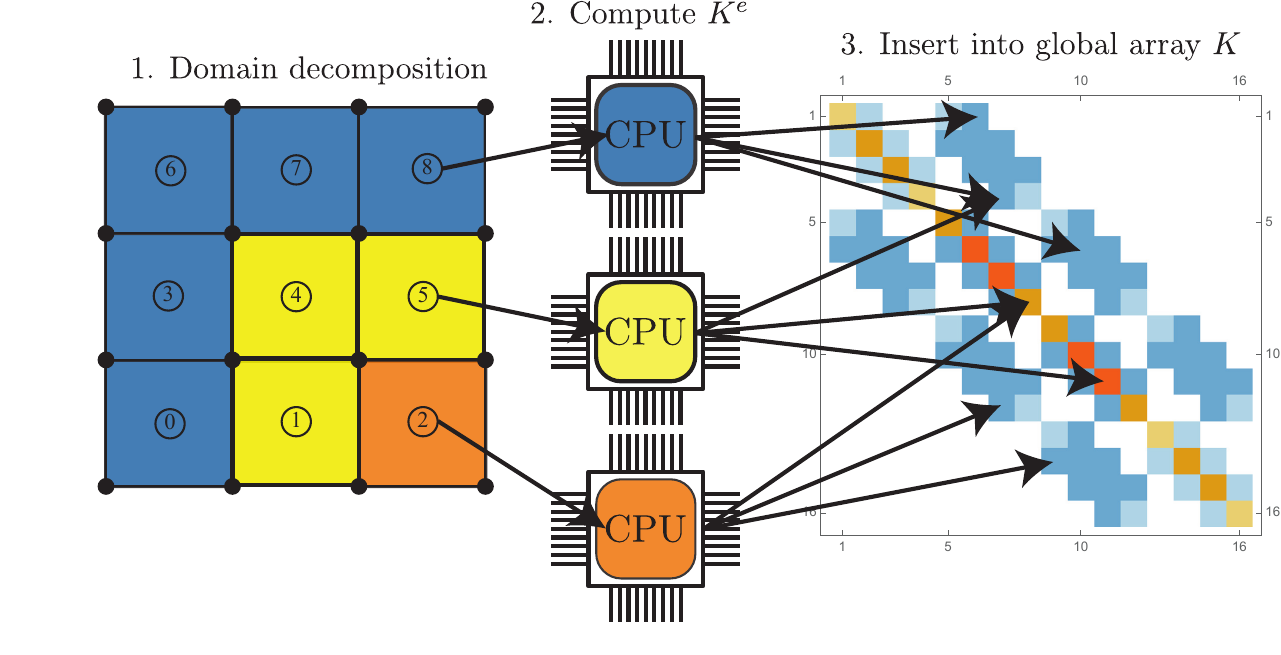}
         \vspace{-10pt}
         \caption{Classical assembly.}
         \label{fig:quantum-vs-classical-a}
    \end{subfigure}
    \\[5pt]
    \begin{subfigure}[b]{1.\textwidth}
         \centering
         \includegraphics[scale=0.7]{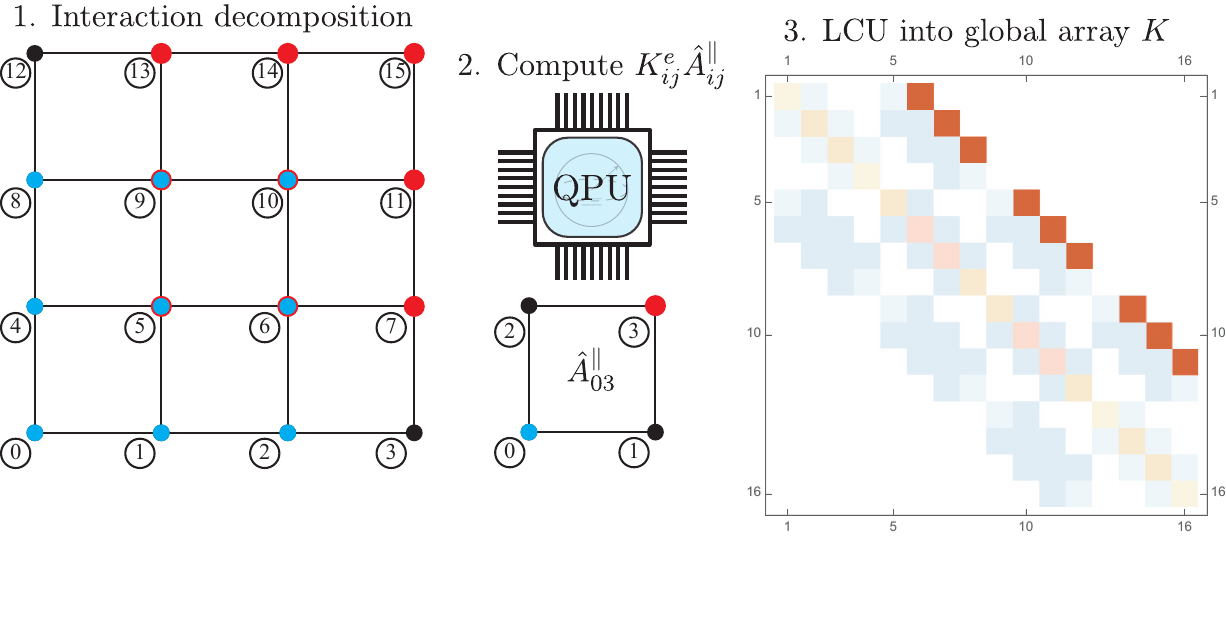}
         \vspace{-30pt}
         \caption{Quantum assembly.}
         \label{fig:quantum-vs-classical-b}
    \end{subfigure}
    \caption{Comparison between classical and quantum strategies for the parallel assembly of finite element arrays. On top (a), we depict how a typical classically assembly process takes place in parallel. An array of classical processors (CPU's) computes the elemental contributions from different parts of the mesh, and inserts the contributions directly into the global array in memory. On the bottom (b), we depict our proposed method for the quantum assembly in parallel. Instead of partitioning the mesh into groups of elements (i.e., domain decomposition), we group local $ij$ interactions and compute the contributions arising from all such interactions simultaneously using a quantum computer (depicted abstractly as a `QPU'). All pairs of interactions are then summed together in an LCU to obtain the assembled global array. As an example, we highlight the interactions between local node numbers $0$ and $3$ in (blue and red, respectively) in a two-dimensional mesh. The sparsity pattern for the corresponding contributions to the global array are shown in the right-hand matrix.}
    \label{fig:quantum-vs-classical}
\end{figure}

\begin{defn}[Unit of Interaction]
\label{defn:unit-of-interaction}
    Let $\mathrm{IX}$ be an $\nen{} \times \numel{}$ connectivity matrix of a mesh consisting of elements with $\nen{}$ nodes locally. We denote and define the unit of interaction between nodes $j,k \in [\nen{}]$ as 
    \begin{equation}
        \uoi_{jk} := \sum_{\el{} \in [\numel{}]} \ketbra{\mathrm{IX}(j,\el{})}{\mathrm{IX}(k,\el{})} \,.
        \label{eq:uoi-defn}
    \end{equation}
\end{defn}
Practically, the unit of interaction identifies the entries in the global stiffness matrix that couple local nodes $j$ and $k$ across all elements simultaneously. We refer to these couplings as `interactions'. In order for the entries of $\uoi_{jk}$ to be exactly $1$, we require that the local-to-global degree of freedom map $\mathrm{IX}(j,\bullet) \colon [\numel{}] \to [\numnp{}]$ be injective for all $j \in [\nen{}]$. The existence of such a map is tantamount to demanding that the mesh does not have any ``extraordinary vertices''---nodal points that are shared by more elements than the number of element nodes $\nen{}$ (see Figure~\ref{fig:IX-injective}). For other mesh types (such as triangular meshes), we need to partition the columns of $\IX{}$ into several matrices $\IX{}^{(0)}, \dots, \IX{}^{(t-1)}$ where each of the $\IX{}^{(s)}(j,\bullet)$ is injective for all $j \in [\nen{}]$ and $s \in [t]$, and $t$ is the number of partitions. All sums involving the connectivity array $\IX{}$ must then involve a sum ``$\sum_{s \in [t]}$'' over each partition, i.e.,
\begin{align}
    K = \sum_{s \in [t]} \sum_{j,k \in [\nen{}]} K_{jk}^\el{} \left( \sum_{\el{} \in [\numel{}^{(s)}]} \ketbra{\mathrm{IX}^{(s)}(j,\el{})}{\mathrm{IX}^{(s)}(k,\el{})} \right) \,,
    \label{eq:K-quantum-summation-order-with-partition}
\end{align}
where $\numel{}^{(s)}$ is the number of elements in each partition. The minimal number of partitions is $t=1$, whereas in the worst case, the maximal number of partitions is  $t = \numel$.
Note that choosing $t = \numel{}$ reverts the summation back to the classical order of assembly given by Eq.~\eqref{eq:K-classical-summation-order} (indeed, the number of elements in each partition in this case is exactly one, i.e., $\numel{}^{(s)} = 1$).
For simplicity, in this work we will assume that the mesh does not contain any extraordinary vertices, and that the node numbering scheme is chosen so that $t \equiv 1$.

Instead of mapping from the local degree of freedom $j \in [\nen{}]$ to its global degree of freedom number, one may equivalently consider the inverse map from a global degree of freedom $J \in [\numnp{}]$ to its local degree of freedom number. This map will be multi-valued whenever a degree of freedom is shared by more than one element. 
One may also consider an inverse map that is more similar to the unit of interaction by considering the interaction between global degrees of freedom $J,K \in [\numnp{}]$, and mapping to the local elemental degree of freedom numbers $j,k \in [\nen{}]$ where such an interaction takes place. This map, however, can still be multi-valued for self-interactions (and indeed, the assembly procedure requires precisely that one sum over all such multiplicities). 

In particular, the multi-valued nature of the map from global interactions to local interactions has been overlooked in the prominent quantum algorithm developed by Ref.~\cite{clader2013preconditioned}. In the supplementary material of Ref.~\cite{clader2013preconditioned}, the authors implement the system of linear equations that emerges from a FEM discretization of Maxwell's equations using degree one N\'{e}d\'{e}lec edge elements~\cite{nedelec1980mixed}. These elements yield a $7$-sparse matrix, which the authors decompose into a sum of one-sparse access oracles. The authors then supply an analytical formula for the one-sparse quantum oracles (Equation~(75) in Ref.~\cite[supplementary material]{clader2013preconditioned}). This formula, however, does not multiply the self-interaction of an edge by the number of elements that it belongs to. 
This bookkeeping could be integrated into the algorithm of~\cite{clader2013preconditioned} for order one elements in a straightforward manner, though it becomes more involved if one wishes to increase the order of the element, or the dimensionality of the problem. 
Our ``unit of interaction'' device circumvents this issue by looping over all interactions at the element level, and partitioning the connectivity array whenever such interactions on the local level do not correspond to a unique set of interactions on a global level.
Ultimately, the finite element assembly procedure is designed precisely to automate this bookkeeping---a feature that is preserved in the Qu-FEM algorithm.

\begin{figure}
     \centering
     \begin{subfigure}[b]{0.2\textwidth}
         \centering
         \includegraphics[scale=0.8]{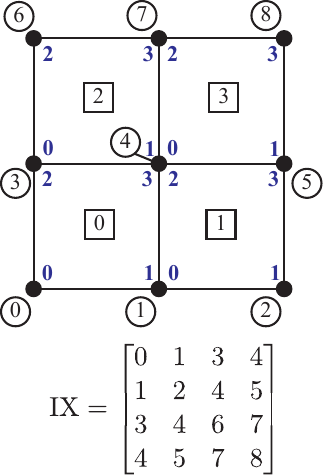}
         \caption{Injective $\IX{}$ map.}
         \label{fig:IX-injective-a}
     \end{subfigure}
     \hfill
     \begin{subfigure}[b]{0.7\textwidth}
         \centering
         \includegraphics[scale=0.8]{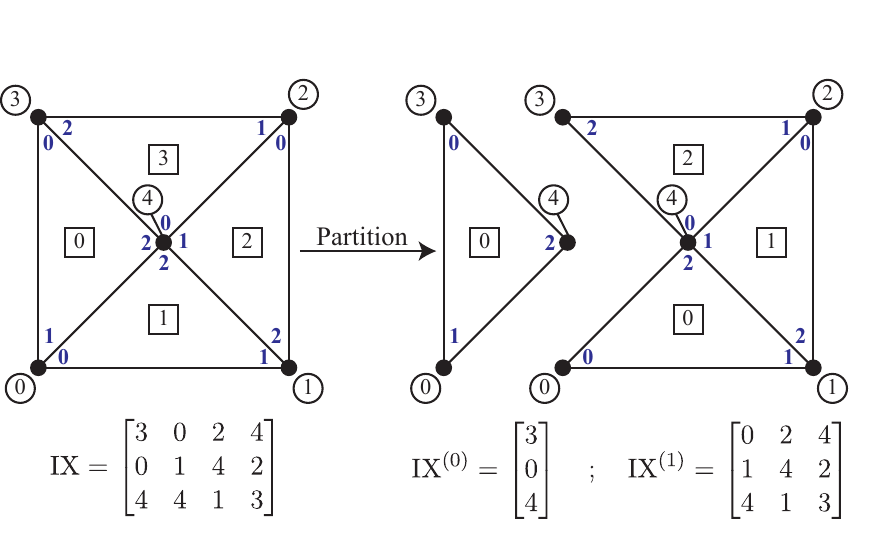}
         \caption{Non-injective $\IX{}$ partitioned into two injective maps $\IX{}^{(0)}$ and $\IX{}^{(1)}$.}
         \label{fig:IX-injective-b}
     \end{subfigure}
    \caption{A square domain is partitioned into four square elements in~(\ref{fig:IX-injective-a}), and four triangular elements in~(\ref{fig:IX-injective-b}). The element numbers are boxed, the global degree of freedom numbers are circled, and the local degree of freedom numbering is given in blue inside each element. In (\ref{fig:IX-injective-a}), we see that the local-to-global degree of freedom map $\IX{}(j,\bullet)$ is injective in the sense that no global degree of freedom that is shared between multiple elements is associated with the same local degree of freedom number. In (\ref{fig:IX-injective-b}), we see that the global degree of freedom ``$4$'' is necessarily associated with a repeated local degree of freedom number. This is remedied by partitioning the mesh into two pieces, each with an injective connectivity matrix.}
    \label{fig:IX-injective}
\end{figure}

Another useful device for assembling finite element arrays is a matrix that simply maps the local node numbers to the global node numbers. This is a factorization of the unit of interaction that arises from observing that $\ketbra{\mathrm{IX}(j,\el{})}{\mathrm{IX}(k,\el{})} = (\ketbra{\mathrm{IX}(j,\el{})}{\el{}}) (\ketbra{\mathrm{IX}(k,\el{})}{\el{}})^\dag $. Taking the superposition of all such projectors $\ketbra{\mathrm{IX}(j,\el{})}{\el{}}$ will enable the efficient assembly of finite element arrays with non-constant elemental contributions, which we will see in Section~\ref{sec:numerical-integration}.

\begin{defn}[Local-to-Global Node Number Indicator Matrix]
\label{defn:local-to-global-node-number-indicator-matrix}
    Let $\mathrm{IX}$ be a $\nen{} \times \numel{}$ connectivity matrix of a mesh (also called the local-to-global node number mapping) consisting of elements with $\nen{}$ nodes locally. We denote and define the local-to-global node number indicator matrix for $j \in [\nen{}]$ as 
    \begin{equation}
        \uoi_{j} := \sum_{\el{} \in [\numel{}]} \ketbra{\mathrm{IX}(j,\el{})}{\el{}} \,.
    \end{equation}
\end{defn}
\noindent
Observe that the unit of interaction $\uoi_{jk}$ is related to the local-to-global node number indicator matrix $\uoi_{j}$ by 
\begin{align}
    \uoi_{jk} &=  \uoi_j \left( \uoi_k \right)^\dag \,.
    \label{eq:Ajk-factoring}
\end{align}
As we demonstrate in the coming sections, however, it is often easier to identify explicit block-encodings for $\uoi_{jk}$ directly rather than as a product of block-encodings from $\uoi_j$.

Finally, we give a block-encoding of $\uoi_j$ in terms of a modified version of a connectivity oracle $\Oc_\IX$. Under the previously stated injectivity condition of the $\IX(j,\bullet)$ map (i.e., for a fixed local degree of freedom, the map between the element numbers and global degrees of freedom is one-to-one; see \cref{fig:IX-injective}), there exists a unitary $\Oc_\IX$ that satisfies 
\begin{align}
	\Oc_\IX \ket{j}\ket{\el} &= \ket{j}\ket{\IX(j,\el)}
	\,,
\end{align}
where $\el{} \in [\numel]$.
Define the circuit
\begin{align}
	\Oc_{< \numel}\ket{\el}\ket{0} &:= \begin{cases}
		\ket{\el}\ket{1} & \text{if } \el < \numel \,,\\
		\ket{\el}\ket{0} & \text{otherwise.}
	\end{cases}
	\,,
\end{align}
which compares the input element number $\el$ to the number of elements $\numel$. This oracle is needed to ensure that we are always referencing an element number that is within the bounds of the set $[\numel]$. 
Then the following circuit gives a $(1,m+1)$-block-encoding of $\uoi_j$ 
\tikzexternalenable
\begin{align}
	\begin{quantikz}
		\lstick{$\ket{j}$} &  
		\gategroup[3,steps=2,style={dashed,rounded
			corners,fill=blue!15, inner
			xsep=2pt},background,label style={label
			position=below,anchor=north,yshift=-0.2cm}]{{$(1,1)$-$\mathrm{BE}(\uoi_{j})$}}
		& \gate[2]{\Oc_\IX} & \rstick{$\ket{j}$} \\
		\lstick{$\ket{\el}$} & \gate[2]{\Oc_{< \numel}} &  & \rstick{$\ket{\IX(j,\el)}$} \\
		\lstick{$\ket{0}$} & & \targ{} & \meter{}
	\end{quantikz}
	\,.
	\label{eq:Aj-general-IX-BE}
\end{align}
\tikzexternaldisable
Note that the qubits holding the $\ket{j}$ state may be discarded at the end of the circuit above.

\subsection{Block-Encoding of Global Arrays for One-Dimensional Lagrange Elements}
\label{subsec:BE-global-arrays-Lagrange-elements}

\begin{figure}
    \centering
    \includegraphics[width=0.7\linewidth]{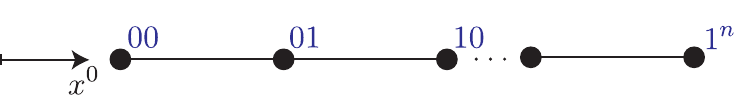}
    \caption{One-dimensional mesh consisting of a total number of global nodal points of~$\numnp{} = 2^n$. 
    }
    \label{fig:1d-mesh}
\end{figure}

Let us begin by considering a one-dimensional mesh (see Figure~\ref{fig:1d-mesh}). The number of nodal points $\numnp{}$ and the number of elements $\numel{}$ for a series of elements with $\nen{}$ degrees of freedom are related by 
\begin{align}
    \numnp{} &= \numel{}(\nen{}-1) + 1 \,.
    \label{eq:numnp-relation}
\end{align}
For $\Qc_p(\Omega^e)$ elements, we have that $\nen{} = p + 1$. Classically, one chooses $p$ and $\numel{}$, from which the number of global nodal points $\numnp{}$ follows. Here, we wish to have both $\numnp{} = 2^n$ nodal points and $\nen{} = 2^m$ local element nodes as powers of two for convenience. The more general case of $\nen{}$ not being a power of two can be dealt with by proper bookkeeping, but we nevertheless want to show that it is possible to achieve $\numnp{} = 2^n$ and $\nen{} = 2^m$, while still satisfying the constraint given by Eq.~\eqref{eq:numnp-relation}. From Eq.~\eqref{eq:numnp-relation}, the number of elements is given by 
\begin{align}
    \numel{} (2^m - 1) &= 2^n -1 \,.
    \label{eq:numel-relation}
\end{align}
This implies that we must have that $(2^m - 1)$ divides $(2^n - 1)$ in order to have an integer number of elements, which first-order elements (for which $\nen{} = 2 \Rightarrow m = 1$) always satisfy. For higher order elements, we may choose $n$ to be a multiple of $m$, say $n = km$ for some integer $k \ge 1$, from which we have 
\begin{align}
    2^n - 1 &= \left(2^m\right)^k - 1\\
    &= \left( 2^m - 1 \right) \left[ \left(2^m\right)^{k-1} + \left(2^m\right)^{k-2} + \dots + \left(2^m\right) + 1\right] \,, \\
    \Rightarrow \numel{} &= \left(2^m\right)^{k-1} + \left(2^m\right)^{k-2} + \dots + \left(2^m\right) + 1\,,
\end{align}
so that $\numel$ is an integer. It follows that arbitrarily large values of $\nen{}$ may be chosen to achieve higher order elements (i.e., $p$-refinement).

\subsubsection{\texorpdfstring{Order one elements ($\Qc_1(\Omega^e)$)}{Order one elements (Q1(Ωe))}}
\label{subsec:BE-1D-degree-1}

We first demonstrate how to assemble the global stiffness matrix $K$ (Eq.~\eqref{eq:K-global-integral-defn}) and the global mass matrix $M$ (Eq.~\eqref{eq:M-global-integral-defn}) for $\Qc_1(\Omega^e)$ elements on a quantum computer. 
We do this by way of Example~\ref{ex:1d-K-global}. There, we also demonstrate an explicit block-encoding for the unit of interaction for meshes of $\Qc_1(\Omega^e)$ elements. 
In addition to enabling the solution of one-dimensional problems, the $p=1$ case will be used when we generalize this procedure to $\Qc_p(\Omega^e)$ elements in the next subsection.

\begin{examplebox}{Parallel Assembly of Global Arrays for $Q_1(\Omega^e)$ Elements}
\label{ex:1d-K-global}
Consider the one-dimensional mesh in Figure~\ref{fig:1d-mesh} with $\numnp{} = 2^n$. 
In one dimension, Eq.~\eqref{eq:uoi-defn} for the unit of interaction $\uoi_{00}$ is given by the $2^n \times 2^n$ matrix
\begin{align}
    \uoi_{00} &= \begin{bmatrix}
        1 & 0 & \cdots & 0 & 0\\
        0 & 1 & \cdots & 0 & 0\\
        \vdots & \vdots & \ddots & \vdots & \vdots\\
        0 & 0 & \cdots & 1 & 0 \\
        0 & 0 & \cdots & 0 & 0
    \end{bmatrix}
    \,.
\end{align}
In this case, the unit of interaction is a projector that omits the last basis state (i.e., the self-interaction of the right-most node in the mesh).
The $(1,1)$-block-encoding of this unit of interaction is closely related to the reversible quantum logical OR gate on $n$ qubits, which for two qubits can be implemented as 
\tikzexternalenable
\begin{align}
    \begin{quantikz}
        & \gate[3]{\mathrm{OR}^{(2)}} &\\
        && \\
        \lstick{$\ket{0}$} && 
    \end{quantikz}
    :=
    \begin{quantikz}
        & \targ{} & \ctrl{2} & \targ{} & \\
        & \targ{} & \ctrl{0} & \targ{} & \\
        \lstick{$\ket{0}$}& \targ{} & \targ{} & & 
    \end{quantikz}
    \,.
\end{align}
\tikzexternaldisable
In the computational basis, the flag qubit at the bottom of the circuit is equal to $\ket{1}$ if and only if all other qubits are in the $\ket{0}$ state.
For $n$ qubits, the $\mathrm{OR}$ gate is 
\tikzexternalenable
\begin{align}
\begin{quantikz}
    \lstick[4]{$\ket{0^n}$}& \targ{} 
    \gategroup[5,steps=3,style={dashed,rounded
    corners,fill=blue!15, inner
    xsep=2pt},background,label style={label
    position=below,anchor=north,yshift=-0.2cm}]{{$\mathrm{OR}^{(n)}$}}
    & \ctrl{4} & \targ{} & \\
    & \targ{} & \ctrl{0} & \targ{} & \\
    & \wireoverride{n} & \gate[nwires=1,style={fill=blue!15,draw=blue!15,text height=1cm}]{\vdots} \wireoverride{n}& \wireoverride{n}& \wireoverride{n}\\
    & \targ{} & \ctrl{0} & \targ{} & \\
    \lstick{$\ket{0}$}& \targ{} & \targ{} & & 
\end{quantikz}
\,.
\label{eq:qOR-n-qubit}
\end{align}
\tikzexternaldisable
Using the $\mathrm{OR}^{(n)}$ along with Pauli-$X$ gates, we obtain the first unit of interaction $\uoi_{00}$. The flag qubit for the $\mathrm{OR}^{(n)}$ gate is also the ancilla qubit (i.e., the `success flag') for the block-encoding. The remaining units of interaction for $\Qc_1(\Omega^e)$ elements can be represented by shifted versions of the first unit of interaction as:

\begin{minipage}[t]{0.5\linewidth}
    \centering
    \tikzexternalenable
    \begin{align}
    \scalebox{0.89}{%
    \begin{quantikz}
        \lstick[3]{$\ket{i}$} & \targ{} 
        \gategroup[4,steps=3,style={dashed,rounded
        corners,fill=blue!15, inner
        xsep=2pt},background,label style={label
        position=below,anchor=north,yshift=-0.2cm}]{{$(1,1)$-$\mathrm{BE}(\uoi_{00})$}}
        & \gate[4]{\mathrm{OR}^{(n)}} & \targ{} & \\[-0.5cm]
        \phantom{a}\vdots\wireoverride{n} &\wireoverride{n}&\wireoverride{n}&\wireoverride{n}&\wireoverride{n} \\[-0.4cm]
        & \targ{} && \targ{}& \\
        \lstick{$\ket{0}$} & \targ{} &&& \meter{}
    \end{quantikz}
	}
    \label{eq:BE-1D-A00}
    \end{align}
    \tikzexternaldisable
\end{minipage}
\begin{minipage}[t]{0.5\linewidth}
    \centering
    \tikzexternalenable
    \begin{align}
    \scalebox{0.89}{%
    \begin{quantikz}
        \lstick[3]{$\ket{i}$} &  \gate[3]{S^{-1}}
        \gategroup[4,steps=2,style={dashed,rounded
        corners,fill=blue!15, inner
        xsep=2pt},background,label style={label
        position=below,anchor=north,yshift=-0.2cm}]{{$(1,1)$-$\mathrm{BE}(\uoi_{01})$}}
        & \gate[4]{\mathrm{BE}(\uoi_{00})} & \\[-0.5cm]
        \phantom{a}\vdots\wireoverride{n}\\
        &&& \\
        \lstick{$\ket{0}$} &&& \meter{}
    \end{quantikz}
	}
    \end{align}
    \tikzexternaldisable
\end{minipage}
\begin{minipage}[t]{0.5\linewidth}
    \centering
    \tikzexternalenable
    \begin{align}
    \scalebox{0.89}{%
    \begin{quantikz}
        \lstick[3]{$\ket{i}$} &  \gate[4]{\mathrm{BE}(\uoi_{00})}
        \gategroup[4,steps=2,style={dashed,rounded
        corners,fill=blue!15, inner
        xsep=2pt},background,label style={label
        position=below,anchor=north,yshift=-0.2cm}]{{$(1,1)$-$\mathrm{BE}(\uoi_{10})$}}
        & \gate[3]{S^1} & \\[-0.5cm]
        \phantom{a}\vdots\wireoverride{n}\\
        &&& \\
        \lstick{$\ket{0}$} &&& \meter{}
    \end{quantikz}
	}
    \end{align}
    \tikzexternaldisable
\end{minipage}
\begin{minipage}[t]{0.5\linewidth}
    \centering
    \tikzexternalenable
    \begin{align}
    \scalebox{0.89}{%
    \begin{quantikz}
        \lstick[3]{$\ket{i}$} &  \gate[3]{S^{-1}}
        \gategroup[4,steps=3,style={dashed,rounded
        corners,fill=blue!15, inner
        xsep=2pt},background,label style={label
        position=below,anchor=north,yshift=-0.2cm}]{{$(1,1)$-$\mathrm{BE}(\uoi_{11})$}}
        & \gate[4]{\mathrm{BE}(\uoi_{00})} & \gate[3]{S^1} & \\[-0.5cm]
        \phantom{a}\vdots\wireoverride{n}\\
        &&&& \\
        \lstick{$\ket{0}$} &&&& \meter{}
    \end{quantikz}
	}
    \end{align}
    \tikzexternaldisable
\end{minipage}
In the above circuit, $S^1$ and $S^{-1}$ are the modular left and right shift operators, which implement the following transformation on $N = 2^n$ states
\begin{equation}
	\begin{split}
		S^{-1}\ket{i} &= \ket{(i-1)\mod N} \,,\\
		S^1\ket{i} &= \ket{(i+1)\mod N}.
	\end{split}
\end{equation}
Note that $S^{-1} = (S^1)^\dag$. 
In matrix form, we have that 
\begin{align}
	S^1 &= \begin{bmatrix}
		0 & 0 & \cdots & 0 & 1\\
		1 & 0 & \cdots & 0 & 0\\
		0 & 1 & \cdots & 0 & 0\\
		\vdots & \vdots & \ddots & \vdots & \vdots\\
		0 & 0 & \cdots & 1 & 0
	\end{bmatrix}
	\,.
\end{align}\mbox{}

During circuit compilation, the Pauli $X$ gates in Eq.~\eqref{eq:BE-1D-A00} cancel out with the Pauli $X$ gates in Eq.~\eqref{eq:qOR-n-qubit}, resulting in a simple multi-controlled NOT gate $C^n(\NOT)$:
\tikzexternalenable
\begin{align}
\begin{quantikz}
    \lstick[3]{$\ket{i}$} & \ctrl{2} 
    \gategroup[4,steps=1,style={dashed,rounded
    corners,fill=blue!15, inner
    xsep=2pt},background,label style={label
    position=below,anchor=north,yshift=-0.2cm}]{{$(1,1)$-$\mathrm{BE}(\uoi_{00})$}}
    & \\
    \phantom{a}\vdots\wireoverride{n}\\
    & \ctrl{1} & \\
    \lstick{$\ket{0}$} & \targ{} &
\end{quantikz}
\,.
\label{eq:BE-1D-A00-simple}
\end{align}
\tikzexternaldisable

The block-encodings for the remaining units of interaction need not be shifted versions of $\uoi_{00}$. For example, a $(1,1)$-block-encoding of $\uoi_{11}$ with a reduced Toffoli count is given by 
\tikzexternalenable
\begin{align}
\begin{quantikz}
    \lstick[3]{$\ket{i}$} &  
    \gategroup[4,steps=2,style={dashed,rounded
    corners,fill=blue!15, inner
    xsep=2pt},background,label style={label
    position=below,anchor=north,yshift=-0.2cm}]{{$(1,1)$-$\mathrm{BE}(\uoi_{11})$}}
    & \gate[4]{\mathrm{OR}^{(n)}} & \\[-0.5cm]
    \phantom{a}\vdots\wireoverride{n}\\
    &&& \\
    \lstick{$\ket{0}$} & \targ{} && \meter{}
\end{quantikz}
\,.
\end{align}
\tikzexternaldisable

Finally, we can assemble the global stiffness matrix by performing the LCU 
\begin{align}
    K &= \sum_{i,j \in [2]} K^e_{ij} \uoi_{ij} 
    \,,
\end{align}
which may be represented as the circuit 

\tikzexternalenable
\begin{align}
\scalebox{0.9}{%
\begin{quantikz}[column sep=0.1cm]
    \lstick{$\ket{0}$}& \qwbundle{\log(p+1)} &&& \gate[2]{\textsc{prep}_{K^e}} 
    \gategroup[6,steps=10,style={dashed,rounded
    corners,fill=blue!15, inner
    xsep=2pt},background,label style={label
    position=below,anchor=north,yshift=-0.2cm}]{{$\left( \sum_{j,k \in [p+1]} |K_{jk}^e|, \Oc(\log(p+1)) \right)\mathrm{-BE}(K)$}}
    & \octrl{2} && \octrl{2} && \ctrl{2} && \ctrl{2} && \gate[2]{\widetilde{\textsc{prep}}_{K^e}} & \meter{} \rstick{$\bra{0}$}\\
    \lstick{$\ket{0}$} & \qwbundle{\log(p+1)} &\hphantomgate{}&\hphantomgate{}& & \octrl{0} && \ctrl{0} && \octrl{0} && \ctrl{1} && & \meter{} \rstick{$\bra{0}$}\\
    \lstick{$\ket{i}$} & \qwbundle{n} &&& & \gate[2]{U_{\uoi_{00}}} && \gate[2]{U_{\uoi_{01}}} && \gate[2]{U_{\uoi_{10}}} && \gate[2]{U_{\uoi_{11}}} && & \\
    \lstick{$\ket{0}$} &&&&&& \octrl{1} && \octrl{1} && \octrl{1} && \octrl{1} && \meter{} \rstick{$\bra{0}$}\\
    \lstick[1]{$\ket{1}$} &&&&&& \gate[2]{S^{-1}} && \gate[2]{S^{-1}} && \gate[2]{S^{-1}} && \gate[2]{S^{-1}} && \meter{} \rstick{$\bra{0}$}\\
    \lstick{$\ket{0}$}& \qwbundle{\log((p+1)^2)} &&&&&&&&&&&&& \meter{} \rstick{$\bra{0}$}
\end{quantikz}
}
\label{eq:BE-K-circuit}
\end{align}
\tikzexternaldisable
Here, we have used the compression gadget from Ref.~\cite{fang2023time} (utilized in the bottom three wires in \cref{eq:BE-K-circuit}) to reduce the number of ancilla needed to successfully apply all four block-encodings. This will be useful when we generalize $p^\text{th}$ order elements, where we will have $\nen{}^2$ terms in the LCU.
\Cref{eq:BE-K-circuit} block-encodes the global stiffness matrix $K$ for $\Qc_1(\Omega^e)$ elements.\\

Since 
\begin{align}
    M &= \sum_{i,j \in [2]} M^e_{ij} \uoi_{ij} 
    \,,
\end{align}
the block-encoding of the global mass matrix $M$ is exactly the same as the above circuit, but with the prepare oracles replaced with those for the elemental mass matrix $M^e$. 
Using the basis functions for the $\Qc_1(\Omega^e)$ element given by Eq.~\eqref{eq:Ne-1d}, the elemental stiffness and mass matrices can be computed analytically as 
\begin{align}
    K^e &= \begin{bmatrix}
        1 & -1 \\
        -1 & 1 \\
    \end{bmatrix}
    \quad \text{and} \quad
    M^e = \begin{bmatrix}
        \frac{1}{3} & \frac{1}{6} \\[3pt]
        \frac{1}{6} & \frac{1}{3} \\
    \end{bmatrix}
    \,.
\end{align}
From these elemental arrays, we see that Eq.~\eqref{eq:BE-K-circuit} provides a $(4,6)$-block-encoding of $K$, and a $(1,6)$-block-encoding of $M$. 
Explicit block-encodings for the prepare oracles can be given as follows:
\begin{align}
    \textsc{prep}_{K^e} &= S^{-1}\left((S H Z) \otimes H\right) = \frac{1}{2} 
    \begin{bmatrix}
        1 & -1 & -1 & 1 \\
        i & i & i & i \\
        i & -i & i & -i \\
        1 & 1 & -1 & -1 \\
    \end{bmatrix} \,,\\
    \widetilde{\textsc{prep}}_{K^e} &= \left((Z H S) \otimes H\right)S^{1} = \frac{1}{2}
    \begin{bmatrix}
        1 & i & i & 1 \\
        -1 & i & -i & 1 \\
        -1 & i & i & -1 \\
        1 & i & -i & -1 \\
    \end{bmatrix} \,, \\
    \textsc{prep}_{M^e} &= S^{-1} \left( R_y(\theta_M) \otimes H \right) = 
    \begin{bmatrix}
        \frac{1}{\sqrt{3}} & -\frac{1}{\sqrt{3}} & -\frac{1}{\sqrt{6}} & \frac{1}{\sqrt{6}} \\
        \frac{1}{\sqrt{6}} & \frac{1}{\sqrt{6}} & \frac{1}{\sqrt{3}} & \frac{1}{\sqrt{3}} \\
        \frac{1}{\sqrt{6}} & -\frac{1}{\sqrt{6}} & \frac{1}{\sqrt{3}} & -\frac{1}{\sqrt{3}} \\
        \frac{1}{\sqrt{3}} & \frac{1}{\sqrt{3}} & -\frac{1}{\sqrt{6}} & -\frac{1}{\sqrt{6}} \\
    \end{bmatrix} \,, \\
    \widetilde{\textsc{prep}}_{M^e} &= \textsc{prep}_{M^e}^\dag = \left( R_y(\theta_M)^\dag \otimes H \right) S^{1} \,,
\end{align}
where $S := \ketbra{0}{0} + i \ketbra{1}{1}$ is the phase gate (not to be confused with the shift matrices $S^1$ and $S^{-1}$), and $R_y(\theta) := e^{-i\theta Y/2}$ is the rotation of a qubit by angle $\theta$ around the $y$-axis of the Bloch sphere. Here, we choose $\theta_{M} := 2\arccos\left(\sqrt{2/3}\right)$ so that 
\begin{align}
    R_y(\theta_M) &= 
    \begin{bmatrix}
        \cos(\theta_M/2) & -\sin(\theta_M/2)\\
        \sin(\theta_M/2) & \cos(\theta_M/2)
    \end{bmatrix}
    = \begin{bmatrix}
        \sqrt{\frac{2}{3}} & -\frac{1}{\sqrt{3}} \\[3pt]
        \frac{1}{\sqrt{3}} & \sqrt{\frac{2}{3}} \\
    \end{bmatrix}
    \,.
\end{align}

\end{examplebox}

The LCU circuit of the units of interaction with the elemental stiffness matrix $K^e$ in Eq.~\eqref{eq:BE-K-circuit} yields the following result for assembling global arrays for PDEs with constant elemental arrays.

\begin{thm}[Query Complexity for the Assembly of $\Qc_1(\Omega^e)$ Elements in $1$D]
\label{thm:query-complexity-Assembly-1D-degree-1}
    Let $A$ be a finite element array of size of size $2^n \times 2^n$ assembled from constant elemental contributions $A_{ij}^e$ using first-order Lagrange elements $\Qc_1(\Omega^e)$, i.e., $A = \sum_{j,k \in [2]} A^e_{jk} \uoi_{jk}$ where $\uoi_{jk}$ are the units of interaction for a one-dimensional mesh of $\Qc_1(\Omega^e)$ elements with $\numnp{} = 2^n$. Then one can implement a 
    \begin{align}
        \left( \sum_{j,k \in [2]} |A_{jk}^e| , 6 \right)\mathrm{-BE}(A)
        \label{eq:BE-cost-A}
    \end{align}
    using $\Oc\left( n \right)$ Toffoli or simpler gates.
\end{thm}

\begin{proof}
    The circuit in Eq.~\eqref{eq:BE-K-circuit} implements $\LCU\left( (U_{\uoi_{jk}})_{j,k \in [2]}, (A_{jk}^e)_{j,k \in [2]} \right)$ using $6$ ancilla qubits, so that Eq.~\eqref{eq:BE-cost-A} follows from Lemma~\ref{lem:linear-combination-of-block-encodings}. Additionally, the multi-qubit controlled NOT gate $C^n(\NOT)$ is a $(1,1)$-block-encoding for the $2^n \times 2^n$ unit of interaction $\uoi_{00}$ (see Eq.~\eqref{eq:BE-1D-A00}), which uses $\Oc(n)$ Toffoli gates~\cite{barenco1995elementary}. The remaining units of interaction are given by $\uoi_{jk} = S^{j} \uoi_{00} S^{-k}$, so that a $(1,1)$-block-encoding for $\uoi_{jk}$ is given by $(S^j \otimes I)C^n(\NOT)(S^{-k} \otimes I)$. Since the shift matrix on $n$ qubits can be implemented with $\Oc(n)$ Toffoli gates\footnote{The implementations of $C^n(\NOT)$ and $S^1$ using $\Oc(n)$ Toffoli gates requires one additional ancilla qubit~\cite{barenco1995elementary}. Without additional ancilla, this requires $\Oc(n^2)$ Toffoli gates. In this work, we only count ancilla that require post-selection in the block-encoding cost; hence, we choose the Toffoli-optimal $\Oc(n)$ constructions.
    }~\cite{kharazi2024explicitblockencodingsboundary}, it follows that all the $(1,1)\mathrm{-BE}(\uoi_{jk})$ require $\Oc(n)$ Toffoli gates. Finally, we note that the state-preparation pair $(\PREP_{A^e},\widetilde{\PREP}_{A^e})$ for the classically precomputed $2 \times 2$ elemental matrix $A^e$ can be implemented using $\Oc(1)$ Toffoli gates~\cite{babbush2018encoding}.
\end{proof}

\begin{rem}
    For periodic boundary conditions, the unit of interaction simplifies to $\uoi_{00} \equiv I$. The remaining units of interaction are then $\uoi_{jk} = S^{j-k}$, so that Theorem~\ref{thm:query-complexity-Assembly-1D-degree-1} holds for periodic boundary conditions as well. 
    More general boundary conditions---such as Dirichlet and Neumann boundary conditions---are addressed in Section~\ref{sec:constraints}.
\end{rem}

\noindent
Examples of finite element arrays that follow the complexity in \cref{thm:query-complexity-Assembly-1D-degree-1} are the assembled $1$D stiffness $K = \sum_{i,j \in [p+1]} K^e_{ij} \uoi_{ij}$ and mass $M = \sum_{i,j \in [p+1]} M^e_{ij} \uoi_{ij}$ matrices of size $N \times N$ given in Example~\ref{ex:1d-K-global}.
This is a special case of the more general result for $\Qc_p(\Omega^e)$ elements derived in the next section; we nevertheless state it separately here due to its correspondence with Example~\ref{ex:1d-K-global}. Additionally, the Toffoli count for the order one elements scales linearly with the number of qubits $n$, demonstrating that a direct block-encoding of the unit of interaction can lead to more optimal circuit costs.

\subsubsection{\texorpdfstring{Order $p$ elements ($\Qc_p(\Omega^e)$)}{Order $p$ elements (Qp(Ωe))}}

\label{subsec:BE-1D-degree-p}

We now turn our attention to implementing higher order elements (i.e., $\Qc_p(\Omega^e)$ elements) in one dimension. Instead of increasing the number of elements $\numel$ in the mesh, it is sometimes desirable to increase the order $p$ of the elements in order to increase the accuracy of the solution~\cite{brenner2008mathematical,Solin2003,Babuska1981}. The assembly of $\Qc_p(\Omega^e)$ elements corresponds to assembling global arrays of the form 
\begin{align}
    A &= \sum_{j,k \in [p+1]} A_{jk}^e \uoip_{jk} \,,
\end{align}
where $A^e$ is the elemental array, and $\uoip_{jk}$ are the units of interaction for $\Qc_p(\Omega^e)$ elements. Recall from Eq.~\eqref{eq:numnp-relation} that to guarantee $\numnp{}$ is a power of two and $\numel{}$ is an integer, we choose 
\begin{align}
    p = 2^m - 1 \quad\text{and}\quad \numnp{} = 2^{mk}\,,
    \label{eq:p-numnp-powers-of-two}
\end{align}
where $1 \le m,k \in \bbZ$. 
Observe that the global array is assembled by looping through all $\nen{}^2 = (p+1)^2$ interactions present in a $\Qc_p(\Omega^e)$ element. 
Similar to Section~\ref{subsec:BE-1D-degree-1}, the strategy will be to find a block-encoding $U_{\uoip_{00}} \in (1,m)\mathrm{-BE}(\uoip_{00})$, from which block-encodings for the remaining units of interaction are given by 
\begin{align}
    U_{\uoip_{jk}} &= \left( S^j \otimes I^{\otimes m} \right) U_{\uoip_{00}} \left( S^{-k} \otimes I^{\otimes m} \right)
    \,.
    \label{eq:Ajk-relation-to-A00}
\end{align}
Thus, without loss of generality, we will focus on block-encoding $\uoip_{00}$. Before giving the circuit for this block-encoding, it is instructive to examine the structure of $\uoip_{00}$ (for one-dimensional Lagrange elements) for a particular choice of $p > 1$. We do this in Example~\ref{ex:uoi-p-equals-3} below.\\

\begin{examplebox}{The unit of interaction $\uoim{3}_{00}$ for $\Qc_3(\Omega^e)$ elements}
\label{ex:uoi-p-equals-3}

\begin{minipage}[t]{1.\linewidth}
    \vspace*{0pt}
    \begin{center}
        \includegraphics[width=1.\linewidth]{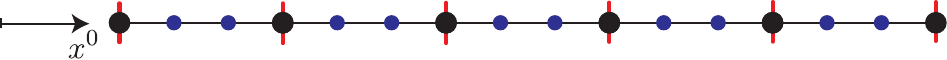}
        \captionof{figure}{A one-dimensional mesh with $m,k=2$ (that yields $p = 3$, $\numnp{} = 16$, and $\numel{} = 5$). All circles denote nodal points, with black circles denoting nodal points at the boundary of an element, and the smaller, blue circles denoting nodal points interior to an element. The vertical red lines delineate the beginning and end of an element.}\label{fig:1d-mesh-p-3}
    \end{center}
\end{minipage}
\vspace{5pt}

Consider the mesh of Lagrange elements in \cref{fig:1d-mesh-p-3} with $m,k=2$ so that $p = 3$ and $\numnp{} = 16$ (see Eq.~\eqref{eq:p-numnp-powers-of-two}). The number of (degree $3$) elements in this mesh is then $\numel{} = \frac{\numnp{} - 1}{p} = 5$. Numbering the (global) degrees of freedom sequentially from left-to-right and using Eq.~\eqref{eq:uoi-defn} yields 
\setcounter{MaxMatrixCols}{16}
\begin{align}
    \uoim{3}_{00} &= 
    \begin{bmatrix}
        1& & & & & & & & & & & & & & &\\
        & 0& & & & & & & & & & & & & &\\
        & & 0& & & & & & & & & & & & &\\
        & & & 1& & & & & & & & & & & &\\
        & & & & 0& & & & & & & & & & &\\
        & & & & & 0& & & & & & & & & &\\
        & & & & & & 1& & & & & & & & &\\
        & & & & & & & 0& & & & & & & &\\
        & & & & & & & & 0& & & & & & &\\
        & & & & & & & & & 1& & & & & &\\
        & & & & & & & & & & 0& & & & &\\
        & & & & & & & & & & & 0& & & &\\
        & & & & & & & & & & & & 1& & &\\
        & & & & & & & & & & & & & 0& &\\
        & & & & & & & & & & & & & & 0&\\
        & & & & & & & & & & & & & & &0
    \end{bmatrix}
    \,,
\end{align}

with all off-diagonal elements being zero. This matrix is a projector, for which the general block-encoding strategy is to introduce an ancilla and perform controlled flips on the basis states where the matrix is supported~\cite{camps2023explicitquantumcircuitsblock}. To extract the pattern in which the ones appear along the diagonal, we can rewrite this matrix as
\begin{align}
    \uoim{3}_{00} &= 
    \begin{bmatrix}
        \begin{bmatrix}
            1 \\ &1\\ &&1\\ &&&1\\ &&&&1
        \end{bmatrix} 
        \otimes 
        \begin{bmatrix}
            1 \\ &0\\ &&0
        \end{bmatrix}  &\\
        & 0
    \end{bmatrix}
    = \begin{bmatrix}
        I_{\log 5} \otimes \ket{0}_{\log 3}\bra{0}_{\log 3} &\\
        & 0
    \end{bmatrix}
    \label{eq:uoi-p-3}
    \,.
\end{align}
In other words, the $i^\text{th}$ entry along the diagonal of $\uoim{3}_{00}$ can be obtained from the function  
\begin{align}
    f(i) &= \begin{cases}
        1 & \text{if } i \mod 3 = 0 \text{ and } i \ne 15\\
        0 & \text{otherwise}
    \end{cases} 
    \,.
\end{align}

\end{examplebox}

The result from Example~\ref{ex:uoi-p-equals-3} in Eq.~\eqref{eq:uoi-p-3} generalizes to $p^\text{th}$ order elements as 
\begin{align}
    \uoip_{00} &= 
    \begin{bmatrix}
        I_{\log\numel{}} \otimes \ket{0}_{\log p}\bra{0}_{\log p} &\\
        & 0
    \end{bmatrix}
    \,,
    \label{eq:A00-p-general-form}
\end{align}
where once again, all off-diagonal elements are zero. That is, the $i^\text{th}$ entry along the diagonal of $\uoip_{00}$ corresponds to the function
\begin{align}
    f(i) &= \begin{cases}
        1 & \text{if } i \pmod p = 0 \text{ and } i \ne \numnp{}-1\\
        0 & \text{otherwise}
    \end{cases} 
    \,.
    \label{eq:uoi-diagonal-function-order-p}
\end{align}
To block-encode $\uoip_{00}$, we thus need to compute the remainder $r := i \mod p$, and apply a controlled flip on an ancilla qubit if $r = 0$. Note that there are several quantum circuits that implement classical arithmetic operations efficiently (see Ref.~\cite{wang2024comprehensive} for a review).

In particular, a circuit for ``quantum division'' can be used to compute $i \mod p$. We represent this circuit as a unitary $U_{\div}$ that takes in a dividend $\ket{i}_n$, a divisor $\ket{p}_m$, a set of ancilla $\ket{0}_n$, and returns the quotient $\ket{Q}_n$, the divisor $\ket{p}_m$, and the remainder $\ket{r}_n$, i.e., 
\begin{align}
    U_{\div} \ket{i}_n \ket{p}_m \ket{0}_n &= \ket{Q}_n \ket{p}_m \ket{(i \mod p)}_n
    \,.
    \label{eq:quantum-division-unitary}
\end{align}
Ref.~\cite{orts2024quantum} gives a circuit for quantum division that has a Toffoli count of $46mn - 46m^2 + 48m - 2n - 2$. 
Using an additional set of ancilla $\ket{0}_m$ and a $C^n(\NOT)$ gate, we can form a unitary $U_{\% p}$ that utilizes $\Oc(n m)$ Toffoli gates and acts as\footnote{Internal to the $U_{\% p}$ circuit are an additional $n + m$ ancilla qubits that are restored to the state $\ket{0}_{n+m}$ at the end of the computation. Since none of these qubits need to be post-selected, we do not include them in the block-encoding cost.} 
\begin{align}
    U_{\% p} \ket{i}_n \ket{0}_m &= \ket{i}_n \ket{(i \mod p)}_m 
    \,.
\end{align}
To represent this as quantum circuitry, write $p$ in binary as $p_{m-1} \cdots p_0$, and let $U_p := \bigotimes_{i=1}^m \NOT^{p_{m-i}}$ be the unitary that prepares $\ket{p}$ (i.e., $U_p\ket{0}_m = \ket{p}$). The modular arithmetic unitary $U_{\% p}$ is then given by
\tikzexternalenable
\begin{align}
\begin{quantikz}
    \lstick{$\ket{i}_n$}& \qwbundle{n} & 
    \gategroup[4,steps=8,style={dashed,rounded
    corners,fill=blue!20, inner
    xsep=2pt},background,label style={label
    position=below,anchor=north,yshift=-0.2cm}]{{$U_{\% p}$}}
    && & \gate[3]{U_{\div}} & & \gate[3]{U_{\div}^\dagger} & && \rstick{$\ket{i}_n$}\\
    \wireoverride{n}&\wireoverride{n}& 
    \ket{0}_m\wireoverride{n}& \qwbundle{m} & \gate{U_p} & & & & \gate{U_p^\dagger} & \ket{0}_m \\
    \wireoverride{n}&\wireoverride{n}&  
    \ket{0}_n\wireoverride{n}& \qwbundle{n} & & & \ctrl{1} & & & \ket{0}_n\\
    \lstick{$\ket{0}_m$} & \qwbundle{m} & & & & & \targ{} & & && \rstick{$\ket{(i \mod p)}_m$}
\end{quantikz}
\,.
\label{eq:U-mod-p}
\end{align}
\tikzexternaldisable
From Eq.~\eqref{eq:uoi-diagonal-function-order-p}, a block-encoding for the unit of interaction $\uoip_{00}$ is then given by 
\tikzexternalenable
\begin{align}
\begin{quantikz}
    \lstick{$\ket{i}_n$}& \qwbundle{n} & 
    \gategroup[4,steps=7,style={dashed,rounded
    corners,fill=blue!20, inner
    xsep=2pt},background,label style={label
    position=below,anchor=north,yshift=-0.2cm}]{{$U_{\uoip_{00}} \in (1,2)\mathrm{-BE}(\uoip_{00})$}}
    & & \gate[2]{U_{\% p}} & & \gate[2]{U_{\% p}^\dagger} & \ctrl{3} && \rstick{$\ket{i}_n$}\\
    \wireoverride{n}&\wireoverride{n}& 
    \ket{0}_m\wireoverride{n}& \qwbundle{m}  & & \octrl{1} & & & \ket{0}_m \\
    \lstick{$\ket{0}$} && & & \targ{} & \targ{} & & && \meter{} \rstick{$\bra{0}$}\\
    \lstick{$\ket{0}$} && & & & & & \targ{} && \meter{} \rstick{$\bra{0}$}
\end{quantikz}
\,.
\label{eq:BE-uoi-pth-order}
\end{align}
\tikzexternaldisable
To verify this algebraically, observe that\footnote{Recall that with $b$ as a binary string of size $k-j+1$, the notation $C_{b}^{j:k}(U_m)$ is used to denote the controlled application of unitary $U$ on qubit number $m$, conditioned on qubits $j$ through $k$ being in the $\ket{b}$ state.} 
\begin{align}
    \ket{i}_n \ket{0}_m \ket{0} \ket{0}
    &\xrightarrow{U_{\% p} \otimes \NOT \otimes I} \ket{i}_n \ket{(i \mod p)}_m \ket{1} \ket{0}\\
    &\xrightarrow{I^{\otimes n} \otimes C_{0^m}^{n:n+m-1}(\NOT_{n+m}) \otimes I}
    \ket{i}_n \ket{(i \mod p)}_m \ket{\textsc{bool}(i \mod p \ne 0)} \ket{0}\\
    &\xrightarrow{U_{\% p}^\dagger \otimes I \otimes I} 
    \ket{i}_n \ket{0}_m \ket{\textsc{bool}(i \mod p \ne 0)} \ket{0}\\
    &\xrightarrow{ C_{1^n}^{0:n-1}(\NOT_{n+m+1}) }
    \ket{i}_n \ket{0}_m \ket{\textsc{bool}(i \mod p \ne 0)} \ket{\textsc{bool}(i = 2^n -1)}
    \,,
\end{align}
where $\textsc{bool}(\bullet) \in \{0,1\}$ evaluates the logical boolean. 
It then follows from Eq.~\eqref{eq:uoi-diagonal-function-order-p} that 
\begin{align}
    (I^{\otimes n} \otimes \bra{0} \otimes \bra{0}) U_{\uoip_{00}} (I^{\otimes n} \otimes \ket{0} \otimes \ket{0}) &= \sum_{i \in [\numnp{}]} f(i) \ketbra{i}{i}
    = \uoip_{00}
    \,.
\end{align}
Block-encodings for the remaining units of interaction $\uoip_{jk}$ are given by 
\tikzexternalenable
\begin{align}
\begin{quantikz}
    \lstick[3]{$\ket{i}_n$} &  \gate[3]{S^{-k}}
    \gategroup[4,steps=3,style={dashed,rounded
    corners,fill=blue!15, inner
    xsep=2pt},background,label style={label
    position=below,anchor=north,yshift=-0.2cm}]{{$U_{\uoip_{jk}} \in (1,2)$-$\mathrm{BE}(\uoip_{jk})$}}
    & \gate[4]{U_{\uoip_{00}}} & \gate[3]{S^j} & \\[-0.5cm]
    \phantom{a}\vdots\wireoverride{n}\\
    &&&& \\
    \lstick{$\ket{0}_2$} &&&& \meter{}
\end{quantikz}
\,.
\label{eq:BE-uoi-pth-order-shifts}
\end{align}
\tikzexternaldisable
We summarize this result in the following proposition.

\begin{prop}[Block-encoding the unit of interaction for $\Qc_p(\Omega^e)$ elements]
\label{prop:BE-unit-of-interaction-p}
    Equations~\eqref{eq:BE-uoi-pth-order} and~\eqref{eq:BE-uoi-pth-order-shifts} give $(1,2)$-block-encodings for the units of interaction $\uoip_{jk}$ using $\Oc(n m)$ Toffoli or simpler gates.
\end{prop}

With the units of interaction in hand, the assembly of global arrays follows the LCU  
\begin{align}
    A &= \sum_{j,k \in [p+1]} A^e_{jk} \uoip_{jk} 
    \,.
\end{align}
We utilize the compression gadget~\cite{fang2023time} to share the ancilla qubit for the block-encoding of all the units of interactions in the following circuit:

\tikzexternalenable
\begin{align}
	\begin{quantikz}
		\lstick{$\ket{0}$}& \qwbundle{\log(\mathrm{nen^2})} &&& \gate{\textsc{prep}_{A^e_{jk}}}
		\gategroup[5,steps=4,style={dashed,rounded
			corners,fill=blue!20, inner
			xsep=2pt},background,label style={label
			position=below,anchor=north,yshift=-0.2cm}]{{$\mathrm{BE}(A)$}}
		&\qw \mathlarger{\mathlarger{\mathlarger{\mathlarger{\oslash}}}} \vqw{1} 
		&\qw \mathlarger{\mathlarger{\mathlarger{\mathlarger{\oslash}}}} \vqw{3}
		& \gate{\widetilde{\textsc{prep}}_{A^e_{jk}}} &\meter{} \rstick{$\bra{0}^{\otimes \log(\mathrm{nen^2})}$}\\
		\lstick{$\ket{i}$}& \qwbundle{\log(\mathrm{numnp})}  &&& & \gate[2]{U_{\uoip_{jk}}} & & & \\
		\lstick{$\ket{0}$}& \qwbundle{2} & \qw  &&&& \octrl{1} && \meter{} \rstick{$\bra{0}^{\otimes 2}$}\\
		\lstick{$\ket{1}$} && \qw &&&& \gate[2]{S^{-1}} & &  \meter{} \rstick{$\bra{0}$}\\
		\lstick{$\ket{0}$} & \qwbundle{\log(\mathrm{nen}^2)} &&&& & & & \meter{} \rstick{$\bra{0}^{\otimes\log(\nen{}^2)}$}
	\end{quantikz}
	\,.
	\label{eq:A-BE-circuit-pth-order}
\end{align}
\tikzexternaldisable

\begin{thm}[Query Complexity for the Assembly of $\Qc_p(\Omega^e)$ Elements in $1$D]
\label{thm:query-complexity-Assembly-1D-degree-p}
    Let $A = \sum_{i,j \in [p+1]} A^e_{ij} \uoi_{ij}$ be the assembled $1$D stiffness matrix of size $2^n \times 2^n$ for $p^\text{th}$ order Lagrange elements $\Qc_p(\Omega^e)$, where the elemental contributions $A^e$ are constant. Then one can implement a 
    \begin{align}
        \left( \sum_{j,k \in [p+1]} |A_{jk}^e|, 4m+3 \right)\mathrm{-BE}(A)
        \label{eq:eq:BE-cost-A-degree-p}
    \end{align}
    using $\Oc\left( n m (p+1)^2\right)$ Toffoli or simpler gates (here, $m = \lceil \log (p + 1) \rceil$).
\end{thm}

\begin{proof}
    The proof is similar to the $p=1$ case in Theorem~\ref{thm:query-complexity-Assembly-1D-degree-1}. Here, the circuit in Eq.~\eqref{eq:A-BE-circuit-pth-order} implements $\LCU\left( (U_{\uoi_{jk}})_{j,k \in [p+1]}, (A_{jk}^e)_{j,k \in [p+1]} \right)$ using $4m + 3$ ancilla qubits, so that Eq.~\eqref{eq:eq:BE-cost-A-degree-p} follows from Lemma~\ref{lem:linear-combination-of-block-encodings}. 
    The state-preparation pair $(\PREP_{A^e},\widetilde{\PREP}_{A^e})$ for the classically precomputed $(p+1) \times (p+1)$ elemental matrix $A^e$ can be implemented using $\Oc((p+1)^2)$ Toffoli gates\footnote{Ref.~\cite{babbush2018encoding} demonstrates a technique for preparing a state with $L$ unique coefficients using $4L + \Oc(\log(1/\epsilon))$ Toffoli or simpler gates. Here, $\epsilon$ is the largest absolute error in the prepared amplitudes, which we take here to be a constant.}.
    Additionally, the $(1,2)$-block-encoding for the units of interaction $\uoip_{jk}$ uses $\Oc(n m)$ Toffoli or simpler gates. Since there are $(p+1)^2$ units of interaction, it follows that the LCU circuit has a $\Oc(n m (p+1)^2)$ Toffoli or simpler complexity.
\end{proof}

\begin{rem}
    Omitting the last multi-controlled $\NOT{}$ gate from Eq.~\eqref{eq:BE-uoi-pth-order} (i.e., omitting $C_{1^n}^{0:n-1}(\NOT_{n+m+1})$) gives the unit of interaction $\uoip_{00}$ for periodic boundary conditions.
\end{rem}

Since quantum circuits for classical arithmetic are still an active area of research~\cite{wang2024comprehensive}, future developments in quantum division circuits could reduce the query complexity for the $\uoip$ block-encoding. 
Henceforth, we will drop the superscript $(\bullet)^p$ on the unit of interaction, as the order of the element will be clear from the number of terms in the LCU.

\subsection{Block-Encoding of Global Arrays for Lagrange Tensor Product Elements}
\label{subsec:BE-global-arrays-tensor-product-elements}

\begin{figure}
    \centering
    \includegraphics[width=0.5\linewidth]{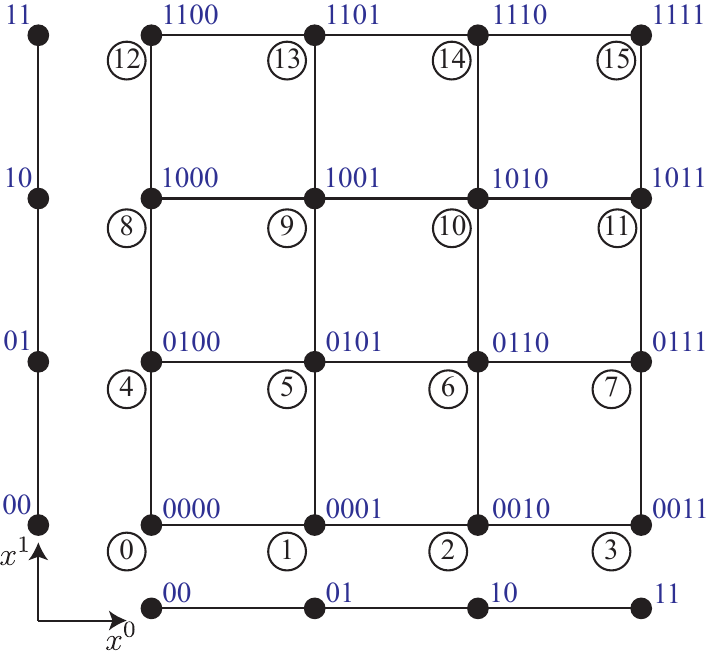}
    \caption{Global node numbering of a two-dimensional mesh induced by the tensor product of two one-dimensional meshes.}
    \label{fig:tensor-product-mesh-numbering}
\end{figure}

For dimensions $d > 1$, elements and meshes are formed as a Cartesian product (see Figure~\ref{fig:tensor-product-mesh-numbering}). The global finite element arrays also inherit a tensor product structure (see Section~\ref{subsec:tensor-products-of-Lagrange-elements}), which we will show can be used to obtain efficient block-encodings. As an example, let us consider a two-dimensional mesh of $\Qc_1(\Omega^{e,2})$ elements. From Eq.~\eqref{eq:K-quantum-summation-order}, the two-dimensional stiffness matrix is assembled as 
\begin{align}
    \Kc^{(2)} &= \sum_{j,k \in [2]^2} K_{jk}^{e,2} \uoi_{jk} \,.
    \label{eq:K-2D-1}
\end{align}
Following the node numbering convention induced by the tensor product (as shown in Figure~\ref{fig:tensor-product-mesh-numbering}), the unit of interaction has the tensor product structure 
\begin{align}
    \uoi_{jk} &= \uoi_{j_1j_0,k_1k_0} = \uoi_{j_1k_1} \otimes \uoi_{j_0k_0} \,,
    \label{eq:uoi-jk-tensor-product-structure}
\end{align}
where we write $j$ and $k$ as the $\nen{}$-ary strings $j = j_1j_0$ and $k = k_1k_0$, with $j_1,j_0,k_1,k_0 \in [\nen{}]$ (note that here, $\nen = p + 1 = 2$). Additionally, the elemental stiffness matrix $K^{e,2}$ inherits a tensor product structure from Eq.~\eqref{eq:Ke-tensor-product-structure}. Equation~\eqref{eq:K-2D-1} can then be factored into a tensor product of one-dimensional arrays as
\begin{align}
    \Kc^{(2)} &= \sum_{j_1,j_0,k_1,k_0 \in [\nen{}]} \left( K^e_{j_1 k_1} M^e_{j_0 k_0} + M^e_{j_1 k_1} K^e_{j_0 k_0} \right) \uoi_{j_1k_1} \otimes \uoi_{j_0k_0}\\
    \begin{split}
    &= \left(\sum_{j_1,k_1 \in [\nen{}]} K^e_{j_1 k_1} \uoi_{j_1k_1}  \right) \otimes \left(\sum_{j_0,k_0 \in [\nen{}]} M^e_{j_0 k_0} \uoi_{j_0k_0}  \right) 
    \\ & \qquad + 
    \left(\sum_{j_1,k_1 \in [\nen{}]} M^e_{j_1 k_1} \uoi_{j_1k_1}  \right) \otimes \left(\sum_{j_0,k_0 \in [\nen{}]} K^e_{j_0 k_0} \uoi_{j_0k_0}  \right)
    \end{split}\\
    &= K \otimes M + M \otimes K \,.
\end{align}
Similarly, one can show using Eq.~\eqref{eq:Me-tensor-product-structure} that the two-dimensional mass matrix $\Mc^{(2)}$ can be written as the tensor product
\begin{align}
    \Mc^{(2)} &= M \otimes M \,.
\end{align}
This procedure further generalizes to $d$-dimensional meshes of order $p$ elements (i.e., $\Qc_p(\Omega^{e,d})$).

\begin{thm}[Block-encoding of the $d$-dimensional Stiffness and Mass Matrices for $\Qc_p(\Omega^{e,d})$ elements]
\label{thm:query-complexity-Assembly-dD-degree-p}
    For a $d$-dimensional Cartesian mesh of $\Qc_p(\Omega^{e,d})$ elements, the stiffness matrix $\Kc^{(d)}$ and mass matrix $\Mc^{(d)}$ are given by 
    \begin{align}
        \Kc^{(d)} 
        &= \sum_{j,k \in [p+1]^d} K_{jk}^{e,d} \uoi_{jk} \label{eq:K-interaction-summation}\\
        &= K \otimes \underbrace{M \otimes \dots \otimes M}_{(d-1)\text{-times}}
        + M \otimes K \otimes  \underbrace{M \otimes \dots \otimes M}_{(d-2)\text{-times}} + \dots + M \otimes \dots \otimes M \otimes K
        \,,
        \label{eq:K-tensor-product-structure}\\
        \Mc^{(d)} 
        &= \sum_{j,k \in [p+1]^d} M_{jk}^{e,d} \uoi_{jk} \label{eq:M-interaction-summation}\\
        &= \underbrace{M \otimes \dots \otimes M}_{d\text{-times}}
        \,.
        \label{eq:M-tensor-product-structure}
    \end{align}
    Define $\alpha_K := \sum_{j,k \in [p+1]} |K_{jk}^e|$. Then there exists a $\left(d \alpha_K, \Oc(d m, \log d) \right)\mathrm{-BE}(\Kc^{(d)})$ that utilizes $\Oc(d^2 p^2 n m)$ Toffoli or simpler gates, and a $\left(1, \Oc(d m, \log d)\right)\mathrm{-BE}(\Mc^{(d)})$
    that utilizes $\Oc(d p^2 n m)$ Toffoli or simpler gates.
\end{thm}

\begin{proof}
    Equations~\eqref{eq:K-tensor-product-structure} and~\eqref{eq:M-tensor-product-structure} follow from the paragraph preceding the theorem. Using these equations and Theorem~\ref{thm:query-complexity-Assembly-1D-degree-p}, the block-encodings for $\Kc^{(d)}$ and $\Mc^{(d)}$ are then implemented as a tensor product of block-encodings (Corollary~\ref{cor:product-of-block-encodings}) and an LCU of block-encodings (Lemma~\ref{lem:linear-combination-of-block-encodings}). The subnormalization cost $\Mc^{(d)}$ is one by virtue of \cref{prop:mass-matrix-subnormalization}.
\end{proof}

The expanded summations in Eqs.~\eqref{eq:K-interaction-summation} and~\eqref{eq:M-interaction-summation} can be interpreted as a sum between all $(p + 1)^{2d}$ interactions of the tensor product element $\Qc_p(\Omega^{e,d})$. A sum of this size, however, scales exponentially with the dimension. Meanwhile, the factored summations (Eqs.~\eqref{eq:K-tensor-product-structure} and~\eqref{eq:M-tensor-product-structure}) contain at most $\Oc(d^2)$ sums of size $\nen{}^2$, which is an exponential improvement.

\section{Assembly via Numerical Integration on the Quantum Computer}
\label{sec:numerical-integration}

In this section, we utilize the tools introduced in Section~\ref{sec:assembly-of-global-arrays-linear} to assemble finite element arrays with non-constant elemental contributions (i.e., PDEs with variable coefficients)---a class of problems that correspond to PDEs of practical interest in engineering and the sciences. For example, in the modified Poisson equation (Eq.~\eqref{eq:modified-Poisson-eqn}), non-constant elemental contributions correspond to the coefficients $D$ and $k$ being non-constant (i.e., the PDE coefficients may be general scalar fields).
In this case, assembly is achieved by performing Gauss-Legendre integration on the quantum computer, which requires the implementation of an additional oracle that specifies the quadrature points in the mesh. We provide an explicit implementation of this oracle for Cartesian meshes, and in so doing, solve the \textit{General Assembly Problem}---namely, the task of assembling global arrays from element-level contributions (see Definition~\ref{def:general-assembly-problem})---for quadrilateral $\Qc_p(\Omega^e)$ elements.
In addition, numerical integration can be used to prepare the force vector $\ket{f}$ (which includes any Neumann boundary conditions) on the right-hand side of the linear system that results after assembly (see Eq.~\eqref{eq:FE-LSP-normalized}) as a quantum state.
We begin by presenting the framework in terms of oracular access to the assembly tools, which in principle shows that the method can be generalized to meshes of more general element types (such as triangular elements).

\subsection{The General Assembly Problem}

We will reserve the formal definition of the \textit{General Assembly Problem} for the end of this section, after we have discussed quadrature on the quantum computer. Briefly, however, the general assembly problem refers to the assembly of a global finite element array from local elemental contributions. In general, each elemental contribution requires the evaluation of an integral over the element. Evaluating these integrals analytically becomes infeasible when the integrand varies across each element (since we have an exponential number of elements over which we would like to assemble).
Moreover, obtaining an analytic expression is unnecessary since quadrature schemes such as Gauss-Legendre quadrature can integrate polynomials exactly, and we approximate each function in the integrand with a polynomial for the purposes of QSP.

For a concrete example of the general assembly problem, consider the modified Poisson equation (Eq.~\eqref{eq:modified-Poisson-eqn})---or any PDE containing a Laplacian term $D\nabla^2$---introduced in Section~\ref{subsec:weak-formulation-modified-Poisson}. 
Recall again (from Eq.~\eqref{eq:K-classical-summation-order}) that the stiffness matrix is assembled by summing together the contributions from each node of each element, i.e.,
\begin{align}
    K 
    = \sum_{J,K \in \numnp} K_{JK} \ketbra{J}{K}
    = \sum_{\el{} \in [\numel{}]} \sum_{j,k \in [\nen{}]} K_{jk}^\el{} \ketbra{\mathrm{IX}(j,\el{})}{\mathrm{IX}(k,\el{})}
    \,,
    \label{eq:stiffness-matrix-defn-2}
\end{align}
where $\IX$ is the connectivity matrix, which maps local degrees of freedom to global degrees of freedom (i.e., $J = \IX(j,\el)$ and $K = \IX(k,\el)$).
Classically, the sum over elements is typically the outer loop, while the sum over the local degrees of freedom is the inner loop. As in Section~\ref{sec:assembly-of-global-arrays-linear},  we will  reverse the order of  summation   
\begin{align}
    K = \sum_{j,k \in [\nen{}]} \sum_{\el{} \in [\numel{}]} K_{jk}^\el{} \ketbra{\mathrm{IX}(j,\el{})}{\mathrm{IX}(k,\el{})} \,.
    \label{eq:stiffness-matrix-defn-3}
\end{align}
Recall that the elemental stiffness matrix is given by (Eq.~\eqref{eq:Kij-e})
\begin{align}
    K_{jk}^\el{} &= \int_{\Omega^{\el{}}} D(x) \nabla N_j(x) \cdot \nabla N_k(x) \,d{\Omega^\el} \,,
\end{align}
where we allow the diffusivity $D \colon \Omega \to \bbR$ to vary spatially. 

More generally, a global finite element array $F$ is assembled from its elemental contributions using the connectivity matrix in the same manner as Eq.~\eqref{eq:stiffness-matrix-defn-2}:
\begin{align}
	F 
	&= \sum_{J,K \in \numnp} F_{JK} \ketbra{J}{K}
	= \sum_{\el{} \in [\numel{}]} \sum_{j,k \in [\nen{}]} F_{jk}^\el{} \ketbra{\mathrm{IX}(j,\el{})}{\mathrm{IX}(k,\el{})}
	\,,
	\label{eq:global-array-F}
\end{align}
where the elemental array $F^\el$ is given more generally by a bilinear form $b(\bullet,\bullet) \colon \Uc \times \Uc \to \bbR$ over the space of admissible solutions $\Uc$ that emerges from the weak formulation of the PDE (see Appendix~\ref{sec:analysis-of-FEM} for more details), i.e.,
\begin{align}
	F^\el_{jk} &= b[N_j^\el,N_k^\el]
	\,.
\end{align}
For example, for the stiffness matrix $K$, we have that $b[N_j^\el,N_k^\el] = \int_{\Omega} D(x) \nabla N_j^\el(x) \cdot \nabla N_k^\el(x) \,d{\Omega}$, whereas for the mass matrix $M$ we have $b[N_j^\el,N_k^\el] = \int_{\Omega^\el} k(x) N_j(x) N_k(x) \,d{\Omega^\el}$. In general, we consider bilinear forms that can be represented as 
\begin{align}
	b[N_j^\el,N_k^\el] &= \int_{\Omega^\el} f(x) B[ N_j^\el,N_k^\el](x) \,d{\Omega^\el}
	\,,
	\label{eq:bilinear-form-representation}
\end{align}
where $f \colon \Omega \to \bbR$ is a scalar function over the entire mesh, and $B \colon \Uc \times \Uc \to L^1(\Omega)$ is a pointwise bilinear map\footnote{Since the elemental basis functions have support over their element, the domain of integration in Eq.~\eqref{eq:bilinear-form-representation} can be reduced to just the element, i.e., $\int_\Omega (\cdots) \,d\Omega = \int_{\Omega^\el} (\cdots) \,d{\Omega^\el}$}. For example, for the bilinear form associated with the stiffness matrix $K$ we have that $B[N_j^\el,N_k^\el] = \nabla N_j^\el(x) \cdot \nabla N_k^\el(x)$, whereas for the mass matrix $M$ we have that $B[N_j^\el,N_k^\el] = N_j(x) N_k(x)$.
The integral in Eq.~\eqref{eq:bilinear-form-representation} is typically evaluated using numerical integration in the finite element method.

Numerical integration over an element $\Omega^{\el{}}$ involves choosing a set of $G$ quadrature points and weights $\{(x_k^{\el{}},w_k^{\el{}})\}_{k \in [G]}$ with $x_k \in \Omega^{\el{}}$ and $w_k \in \bbR$ such that for any function $f$ defined over $\Omega^{\el{}}$, 
\begin{align}
    \int_{\Omega^\el{}} f \,d\Omega^{\el{}} 
    \approx \sum_{\ell \in [G^{\el}]} w_\ell^{\el{}} f(x_\ell^{\el{}})
    \,.
\end{align}
For (global) integrals of a function $f$ over the entire domain, we have that 
\begin{align}
    \int_{\Omega} f \,d\Omega &= \sum_{\el{} \in [\numel{}]} \int_{\Omega^\el{}} f \,d\Omega^{\el{}}
    \approx \sum_{\el{} \in [\numel{}]} \sum_{\ell \in [G^{\el}]} w_\ell^{\el{}} f(x_\ell^{\el{}})
    \,.
\end{align}
Combining this with Eqs.~\eqref{eq:global-array-F} and~\eqref{eq:bilinear-form-representation}, we can now define the general assembly problem.
\begin{defn}[General Assembly Problem]
\label{def:general-assembly-problem}
	Let $\IX{}$ be the connectivity matrix for a mesh $\Omega$ with elemental basis functions $\ket{N^e}$.  
    Consider the global finite element array $\displaystyle F \in \bbR^{\textnormal{\numnp{}} \times \textnormal{\numnp{}}}$ assembled from the bilinear form $b \colon \Uc \times \Uc \to \bbR$, where $\Uc$ is the space of admissible solutions (i.e., $F$ is assembled from the elemental contributions $F_{jk}^\el = b[N_j^\el,N_k^\el]$). Suppose also that the bilinear form $b$ can be represented using the functions $f \colon \Omega \to \bbR$ and $B \colon \Uc \times \Uc \to L^1(\Omega)$ as in Eq.~\eqref{eq:bilinear-form-representation}. The general assembly problem is to evaluate the sum 
    \begin{align}
    	F = \sum_{\el{} \in [\numel{}]} 
    	\sum_{j,k \in [\nen{}]} \sum_{\ell \in [G^\el]} w_\ell^{\el{}} B\left[
    	N_j^\el{}, N_k^\el{}
    	\right](x_\ell^{\el{}}) f(x_\ell^{\el{}}) \ketbra{\mathrm{IX}(j,\el{})}{\mathrm{IX}(k,\el{})} \,.
    	\label{eq:general-assembly-problem}
    \end{align}
\end{defn}

For structured meshes with uniform element sizes, the basis function evaluations at the quadrature points $B\left[N_j^\el{}, N_k^\el{}\right](x_\ell^{\el{}})$ and the quadrature weights $w_\ell^{\el{}}$ are typically independent of the element $\el \in [\numel]$ (assuming the same quadrature scheme is used for each element). In this case, just as was done in Section~\ref{sec:assembly-of-global-arrays-linear}, we may rearrange Eq.~\eqref{eq:general-assembly-problem} so that the sum over all the elements is the inner-most summation, and factor out the constant function evaluations as
\begin{align}
    F = 
    \sum_{j,k \in [\nen{}]} 
    \sum_{\ell \in [G]} 
    \underbrace{
    w_\ell^{\el{}} B\left[
    N_j^\el{}, N_k^\el{}
    \right](x_\ell^{\el{}})
	}_{=: c_{\ell j k}}
    \left(
    \sum_{\el{} \in [\numel{}]} 
    f(x_\ell^{\el{}}) \ketbra{\mathrm{IX}(j,\el{})}{\mathrm{IX}(k,\el{})}
    \right)
    \,.
    \label{eq:general-assembly-problem-2}
\end{align}
The constants $c_{\ell j k}$ can be precomputed classically using a constant $G \cdot \nen^2$ number of basis function evaluations, and accessed through an oracle. 
We assume that $f(x)$ is well-approximated to some desired accuracy $\epsilon \ge 0$ by a polynomial $q(x)$ (determined by classical means) so that we can implement $f$ using the Multivariate Quantum Eigenvalue Transformation (MQET)~\cite{borns2023MQET} technique (summarized in Appendix~\ref{sec:MQET}). 

Now, for each quadrature point $x_\ell^\el \in \bbR^d$, write $x_\ell^\el = (x_\ell^{\el,0}, \cdots, x_\ell^{\el,d-1})$. Given access to a set of position operators for the quadrature points that satisfy 
\begin{align}
	X_\ell^{(i)} \ket{\el} &= x_\ell^{\ell,i} \ket{\el} \,,
\end{align}
we have that $\Xb_\ell := \{X_\ell^{(i)}\}_{\ell \in [d]}$ is a family of pairwise commuting, diagonalizable operators (see Section~\ref{sec:position-operators-and-Cartesian-coordinates}). Using MQET, this set of operators can be transformed to a block-encoding of the operator 
\begin{align}
	f(\Xb_\ell) &:= \sum_{\el \in \numel} f(x_\ell^\el) \ketbra{\el}{\el}
	\,.
\end{align}
Recalling the local-to-global node number indicator matrix $\uoi_j := \sum_{\el \in \numel} \ketbra{\IX(j,\el)}{\el}$ (See Definition~\ref{defn:local-to-global-node-number-indicator-matrix}), Eq.~\eqref{eq:general-assembly-problem-2} can be written as 
\begin{align}
	F &= \sum_{j,k \in [\nen{}]} 
	\uoi_j \left( 
	\sum_{\ell \in [G]} 
	c_{\ell j k} f(\Xb_\ell)
	\right) \left( \uoi_k \right)^\dagger
	\,.
	\label{eq:solved-general-assembly-problem}
\end{align}
Equation~\eqref{eq:solved-general-assembly-problem} solves the general assembly problem for structured meshes with uniform element sizes using oracular access to the quadrature position operators $\Xb_\ell$. In the following subsections, we provide an explicit construction for these operators for meshes of Lagrange tensor product elements using the Gauss-Legendre quadrature scheme.

\subsection{Gauss-Legendre Quadrature for One-Dimensional Elements}
\label{subsec:quadrature-1D}

One popular choice of quadrature scheme in the finite elements literature is Gauss-Legendre Quadrature of order $G$, which integrates degree $2G-1$ polynomials exactly~\cite{papadopoulos2015280a,chapelle2011finite}. 
Here, we also choose Gauss-Legendre quadrature because of its Cartesian product structure in higher dimensions---a property that will once again lead to the efficient evaluation of element integrals and global assembly on the quantum computer. Additionally, the ability to integrate polynomials exactly will also work well with QSP/MQET, which require polynomial approximations to functions. Since the order of these polynomials is known ahead of time, this property can be used to ensure that the numerical integration error is negligible.

Let $P_G(x)$ denote the degree $G-1$ Legendre polynomial, normalized such that $P_G(1) = 1$. The Gauss points $x_k$ are given by the roots of $P_G(x)$ (of which there must be exactly $G$), while the weights $w_k$ are given by~\cite{abramowitz1948handbook}
\begin{align}
    w_k &= \frac{2}{(1 - x_k)^2 [P_G'(x_k)]^2} \,.
\end{align}
Gauss-Legendre quadrature then approximates a one-dimensional integral over $[-1,1]$ as
\begin{align}
    \int_{[-1,1]} f(x) \,dx &\approx \sum_{k \in [G]} w_k f(x_k) \,.
    \label{eq:1D-quadrature-sum}
\end{align}
If instead we wish to integrate over an interval $\Omega^e = [a,b]$, we can simply pull the integral back to the interval $[-1,1]$ and find that 
\begin{align}
    \int_{[a,b]} f(x) \,dx &= \int_{[-1,1]} f\left( \frac{b-a}{2}t + \frac{a+b}{2} \right) \frac{b-a}{2} \,dt \,,\\
    &\approx \sum_{k \in [G]} \frac{b-a}{2}w_k \cdot f\left( \frac{b-a}{2}t + \frac{a+b}{2} \right) \,.
\end{align}
This means that the Gauss points and weights from the ``standard'' interval $[-1,1]$ map to the interval $[a,b]$ as 
\begin{align}
    \Tilde{x}_k &= \frac{b-a}{2}x_k + \frac{a+b}{2} \,,\\
    \Tilde{w}_k &= \frac{b-a}{2}w_k \,.
\end{align}
In particular, consider a mesh $\Omega = [0,1]$ of one-dimensional Lagrange elements $\Qc_p(\Omega^e)$ with $\numnp{} = 2^n$ and $\nen{} = 2^m$. Then the width of any given element $\el{} \in [\numel{}]$ is $h := 1/\numel{}$. Each element spans the interval $[h \cdot \el{}, h \cdot (\el{} + 1)]$ so that the elemental quadrature points and weights are 
\begin{align}
    x_k^{\el{}} &= \frac{h}{2}x_k + \frac{h(2\el{} + 1
    )}{2} = h\left(\frac{x_k + 1}{2} + \el{} \right) \,, \label{eq:elemental-Gauss-points-1D}\\
    w_k^{\el{}} &= \frac{h}{2}w_k \,,
    \label{eq:elemental-Gauss-weights-1D}
\end{align}
where $k \in [G^\el{}]$.  
For simplicity, assume that the same number of Gauss points $G$ are used for each element, and that $G = 2^g$. 
To efficiently perform Gauss-Legendre quadrature over all elements simultaneously, we will prepare the $k^\text{th}$ Gauss point and weight $(x_k^{\el{}},w_k^{\el{}})$ over all elements simultaneously, and then perform an LCU over all $G$ terms in the quadrature sum (Eq.~\eqref{eq:1D-quadrature-sum}).

Since the quadrature points and weights are tabulated classically, we can access a state preparation pair $(G_L,G_R)$ for the Gauss weights, as well as a state preparation oracle for the standard Gauss points satisfying
\begin{align}
    G_X \ket{\ell}_{g} &= x_\ell \ket{\ell}_{g} \,,
\end{align}
where $\ell \in [G]$.
Consider the general assembly problem (Eq.~\eqref{eq:general-assembly-problem}), represented here as 
\begin{align}
    F \approx 
    \sum_{j,k \in [\nen{}]} \sum_{\ell \in [G]} \sum_{\el{} \in [\numel{}]} w_\ell^{\el{}} B_{jk\ell}^\el f(x_\ell^{\el{}}) \ketbra{\mathrm{IX}(j,\el{})}{\mathrm{IX}(k,\el{})} \,,
    \label{eq:general-assembly-problem-3}
\end{align}
where the constants $B_{jk\ell} \in \bbR$ are given by 
\begin{align}
	B_{jk\ell}^\el := B\left[ N_j^\el{}, N_k^\el{} \right](x_\ell^{\el{}})
	\,.
\end{align}
Here, $B$ is the bilinear form that $F$ is assembled from, and $\ket{N^\el}$ are the elemental basis functions.
The mapping between the coordinates within any element $\el \in [\numel]$ from the interval $[h \cdot \el{}, h \cdot (\el{} + 1)]$ to the interval $[-1,1]$ (i.e., the ``pull-back'' map to the standard element) also takes the form given by Eq.~\eqref{eq:elemental-Gauss-points-1D} (a concept known as \textit{isoparametric mapping}~\cite{papadopoulos2015280a}). This implies that the basis function evaluations at each Gauss point are constant, so we will drop the superscript $(\bullet)^\el$, and simply write $B_{jk\ell} \equiv B_{jk\ell}^\el$.
We will build up $F$ starting with the position operator $X_{\log\numel{}} := \sum_{\el{} \in [\numel{}]} \el{} \ketbra{\el{}}{\el{}}$. Note that the end result must take the form of Eq.~\eqref{eq:solved-general-assembly-problem}, but we will start with Eq.~\eqref{eq:general-assembly-problem-3} in the construction.

\begin{rem}[$\sum_{\el{} \in [\numel{}]} \ketbra{\el{}}{\el{}}$ vs $\sum_{\np{} \in [\numnp{}]} \ketbra{\np{}}{\np{}}$]
    We are working in an $n$-qubit computational basis, with $\numnp{} = 2^n$. Recall that for $p$th order one-dimensional Lagrange elements, we have that $\numnp{} = (\numel{}) p + 1$ (see Eq.~\eqref{eq:numnp-relation}). This means that the identity on $n$ qubits satisfies $I_{\log\numnp{}} = \sum_{\np{} \in [\numnp{}]} \ketbra{\np{}}{\np{}}$, which is different from $\Pi_{\log\numel{}} := \sum_{\el{} \in [\numel{}]} \ketbra{\el{}}{\el{}}$. The latter is an orthogonal projector, and can be thought of as a $(1,\log p)$-block-encoding of the $I_{\log\numel{}}$ matrix. To obtain the position operator $X_{\log\numel{}}$, one can start with a block-encoding of the $n$-qubit position operator $X_n$ (Eq.~\eqref{eq:position-operator}), and then multiply by the orthogonal projector $\Pi_{\log\numel{}}$.
\end{rem}

By observing the form of Eq.~\eqref{eq:elemental-Gauss-points-1D}, we can calculate the position operator $X_\ell^G \in \bbR^{\numnp{} \times \numnp{}}$ that encodes the $\ell^\text{th}$ gauss point $x_\ell^{\el{}}$ of each element along its diagonal as 
\begin{align}
    X_\ell^G := \sum_{\el{} \in [\numel{}]} x_\ell^{\el{}} \ketbra{\el{}}{\el{}} 
    = h\left(\frac{x_\ell + 1}{2}\right) \Pi_{\log\numel{}} + h X_{\log\numel{}}
    = h \Pi_{\log\numel{}} \left[ \left(\frac{x_\ell + 1}{2}\right) I_n + X_{n} \right]
    \,.
    \label{eq:Gauss-point-position-operator}
\end{align}
Section~\ref{sec:position-operators-and-Cartesian-coordinates} gives an $(N-1, \log n)$-block-encoding of $X_n$. We can absorb the subnormalization factor $N-1$ into an LCU with the identity matrix $I_n$, and then multiply by a $(1,1)$-block-encoding $U_{\Pi_{\log\numel{}}}$ of the projector $\Pi_{\log\numel{}}$ using Lemma~\ref{lem:product-of-block-encodings} to obtain 
\begin{align}
    U_{\Pi_{\log\numel{}}} \cdot \LCU\left( (I_n, U_{X_{\log\numel{}}}), \begin{bmatrix}
        h(x_\ell + 1)/2 \\ h(N-1)
    \end{bmatrix} \right) \in (p + h|x_\ell + 1|/2, \log(n) + 2)\mathrm{-BE}(X_\ell^G) \,.
\end{align}
Observe that the subnormalization factor is upper-bounded by $p + h$ and lower-bounded by $p$, bounds which are independent of the particular Gauss point $\ell \in [G]$. Using a suitable polynomial approximation to $\dfrac{f}{\| f \|_{L^\infty(\Omega)}}$, we can use QSP/MQET to transform this block-encoding of Eq.~\eqref{eq:Gauss-point-position-operator} to a block-encoding of 
\begin{align}
    f(X_\ell^G) := \sum_{\el{} \in [\numel{}]} f(x_\ell^{\el{}}) \ketbra{\el{}}{\el{}}
    \,.
    \label{eq:Gauss-point-position-operator-2}
\end{align}
Using two applications of a block-encoding of the local-to-global node number indicator matrix $\uoi_{i} := \sum_{\el{} \in [\numel{}]} \ketbra{\mathrm{IX}(j,\el{})}{\el{}}$ (see Definition~\ref{defn:local-to-global-node-number-indicator-matrix}), we can map \cref{eq:Gauss-point-position-operator-2} to 
\begin{align}
    \uoi_{j} f(X_\ell^G) (\uoi_{k})^\dag := \sum_{\el{} \in [\numel{}]} f(x_\ell^{\el{}}) \ketbra{\mathrm{IX}(j,\el{})}{\mathrm{IX}(k,\el{})}
    \,.
    \label{eq:Gauss-point-position-operator-3}
\end{align}
For Lagrange elements (of any order $p$), a similar relation to Eq.~\eqref{eq:Ajk-relation-to-A00} holds in that 
\begin{align}
	\uoi_j &= S^j \uoi_0 
	\,.
	\label{eq:Aj-relation-to-A0}
\end{align}
Additionally, similar to Eq.~\eqref{eq:A00-p-general-form}, the first local-to-global map $\uoi_0$ has the block form 
\begin{align}
	\uoi_0 &= \begin{bmatrix}
		\ket{0}_n & \ket{p}_n & \cdots & \ket{(\numel - 1)p}_n
		& 0_{\numnp \times (\numnp - \numel)}
	\end{bmatrix}
	\,,
	\label{eq:A0-p-Lagrange-element-form}
\end{align}
which also requires a quantum circuit for classical arithmetic. Here, we let $U_{(\bullet \cdot p) \% N}$ be any unitary that performs the operation 
\begin{align}
	U_{(\bullet \cdot p) \% N} \ket{\el} &= \ket{(\el \cdot p) \mod N}
	\,,
\end{align}
where $\el \in [\numel]$. This quantum circuit can be constructed using a quantum multiplication circuit to form the product $e \cdot p$, followed by a quantum division circuit to store the result $(e \cdot p) \mod N$ at the cost of $\Oc(n m)$ Toffoli gates~\cite{wang2024comprehensive}.
Recall that Eq.~\eqref{eq:Aj-general-IX-BE} gives a $(1,1)\mathrm{-BE}(\uoi_j)$. Defining, 
\begin{align}
	\Oc_j &:= \left(\bra{j} \otimes I_n \right) \Oc_\IX \left(\ket{j} \otimes I_n\right)
	\,,
\end{align}
it follows that for $p^\text{th}$ order Lagrange elements we have $\Oc_j = S^j U_{(\bullet \cdot p) \% N} $. Equation~\eqref{eq:Aj-general-IX-BE} then block-encodes $\uoi_j$ using $\Oc(n m)$ Toffoli or simpler gates.

\begin{rem}
	To relate the matrix $\uoi_j$ to the unit of interaction $\uoi_{jk}$, recall from Eq.~\eqref{eq:Ajk-factoring} that $\uoi_{jk} = \uoi_j \left(\uoi_k\right)^\dagger$, so that by Eq.~\eqref{eq:Aj-relation-to-A0} we have $\uoi_{jk} = S^j \uoi_0 \left(\uoi_0\right)^\dagger S^{-k}$. Note, however, $\uoi_0\left(\uoi_0\right)^\dagger$ is not equal the identity matrix. Additionally, while $\uoi_0$ is not  Hermitian for $p > 1$, for $p = 1$ we have $\uoi_0 = \uoi_{00} = \left(\uoi_0\right)^\dagger$. Observe also that Eq.~\eqref{eq:A0-p-Lagrange-element-form} verifies the product $\uoi_0\left(\uoi_0\right)^\dagger = \sum_{e \in [\numel]} \ketbra{\el \cdot p}{\el \cdot p}$, which is precisely Eq.~\eqref{eq:A00-p-general-form}.
\end{rem}

To complete the assembly, we take an LCU over the Gauss points $\ell \in [G]$ and the local node numbers $j,k \in [\nen{}]$ with the LCU vector $\frac{h}{2}w_\ell B_{jk\ell}$, yielding the assembled global array
\begin{align}
    F 
    &\approx \sum_{j,k \in [\nen{}]} \uoi_{j} \left( \sum_{\ell \in [G]} \frac{h}{2}w_\ell B_{jk\ell} f(X_\ell^G) \right) (\uoi_{k})^\dag\\
    &= 
    \sum_{j,k \in [\nen{}]} \sum_{\ell \in [G]} 
    \frac{h}{2}w_\ell B_{jk\ell} \left(
    \sum_{\el{} \in [\numel{}]}
    f(x_\ell^{\el{}}) \ketbra{\mathrm{IX}(j,\el{})}{\mathrm{IX}(k,\el{})}
    \right) \,.
    \label{eq:general-assembly-problem-4}
\end{align}
We summarize and analyze the cost of implementing this LCU in the following theorem.

\begin{thm}[Block-encoding costs for assembly via numerical integration in one dimension]
\label{thm:assembly-via-numerical-integration-1D}
	Consider a one-dimensional mesh $\Omega = [0,1]$ of $p^\text{th}$ order Lagrange elements $\Qc_p(\Omega^e)$. Let $F$ be a global finite element array assembled from a bilinear form $b$ represented in the form given by Eq.~\eqref{eq:bilinear-form-representation} (i.e., $b(\bullet,\bullet) = \int_\Omega B[\bullet,\bullet](x) f(x) \,dx$). 
	Furthermore, suppose that the basis function evaluations $B\left[ N_j^\el{}, N_k^\el{} \right](x_\ell^{\el{}})$ are constant, and that $f$ is well-approximated by a polynomial of degree $D$.
	Then there exists an $(\alpha,a)$-block-encoding of $F$ that incurs a cost of 
    $\Oc\left( n \cdot (D G + p^2 m) \right)$ Toffoli or simpler gates, where 
	\begin{align}
		\alpha &\le \| f \|_{L^\infty(\Omega)} \cdot \left( \sum_{j,k \in [\nen]} \int_{\Omega^e} \left| B\left[ N_j^e,N_k^e \right](x) \right| \,dx \right)
		\,,
		\label{eq:numerical-integration-subnormalization}
	\end{align}
	and 
	\begin{align}
		a &= \lceil \log(n) \rceil + \lceil \log(G) \rceil + \lceil 2 \log(p+1) \rceil  + 3
		\,.
		\label{eq:numerical-integration-ancilla}
	\end{align}
\end{thm}

\begin{proof}
	The steps preceding the theorem give the construction of the block-encoding of $F$. It remains only to cost this block-encoding out.
	
	Recall that Lemma~\ref{lem:position-operator} gives an $(N-1,\log(n))$-block-encoding for the position operator $X_n$.
	Equation~\eqref{eq:Gauss-point-position-operator} defines an LCU for the Gauss point position operator $X_\ell^G$. 
	In particular, the LCU vector for the sum $ h \left(\frac{x_\ell + 1}{2}\right) I_n + h (N-1) \left(\frac{1}{N-1} X_{n}\right)$ has one-norm 
	\begin{align}
		\left\| \begin{bmatrix} 
			h \left(\frac{x_\ell + 1}{2}\right)\\ 
			h (N-1)  
		\end{bmatrix} \right\|_1 &= h \left(\frac{x_\ell + 1}{2}\right) + p \le h + p 
		\,,
	\end{align}
	where we have used $h (N-1) = \frac{\numnp - 1}{\numel} = \nen - 1 = p$ (by Eq.~\eqref{eq:numnp-relation}). By Lemma~\ref{lem:linear-combination-of-block-encodings}, we then have a 
	\begin{align}
		\left(h \left(\frac{x_\ell + 1}{2}\right) + p, \log(n) + 2 \right)\mathrm{-BE}(X_\ell^G)
		\,.
	\end{align}
	Using QSP, this can be transformed to a block-encoding of $\dfrac{f(X_\ell^G)}{\| f \|_{L^\infty(\Omega)}}$ at the cost of one additional ancilla qubit and $\Oc(D)$ queries to the block-encoding of $X_\ell^G$. Additionally, since $X_n$ utilizes $O(n)$ Toffoli or simpler gates, as does this block-encoding of $f(X_\ell^G)$.
	Next, we form Eq.~\eqref{eq:general-assembly-problem-4} using the LCU vector $\frac{h}{2} w_\ell B_{jk\ell}$. The $\ell$-summation can be rewritten as an integral as 
	\begin{align}
		\sum_{\ell \in [G]} \left| \frac{h}{2} w_\ell B_{jk\ell} \right|
		&= \sum_{\ell \in [G]} \frac{h}{2} w_\ell \left| B_{jk\ell} \right|
		= \int_{\Omega^e} \left| B(N_j^\el,N_k^\el)(x)  \right| \,dx
		\,.
	\end{align}
	Using Lemma~\ref{lem:linear-combination-of-block-encodings} again for the LCU of size $G$, and for the LCU of size $\nen^2 = (p+1)^2$ with the $(1,1)$-block-encodings of $\uoi_j$ that use $\Oc(n m)$ Toffoli gates each, we obtain a block-encoding of $F$ with subnormalization given by \cref{eq:numerical-integration-subnormalization}, utilizing the number of ancilla given by \cref{eq:numerical-integration-ancilla}.
\end{proof}

\subsection{Gauss-Legendre Quadrature for Tensor Product Elements}
\label{subsec:quadrature-dD}

Quadrature for elements in higher dimensions that are built out of a tensor product structure generalizes naturally from quadrature in one dimension, as the set of Gauss points and weights also inherit a tensor product structure. We consider meshes of ~$\Qc_p(\Omega^{e,d})$ Lagrange tensor product elements. Let~$\{(x_\ell,w_\ell)\}_{\ell \in [G]}$ be the set of Gauss points and weights in one-dimension. In $d$ dimensions, we denote the Gauss points and weights induced by the set of one-dimensional quadrature points as~$\{(\xb_\ellb,w_\ellb)\}_{\ellb \in [G]^d}$ where $\ellb = (\ell_0,\dots,\ell_{d-1})$ is a multi-index with $\ell_i \in [G]$ for all $i \in [d]$. The points and weights are given by 
\begin{align}
    \{\xb_\ellb\}_{\ellb \in [G]^d} &= 
    \{x_\ell\}_{\ell_0 \in [G]} 
    \times \cdots \times 
    \{x_\ell\}_{\ell_{d-1} \in [G]}
    \,,\\
    \{w_\ellb\}_{\ellb \in [G]^d} &= 
    \{w_{\ell_0} \times \cdots \times w_{\ell_{d-1}} \}_{\ellb \in [G]^d}
    \,,
\end{align}
where $\xb_\ell \in \bbR^d$.
For each element, there are a total of $G^d$ quadrature points. For each $\ellb \in [G]^d$, this translates to a family of $d$ pairwise commuting Gauss point position operators $\Xb^G_\ellb := \{ X^{G,(i)}_{\ell_i} \}_{i \in [d]}$ defined by 
\begin{align}
    X^{G,(i)}_{\ell_i} &:= \underbrace{ \Pi_{\log\numel{}} \otimes\cdots\otimes \Pi_{\log\numel{}} }_{(d-i-1)\text{-times}}
    \otimes X_{\ell_i}^G \otimes 
    \underbrace{ \Pi_{\log\numel{}} \otimes \cdots \otimes \Pi_{\log\numel{}} }_{i\text{-times}}
    \,,\\
    &= h \Pi_{\log\numel{}}^{\otimes d} \left[ 
    \underbrace{ I_n \otimes\cdots\otimes I_n }_{(d-i-1)\text{-times}}
    \otimes \left(\frac{x_\ell + 1}{2} \cdot I_n + X_{n}\right) \otimes 
    \underbrace{ I_n \otimes \cdots \otimes I_n }_{i\text{-times}}
    \right] & \text{(by Eq.~\eqref{eq:Gauss-point-position-operator})}
    \,.
\end{align}
We write the element number $\el{}$ in the $d$-dimensional mesh $\Omega^{d}$ as a $\numel{}$-ary string of length $d$, $\el{} = \el{}_{d-1} \cdots \el{}_0$ (see Figure~\ref{fig:tensor-product-mesh-numbering} for the numbering convention of the tensor product mesh induced by the numbering of a one-dimensional mesh). Then the Gauss point position operators satisfy 
\begin{align}
    X^{G,(i)}_{\ell_i} \ket{\el{}} &= x_{\ell_i}^{\el{}_i} \ket{\el{}} \,.
\end{align}

Using MQET~\cite{borns2023MQET}, we can evaluate $f(\Xb^G_\ellb) = \sum_{\el{} \in [\numel{}]} f(\xb_{\ellb}^{\el{}}) \ketbra{\el{}}{\el{}}$, and then proceed in the same manner as Section~\ref{subsec:quadrature-1D} to assemble the global array $F$. Namely, we multiply by two applications of the local-to-global node number indicator matrix to form the product $\uoi_{j} f(\Xb_{\ellb}^G) (\uoi_{k})^\dag$. We then take the LCU over all Gauss points $\ellb \in [G]^d$ with the product of the Gauss weights $w_\ellb$ and the basis function evaluations $B_{jk\ellb}$ as the LCU vector, and another LCU over all local element node numbers $j,k \in [\nen{}]$ with the equal superposition state as the LCU vector to form 
\begin{align}
    F \approx \sum_{j,k \in [\nen{}]} \sum_{\ellb \in [G]^d} 
    w_\ellb^{\el{}} B_{jk\ellb}
    \left(
    \sum_{\el{} \in [\numel{}]}  f(\xb_{\ellb}^{\el{}}) \ketbra{\mathrm{IX}(j,\el{})}{\mathrm{IX}(k,\el{})}
    \right)
    \,.
\end{align}
The cost for the $d$-dimensional analog of Theorem~\ref{thm:assembly-via-numerical-integration-1D} is summarized below. 
We also give a detailed example of this assembly procedure in Example~\ref{ex:numerical-integration-of-mass-matrix}.

\begin{thm}[Block-encoding costs for Assembly via Numerical Integration in $d$ dimensions]
\label{thm:assembly-via-numerical-integration-d-dims}
	Consider a $d$-dimensional mesh $\Omega^d = [0,1]^d$ of $p^\text{th}$ order Lagrange elements $\Qc_p(\Omega^{e,d})$. Let $F$ be a global finite element array assembled from a bilinear form $b$ represented in the form given by Eq.~\eqref{eq:bilinear-form-representation} (i.e., $b(\bullet,\bullet) = \int_{\Omega^d} B[\bullet,\bullet](\xb) f(\xb) \,d\xb$). 
	Furthermore, suppose that the basis function evaluations $B\left[ N_j^\el{}, N_k^\el{} \right](\xb_\ellb^{\el{}})$ are constant, and that $f$ is well-approximated by a multivariate polynomial of degree less than $D$ in each variable.
	Then there exists an $(\alpha,a)$-block-encoding of $F$ that incurs a cost of 
    $ \Oc\left( dn \cdot (D^d G^d + p^{2d} m) \right) $
    Toffoli or simpler gates with 
	\begin{align}
		\alpha &\le (D + 2)^{d-1} \cdot \| f \|_{L^\infty(\Omega^d)} \cdot \left( \sum_{j,k \in [p+1]^d} \int_{\Omega^{e,d}} \left| B\left[ N_j^e,N_k^e \right](\xb) \right| \,d\xb \right)
		\,,
	\end{align}
	and 
	\begin{align}
		a &= d(\lceil \log(n) \rceil + \lceil \log(G) \rceil + \lceil 2 \log(p+1) \rceil + 3)
		\,.
	\end{align}
\end{thm}

\begin{proof}
	Applying the MQET result in Eq.~\eqref{eq:MQET-LCU} to the construction preceding the theorem and utilizing the tensor product structure of Lagrange elements (see Eq.~\eqref{eq:uoi-jk-tensor-product-structure}) in the proof of Theorem~\ref{thm:assembly-via-numerical-integration-1D} gives the desired result.
\end{proof}

\begin{examplebox}[]{On-the-fly Assembly of the Mass Matrix for Linear Problems}
\label{ex:numerical-integration-of-mass-matrix}

Consider a $d$-dimensional mesh $\Omega^d$ consisting of uniformly spaced $p^\text{th}$ order $d$-dimensional Lagrange tensor product element $\Qc_p([0,h]^d)$ (with volume $h^d$), and the mass matrix $M$ arising from Eq.~\eqref{eq:M-global-integral-defn}. We have for each element $e \in [\numel{}]$ that 
\begin{align}
    M_{ij}^e &= \int_{\Omega^{e,d}} N_i(\xb) \rho(\xb) N_j(\xb) \,d\xb \,,
    \label{eq:Mije-integral}
\end{align}
where $\xb \in \bbR^d$ and $i,j \in [\nen{}]$. The local indices are related to the global indices by $I := \IX{}(i,e)$ and $J := \IX{}(j,e)$. The assembly of a (global) finite element array from elemental contributions can be concisely represented using an assembly operator~\cite{papadopoulos2015280a} as 
\begin{align}
    [M_{IJ}] &= \bA_e [M^e_{ij}] \,,
\end{align}
where $M \equiv [M_{IJ}]$ with the square brackets here denoting a matrix (instead of a set that an index is running over). Implementing this assembly operator (which loops over elements) is efficient for classical computers, but our quantum algorithm instead loops through the local node numbers (see Eq.~\eqref{eq:general-assembly-problem-2}).\\

Before assembling, let's pull the integral in Eq.~\eqref{eq:Mije-integral} back to a standard $\Qc_p(\Omega_0^{e,d})$ element with domain $\Omega_0^{e,d} = [-1,1]^d$. Equation~\eqref{eq:Mije-integral} becomes 
\begin{align}
    M_{ij}^e &= \int_{\Omega^{e,d}_0} N_i(\xib) \rho(\xb(\xib)) N_j(\xib) J^{e,d} \,d\xib\\
    &\approx \sum_{{\ellb} \in [G]^d} w_{\ellb} N_i(\xib_{\ellb}) \rho(\xb(\xib_{\ellb})) N_j(\xib_{\ellb}) J^{e,d}
    \,,
\end{align}
where $(w_{\ellb},\xib_{\ellb})_{{\ellb} \in [G]^d}$ are the Gauss weights and points in the standard $[0,1]^d$ domain, and $J^{e,d}$ is the Jacobian of the pull-back. Here, we have the analytical result for the Jacobian $J^{e,d} = (h/2)^d$. Observe that the evaluations of the basis functions at the Gauss points $N_i(\xib_{\ellb})$ are actually constant throughout the assembly process for each element, and so they may be precomputed and accessed by an oracle. For $\ell \in [G]^d$, let 
\begin{align}
    c_{ij\ellb} &:= (h/2)^d \cdot w_{\ellb} N_i(\xib_{\ellb}) N_j(\xib_{\ellb}) \,,
\end{align}
so that the elemental mass matrix can be computed by $M_{ij}^e = \sum_{\ell \in [G]^d} c_{ij\ellb} \rho(\xb(\xib_{\ellb}))$. There is a constant $G^d \cdot (p+1)^2$ number of these coefficients to compute, so we assume access to a state preparation pair $(P_{L},R_{R})$ for the vector of coefficients $c_{ij\ellb}$.\\

Finally, the global mass matrix can be assembled using the methods of Section~\ref{subsec:quadrature-dD}. That is, beginning with the family of pairwise commuting Gauss point position operators $\Xb^G_\ellb := \{ X^{G,(i)}_{\ell_i} \}_{i \in [d]}$ and a polynomial approximation $q(\xb)$ to $\rho(\xb)$, we can use MQET (reviewed in Section~\ref{sec:MQET}) to prepare 
\begin{align}
    \rho(\Xb^G_\ellb) 
    \approx q(\Xb^G_\ellb) 
    &= \sum_{\el{} \in [\numel{}]} q(\xb_{\ellb}^{\el{}}) \ketbra{\el{}}{\el{}}\\
    \Rightarrow \bA_e [M^e_{ij}] 
    &= \sum_{j,k \in [\nen{}]} \sum_{\ellb \in [G]^d} c_{ij\ellb} \underbrace{ \left( \sum_{\el{} \in [\numel{}]} q(\xb_{\ellb}^{\el{}}) \ketbra{\mathrm{IX}(j,\el{})}{\mathrm{IX}(k,\el{})} \right)}_{=: \Ab_{jk\ellb}}
    \,.
    \label{eq:numerically-integrated-mass-matrix}
\end{align}
Observe that when the density is constant at $\rho \equiv 1$, the matrix $\Ab_{jk\ellb}$ in Eq.~\eqref{eq:numerically-integrated-mass-matrix} becomes the unit of interaction $\uoi_{jk}$, recovering the results of Sections~\ref{subsec:BE-global-arrays-Lagrange-elements} and~\ref{subsec:BE-global-arrays-tensor-product-elements}.

\end{examplebox}

\section{Boundary Conditions and Constraints}
\label{sec:constraints}

In this section, we incorporate boundary conditions and other constraints into the assembled set of equations from Sections~\ref{sec:assembly-of-global-arrays-linear} and~\ref{sec:numerical-integration}. 
We view Dirichlet boundary conditions as constraints on the system, and enforce them using the method of Lagrange multipliers. 
We begin by modifying the weak form to accommodate these constraints, and provide a block-encoding for the extended linear system that incorporates these constraints in terms of an oracle that flags the constrained degrees of freedom. We then demonstrate an explicit block-encoding of the extended system for Lagrange tensor product elements. We conclude by demonstrating how numerical integration can be used to assemble the modified force vector that enforces any Neumann boundary conditions on the system.

\subsection{Assembling the Constrained System using Lagrange Multipliers}
\label{subsec:assembling-Lagrange-multipliers}

Our primary reason for choosing the method of Lagrange multipliers to enforce boundary conditions is that it enforces boundary conditions without requiring the modification of any of the unitaries or block-encodings that emerge from the assembly procedure. 
Additionally, this method introduces the possibility of implementing more general algebraic constraints on the system and demonstrates how to extend the algorithm to problems with multiple degrees of freedom at each node\footnote{In essence, the presence of a Lagrange multiplier already means that we're solving a problem with two degrees of freedom at each node; the same circuits that assemble block-matrices may also be used to extend the problem to multiple degrees of freedom at each node.}.

Note that on a classical computer, implementing boundary conditions using Lagrange multipliers unnecessarily increases the cost of the algorithm by enlarging the size of the system of equations. Instead, one typically modifies the finite element array directly to enforce Dirichlet boundary conditions~\cite{boffi2008mixed}. 
Quantumly, however, we can double the size of the linear system being solved by adding just one additional qubit, and thus do not suffer from this issue. 
In Appendix~\ref{sec:projector-bcs}, we show that the classical method of directly modifying the assembled arrays may be analogously implemented on a quantum computer using projectors.
This method, however, does not support more general constraints or multiple degrees of freedom at each node, and so we will restrict our attention to the method of Lagrange multipliers in this section.

As in Sections~\ref{sec:assembly-of-global-arrays-linear} and~\ref{sec:numerical-integration}, consider a $d$-dimensional mesh $\Omega$ with $\numnp{} = 2^n$ global degrees of freedom (DOFs). A subset of these degrees of freedom is flagged as being on the boundary $\partial\Omega$ (or more generally, as being subject to a constraint). We label these degrees of freedom ``fixed'' (or more generally, ``constrained''). We denote the total number of degrees of freedom on the boundary by $\numbp{}$. The degrees of freedom in the interior of the mesh, $\mathrm{int}(\Omega)$, are referred to as ``free''. 
We introduce a projection operator (or indicator matrix) $\bbP_\mathrm{int}$ as the matrix that flags the degrees of freedom in the interior, i.e., 
\begin{align}
    \bbP_\mathrm{int} &:= \sum_{\np{} \in \mathrm{int}(\Omega)} \ketbra{\np{}}{\np{}}
    \,,
    \label{eq:interior-dof-projector}
\end{align}
where the notation $\np{}  \in \mathrm{int}(\Omega)$ is used to indicate degrees of freedom that lie in the interior (of which there are exactly $\numnp - \numbp$).
In Appendix~\ref{sec:projector-bcs}, we show that this projector may also be used to incorporate Dirichlet boundary conditions directly into the assembled system of equations.
The projector onto the boundary degrees of freedom is given by the orthogonal complement 
\begin{align}
    \bbP_\mathrm{bd} := I_n - \bbP_\mathrm{int}
    \,,
\end{align}
where $I_n$ is the identity operator on $n$ qubits (i.e., $I_n \equiv I^{\otimes n}$).
The splitting $\Omega \simeq \mathrm{int}(\Omega) \oplus \partial\Omega$ can be (trivially) represented using the short exact sequence 
\tikzset{
  mystyle/.style={
    nodes={inner sep=6pt},
    row sep=1.8em, 
    column sep=2.4em
    }
  } %
\begin{equation}
\begin{tikzcd}[nodes={inner sep=2pt}, mystyle]
    \setwiretype{n} 0 \arrow[r] 
    & \operatorname{int}(\Omega) \arrow[r,"\iota_{\mathrm{int}}", "\bbP_\mathrm{int}"', hookrightarrow] 
    & \Omega \arrow[r,"\partial", "I_n - \bbP_\mathrm{int}"', twoheadrightarrow] 
    & \partial\Omega \arrow[r] 
    & 0
\end{tikzcd}
    \,,
\end{equation}
where $\iota_{\mathrm{int}}$ is the inclusion map for the interior, and $\partial$ is the boundary map. By taking the Hermitian conjugate, we can also reverse the arrows in the diagram above to obtain the short exact sequence 
\begin{equation}
\begin{tikzcd}[nodes={inner sep=2pt}, mystyle]
    \setwiretype{n} 0 \arrow[r] 
    & \partial\Omega \arrow[r,"\iota_{\mathrm{bd}}", "I_n - \bbP_\mathrm{int}"', hookrightarrow] 
    & \Omega \arrow[r,"\pi", "\bbP_\mathrm{int}"', twoheadrightarrow] 
    & \operatorname{int}(\Omega) \arrow[r] 
    & 0
\end{tikzcd}
    \,,
\end{equation}
where $\iota_{\mathrm{bd}}$ is the inclusion map for the boundary, and $\pi$ is a projection map.

The projection operator $\bbP_\mathrm{int}$ can be block-encoded using an oracle $U_B$ that flags the global degrees of freedom on the boundary. That is, for all $\np{} \in [\numnp{}]$, the boundary oracle $U_B$ satisfies
\begin{align}
    U_B \ket{\np{}}\ket{0} &= 
    \begin{cases}
        U_B \ket{\np{}}\ket{0} & \text{if } \np{} \in \operatorname{int}(\Omega) \,,\\
        U_B \ket{\np{}}\ket{1} & \text{if } \np{} \in \partial\Omega \,.
    \end{cases}
    \label{eq:boundary-oracle}
\end{align}
The oracle $U_B$ itself is a $(1,1)$-block-encoding of $\bbP_\mathrm{int}$. This also yields a $(1,1)$-block-encoding of $\bbP_\mathrm{bd}$ given by $U_B (I_n \otimes \NOT)$. More generally, the $U_B$ oracle simply indicates the fixed degrees of freedom. Additionally, the fixed degrees of freedom need not comprise the entire domain boundary (e.g., part of the domain may have Dirichlet boundary conditions, while the rest is subject to Neumann boundary conditions), and need not even lie on the boundary itself.

Classically, the matrix $\bbP_\mathrm{int}$ is made to be only $\numnp{} \times \numbp{}$ to conserve memory. Here, we may double the system size by introducing Lagrange multipliers on every node in the mesh, couple only the Lagrange multipliers on the boundary to the degrees of freedom $\ket{u}$, and set the Lagrange multipliers in the interior to an arbitrary value. This simplifies the process of block-encoding the block-system of equations at the cost of one additional qubit, which will be needed to enlarge the system size regardless. 

Now, consider the QLSP $\Lc \ket{u} = \ket{f}$ produced by the FEM for the modified Poisson equation (see Section~\ref{subsec:weak-formulation-modified-Poisson}). This system is subject to the Dirichlet boundary condition $u = g$ on a subset of the boundary $\Gamma_D \subseteq \partial\Omega$ (Eq.~\eqref{eq:modified-Poisson-eqn-Dirichlet-bc}). We introduce a Lagrange multiplier $\lambda$ to enforce this boundary condition, and arbitrarily set $\lambda = 0$ on $\partial\Omega \setminus \Gamma_D$. To include this constraint, the weak formulation  (Eq.~\eqref{eq:modified-Poisson-weak-form}) of the problem becomes modified as
\begin{subequations}
\begin{align}
	\int_\Omega D\nabla u \cdot \nabla \delta u \,d\Omega + \int_{\Omega} k u \delta u \,d\Omega +  \int_{\Gamma_D} \lambda \delta u \,d\Gamma
	&= \int_\Omega \delta u f \,d\Omega + \int_{\Gamma_N} \delta u h \,d\Gamma
	\,, \\
	\int_{\Gamma_D} u \delta\lambda \,d\Gamma 
	&= \int_{\Gamma_D} g \delta\lambda \,d\Gamma
	\,.
\end{align}
\label{eq:constrained-modified-Poisson-weak-form}
\end{subequations}
The Lagrange multipliers (and their variations) are introduced with one degree of freedom per mesh node $\bar{x}_j \in \Omega$, and are discretized with basis functions that act as Kronecker deltas at the nodal locations:
\begin{align}
	\lambda &= \sum_{j \in [\numnp]} \lambda_j \delta(x - \bar{x}_j)
	\quad \text{and} \quad
	\delta\lambda = \sum_{j \in [\numnp]} \delta\lambda_j \delta(x - \bar{x}_j)
	\,,
\end{align}
while the solution $u$ and its variation $\delta u$ are still discretized according to Eq.~\eqref{eq:u-PG-discretization}. This discretization effectively enforces the Dirichlet boundary conditions pointwise, with the Lagrange multipliers $\lambda_j$ acting as point forces that ``force'' the solution towards the prescribed value $g(\bar{x}_j)$ at node $\bar{x}_j$.
With this discretization, Eq.~\eqref{eq:constrained-modified-Poisson-weak-form} leads to the block-linear system 
\begin{align}
    \begin{bmatrix}
        \Lc & I_n - \bbP_\mathrm{int}\\
        I_n - \bbP_\mathrm{int} & \bbP_\mathrm{int}
    \end{bmatrix}
    \begin{bmatrix}
        \ub \\ \lambdab
    \end{bmatrix}
    &= \begin{bmatrix}
        \fb \\ \bar{\ub}
    \end{bmatrix}
    \,,
    \label{eq:FEM-block-system}
\end{align}
where $\lambdab := \sum_{j \in [\numnp]} \lambda_j \ket{j}$ and $\bar{\ub} := \sum_{\bar{\xb}_j \in \Gamma_D} g(\bar{\xb}_j) \ket{j}$.

Following the steps in Section~\ref{subsec:weak-formulation-modified-Poisson}, Eq.~\eqref{eq:FEM-block-system} can be cast as a QLSP. Additionally, given access to a block-encoding of each block in the matrix on the left-hand side of this linear system, the block-system can be assembled using the circuits developed in Appendix~\ref{sec:block-linear-systems}. It remains to provide an explicit block-encoding for the projection operator $\bbP_\mathrm{int}$, which is done for Lagrange elements in the next subsection.

\subsection{Dirichlet Boundary Conditions for Lagrange Tensor Product Elements}

For a one-dimensional mesh $\Omega$ with $\numnp{} = 2^n$ degrees of freedom, the boundary $\partial\Omega$ has two points, one on the left-most side of the boundary and one on the right, i.e., $\ket{x^L} := \ket{0}$ and $\ket{x^R} := \ket{2^n-1}$. The boundary oracle $U_B$ is then given by $ U_B := (c_{0^n}\text{-}\NOT)(c_{1^n}\text{-}\NOT)$. In terms of quantum circuitry, a $(1,1)$-block-encodings of $\bbP_\mathrm{int}$ and $\bbP_\mathrm{bd}$ are then simply given, respectively, by 
\tikzexternalenable
\begin{align}
	\begin{quantikz}
		\lstick[4]{$\ket{\psi}_n$}& \octrl{1}
		\gategroup[5,steps=2,style={dashed,rounded
			corners,fill=blue!15, inner
			xsep=2pt},background,label style={label
			position=below,anchor=north,yshift=-0.2cm}]{{$U_B \in (1,1)\mathrm{-BE}(\bbP_\mathrm{int}) $}}
		& \ctrl{1} & \\
		& \octrl{3} & \ctrl{3} & \\
		& \gate[nwires=1,style={fill=blue!15,draw=blue!15,text height=1cm}]{\vdots} \wireoverride{n} & \gate[nwires=1,style={fill=blue!15,draw=blue!15,text height=1cm}]{\vdots} \wireoverride{n}& \wireoverride{n}\\
		& \octrl{0} & \ctrl{0} & \\
		\lstick{$\ket{0}$}& \targ{} & \targ{} & 
	\end{quantikz}
	\quad \text{and} \quad
\begin{quantikz}
    \lstick[4]{$\ket{\psi}_n$}& \octrl{1}
    \gategroup[5,steps=3,style={dashed,rounded
    corners,fill=blue!15, inner
    xsep=2pt},background,label style={label
    position=below,anchor=north,yshift=-0.2cm}]{{$(I_{N \times N} \otimes \NOT)U_B \in (1,1)\mathrm{-BE}(\bbP_\mathrm{bd}) $}}
    & \ctrl{1} & & \\
    & \octrl{3} & \ctrl{3} & & \\
    & \gate[nwires=1,style={fill=blue!15,draw=blue!15,text height=1cm}]{\vdots} \wireoverride{n} & \gate[nwires=1,style={fill=blue!15,draw=blue!15,text height=1cm}]{\vdots} \wireoverride{n}& \wireoverride{n}& \wireoverride{n}\\
    & \octrl{0} & \ctrl{0} & & \\
    \lstick{$\ket{0}$}& \targ{} & \targ{} & \targ{} & 
\end{quantikz}
\,.
\end{align}
\tikzexternaldisable
The indicator matrix for the interior degrees of freedom in $d$-dimensions, $\bbP_\mathrm{int}^{(d)}$, is built simply out of a tensor product structure as follows:
\begin{align}
	\bbP_\mathrm{int}^{\otimes d} &= \underbrace{\bbP_\mathrm{int} \otimes \cdots \otimes \bbP_\mathrm{int}}_{d\text{-times}}
	\,,
\end{align}
whence the $d$-dimensional indicator matrix for the boundary degrees of freedom is given by $\bbP_\mathrm{bd}^{(d)} := I_n^{\otimes d} - \bbP_\mathrm{int}^{\otimes d}$.

For the right-hand side of the linear system (Eq.~\eqref{eq:FEM-block-system}), the vector $\ket{\bar{u}}$ can be prepared in several ways. Firstly, if homogeneous Dirichlet boundary conditions are specified, then we can proceed directly to the linear solve. Additionally, if the values for the solution are piecewise constant on each boundary component, or if the function that specifies the boundary condition has some other simple structure, it may be preferable to directly implement an oracle $U_{\bar{u}}$ that prepares the boundary condition state $\ket{\bar{u}}$ according to $U_{\bar{u}}\ket{0}_n = \ket{\bar{u}}$. More generally, however, for an arbitrary function $g\colon \partial\Omega \to \bbR$ that specifies the Dirichlet boundary condition, we assume access to a polynomial $\Tilde{g}$ that interpolates $g$ on the nodal points of $\partial\Omega$ (this can be prepared efficiently in low dimensions by solving a linear system of size $\numbp \times \numbp$). This polynomial can be evaluated on all of $\Omega$, so that using MQET to prepare $\sum_{\xb \in [\numnp{}]} \Tilde{g}(\xb) \ket{\xb}$, we can multiply by the boundary projector $\bbP_\mathrm{bd}$ to obtain $\ket{\bar{u}} = \sum_{\xb_\mathrm{b} \in [\numbp{}]} \Tilde{g}(\xb_\mathrm{b}) \ket{\xb_\mathrm{b}}$. Finally, the block-linear system is assembled using Appendix~\ref{sec:block-linear-systems}.

\subsection{Assembly of Boundary Integrals and Force Vectors}

\begin{figure}
	\centering
	\begin{subfigure}[t]{0.45\textwidth}
		\centering
		\includegraphics[width=1.0\textwidth]{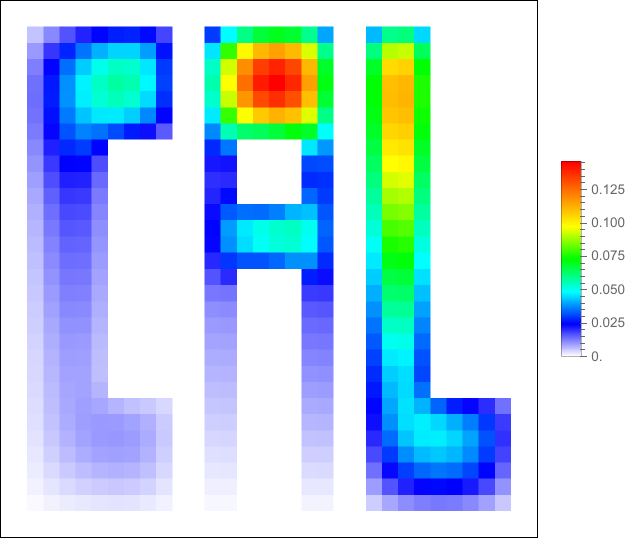}
		\caption{Solution vector $\ket{u}$}
	\end{subfigure}%
	~ 
	\begin{subfigure}[t]{0.45\textwidth}
		\centering
		\includegraphics[width=1.0\textwidth]{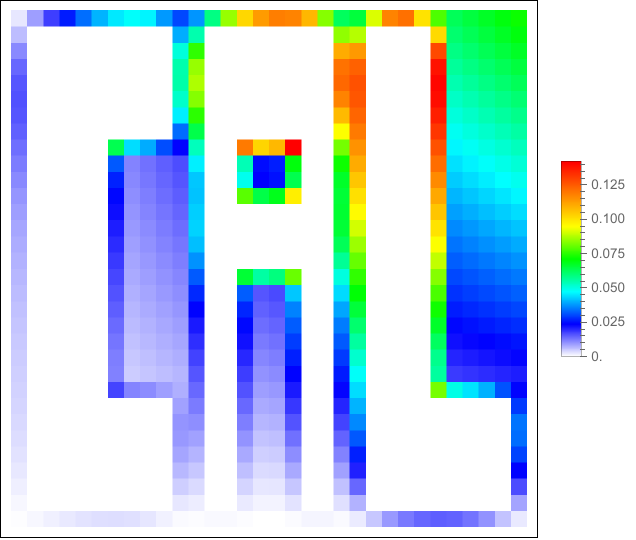}
		\caption{Lagrange multipliers $\ket{\lambda}$}
	\end{subfigure}
	\caption{Heatmaps for the solution vector and Lagrange multipliers for the Poisson equation solved in a \textsc{CAL}-shaped domain using Qu-FEM. The domain is a subset of a $2^5 \times 2^5$ grid of rectangular elements.}
	\label{fig:Wolfram-CAL}
\end{figure}

The force vector $\fb$ in the linear system, given by the assembled Finite Element approximation of the modified Poisson Equation (see Eq.~\eqref{eq:f-integral-defn}), contains two terms: a body force term $ f_{b,J} := \int_{\Omega} f N_I \,d\Omega$ and a boundary term $ f_{N,J} := \int_{\Gamma_N} h N_I \,d\Gamma $, where $J = \IX{}(j,\el)$ is the global degree of freedom index corresponding to the $j^\text{th}$ local node number in the $\el^\text{th}$ element. 
We outline the assembly of integrals over the domain $\int_{\Omega} (\cdots) \,d\Omega$, as the boundary integral $\int_{\Gamma_N} (\cdots) \,d\Gamma$ can be obtained by using an oracle (in the form of Eq.~\eqref{eq:boundary-oracle}) that signals the portion of the boundary $\Gamma_N$ corresponding to Neumann boundary conditions.
To that end, we write force vector $\fb$ in terms of its global components as 
\begin{align}
	\fb &= \sum_{J \in [\numnp]} f_J \ket{J} = \sum_{\el \in \numel} \sum_{j \in \nen} f_j^\el \ket{\IX(j,\el)}
	\,.
\end{align}
The force vector $\fb$ is then assembled from a linear form $F \colon \Uc \to  \bbR$ (see Appendix~\ref{sec:analysis-of-FEM}) defined by
\begin{align}
	F[N_i^e] &:= \int_{\Omega^\el} f(x) N_i^e(x) \,d\Omega^\el
	\,,
\end{align}
from which we have $f_i^\el := F[N_i^e]$. 
Just as was done for the general assembly problem (see Definition~\ref{def:general-assembly-problem}), numerical integration must be performed to assemble the force vector $\fb$. 
We map the problem of assembling the force vector $\ket{f}$ to a general assembly problem by preparing the operator form $\hat{f}$ of $\ket{f} := \frac{\fb}{\|\fb\|}$; see Section~\ref{sec:position-operators-and-Cartesian-coordinates} for the mapping between the operator form and vector form of a quantum state. We can then apply the quantum algorithms for numerical integration in Sections~\ref{subsec:quadrature-1D} and~\ref{subsec:quadrature-dD} to assemble the force vector.

By inspecting Eqs.~\eqref{eq:bilinear-form-representation} and~\eqref{eq:general-assembly-problem}, we find that the general assembly problem for vectors can be defined by the sum
\begin{align}
	\hat{f} &:= \sum_{\el{} \in [\numel{}]} 
	\sum_{j \in [\nen{}]} \sum_{\ell \in [G^\el]} w_\ell^{\el{}}
	N_j^\el{}(x_\ell^{\el{}}) f(x_\ell^{\el{}}) \ketbra{\mathrm{IX}(j,\el{})}{\mathrm{IX}(j,\el{})}	\,.
\end{align} 
Once again, for structured meshes with uniform element sizes, the basis function evaluations $N_j^\el{}(x_\ell^{\el{}})$ will be constant, so we may rearrange the summations in the above equation in the same manner as Eq.~\eqref{eq:solved-general-assembly-problem} by writing 
\begin{align}
	\hat{f} 
	&= \sum_{j \in [\nen{}]} \uoi_j \left( 
	\sum_{\ell \in [G]} c_{j\ell} f(\Xb_\ell)
	\right) \left(\uoi_j\right)^\dagger
	\,,
\end{align}
where $c_{j\ell} := w_\ell^\el N_j^\el{}(x_\ell^{\el{}})$ are LCU coefficients that are classically precomputed. The results of Theorem~\ref{thm:assembly-via-numerical-integration-d-dims} can then be applied to the above operator. 
Applying the operator $\hat{f}$ to the equal superposition state, and then performing $\Oc(1/\Fc_{f}^{[N]})$ rounds of exact amplitude amplification according to the filling fraction $\Fc_{f}^{[\infty]} := \frac{\| f \|_{L^2([0,1]^d)}}{\| f \|_{L^\infty([0,1]^d)}}$ prepares the force vector $\ket{f}$ (see Section~\ref{sec:position-operators-and-Cartesian-coordinates}).
Note that for $\Qc_p([0,h])$ elements, since the basis functions are non-negative and form a partition of unity (by Lemma~\ref{lem:partition-of-unity}), we have $\sum_{j \in [\nen{}]} \sum_{\ell \in [G]} |c_{j\ell}| \equiv 1$.

\subsection{Scaling Estimates}

The quantum algorithm for numerical integration given in Sections~\ref{subsec:quadrature-1D} and~\ref{subsec:quadrature-dD} scales logarithmically with the system $N$, but exponentially with the dimension $d$ (as per Theorem~\ref{thm:assembly-via-numerical-integration-d-dims}).
For many problems of interest, however, the body force $f \equiv 0$ so that $\ket{f_b} \equiv 0$, and Dirichlet or periodic boundary conditions are prescribed so that the forces arising from the Neumann boundary condition satisfy $\ket{f_N} \equiv 0$. In particular, high-dimensional partial differential equations such as the backward Kolmogorov equation (BKE) satisfy this constraint~\cite{kharazi2024explicitblockencodingsboundary}, so we do not suffer from this exponential overhead.
For low-dimensional problems, assembly via numerical integration scales logarithmically with the number of grid points, and polynomially with the number of Guass points and polynomial order of the element, so the methods of Section~\ref{sec:numerical-integration} provide a viable strategy for preparing the force vector.

\section{Numerical Demonstrations of Qu-FEM}
\label{sec:demo}

We implement the Quantum Algorithm for the Finite Element Method (Qu-FEM) in a Mathematica notebook\footnote{Self-contained notebooks that reproduce the results in this section may be found at: \href{https://github.com/AhmadAlkadri/Qu-FEM}{https://github.com/AhmadAlkadri/Qu-FEM}}~\cite{Mathematica} that utilizes Wolfram's Quantum Framework~\cite{mathematicaQC}.
In particular, we implement two special cases of the modified Poisson equation given in Section~\ref{subsec:weak-formulation-modified-Poisson} in two dimensions. 
These examples are intended to verify that the quantum circuitry proposed in this work replicates the well-known values that the stiffness and mass matrices take on for these problems.

\subsection{Poisson's Equation in a \textsc{CAL}-Shaped Domain}

We first showcase the solution of Poisson's equation by explicitly constructing a circuit that solves the PDE on a \textsc{CAL}-shaped domain (Figure~\ref{fig:Wolfram-CAL}). An explicit construction for the oracle that creates this domain may be found in the Mathematica~\cite{Mathematica} notebook. The Poisson's equation we consider corresponds to Eq.~\eqref{eq:modified-Poisson-eqn} when $D \ne 0$ is constant and $k \equiv 0$. For simplicity, we can absorb the diffusion constant $D$ into the source term $f$ and write 
\begin{align}
	-\nabla^2 u &= f
	\,,
\end{align}
with homogeneous boundary conditions (i.e., $u = 0 \quad \text{on } \partial(\textsc{CAL})$). The force function is chosen to be $f(x^0,x^1) = x^1 x^0$. The assembly is simulated using the quantum circuitry developed in Sections~\ref{sec:assembly-of-global-arrays-linear}, \ref{sec:numerical-integration}, and~\ref{sec:constraints}. The post-selection is assumed to succeed via projections onto the $\bra{0}$ basis state, and the linear solve is performed classically. The matrix elements that we obtain for the stiffness and mass matrices match the analytical values we expect from Examples~\ref{ex:1d-K-global} and~\ref{ex:numerical-integration-of-mass-matrix}.
That is, we verify that the simulated Qu-FEM algorithm yields the same result as the classical FEM for the mesh in \cref{fig:Wolfram-CAL}.

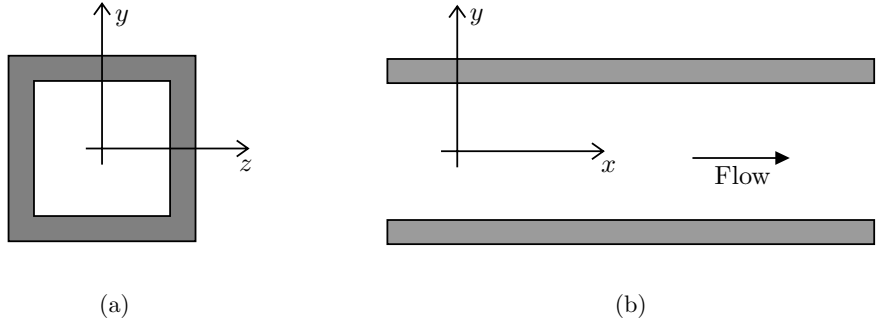
\begin{figure}
	\centering
	\scalebox{0.9}{%
	\begin{subfigure}[b]{0.45\textwidth}
		\centering
		\tikzset{every picture/.style={line width=0.75pt}} %

\begin{tikzpicture}[x=0.75pt,y=0.75pt,yscale=-1,xscale=1,scale=0.9]

\draw  [fill={rgb, 255:red, 128; green, 128; blue, 128 }  ,fill opacity=1 ] (76.44,99.67) -- (191.67,99.67) -- (191.67,214.89) -- (76.44,214.89) -- cycle ;
\draw  [fill={rgb, 255:red, 255; green, 255; blue, 255 }  ,fill opacity=1 ] (92.11,115.33) -- (176,115.33) -- (176,199.22) -- (92.11,199.22) -- cycle ;
\draw  (124.06,157.28) -- (224.06,157.28)(134.06,67.28) -- (134.06,167.28) (217.06,152.28) -- (224.06,157.28) -- (217.06,162.28) (129.06,74.28) -- (134.06,67.28) -- (139.06,74.28)  ;
\draw  [draw opacity=0] (42,61) -- (236,61) -- (236,236.83) -- (42,236.83) -- cycle ;

\draw (217,162.9) node [anchor=north west][inner sep=0.75pt]    {$z$};
\draw (140.5,69.57) node [anchor=north west][inner sep=0.75pt]    {$y$};

\end{tikzpicture}
		\caption{}
		\label{fig:square-duct-a}
	\end{subfigure}
	~
	\begin{subfigure}[b]{0.45\textwidth}
		\centering
		\tikzset{every picture/.style={line width=0.75pt}} %

\begin{tikzpicture}[x=0.75pt,y=0.75pt,yscale=-1,xscale=1,scale=0.9]

\draw  [fill={rgb, 255:red, 155; green, 155; blue, 155 }  ,fill opacity=1 ] (307,100) -- (607,100) -- (607,115) -- (307,115) -- cycle ;
\draw  [fill={rgb, 255:red, 155; green, 155; blue, 155 }  ,fill opacity=1 ] (307,200) -- (607,200) -- (607,215) -- (307,215) -- cycle ;
\draw  (340.06,157.28) -- (440.06,157.28)(350.06,67.28) -- (350.06,167.28) (433.06,152.28) -- (440.06,157.28) -- (433.06,162.28) (345.06,74.28) -- (350.06,67.28) -- (355.06,74.28)  ;
\draw    (495,161.5) -- (552,161.5) ;
\draw [shift={(555,161.5)}, rotate = 180] [fill={rgb, 255:red, 0; green, 0; blue, 0 }  ][line width=0.08]  [draw opacity=0] (8.93,-4.29) -- (0,0) -- (8.93,4.29) -- cycle    ;

\draw (437,161.23) node [anchor=north west][inner sep=0.75pt]    {$x$};
\draw (356.17,66.9) node [anchor=north west][inner sep=0.75pt]    {$y$};
\draw (507,164.5) node [anchor=north west][inner sep=0.75pt]   [align=left] {Flow};

\end{tikzpicture}
		\caption{}
		\label{fig:square-duct-b}
	\end{subfigure}
	}
	\caption{Flow in a square duct of constant cross section: (a) section perpendicular to the direction of flow; and (b) section parallel to the direction of flow.}
	\label{fig:square-duct}
\end{figure}

\subsection{Navier-Stokes for Unidirectional Flow}

Consider the fully developed flow of a Newtonian fluid in a channel of square cross section with unit side length (see Figure~\ref{fig:square-duct}). We assume that the fluid is incompressible and isothermal so that the Navier-Stokes equations can be applied. Here, we label the coordinates as $(x^0,x^1,x^2) = (x,y,z)$ for convenience. The flow is assumed to be unidirectional in the $x$-direction, so that $v_z = v_z(x,y)$. Furthermore, we assume that the pressure drop $\bracfrac{d p}{d x}$ is constant. The governing PDE is then 
\begin{align}
	\nabla^2 u_z &= \frac{1}{\mu}\frac{d p}{d x}
	\,,
\end{align}
where $\mu$ is the viscosity of the fluid. We assume no-slip boundary conditions, so that $u_z = 0$ on the boundary of the square.

The resulting flow profile is plotted in Figure~\ref{fig:Wolfram-NS}. On a quantum computer, rather than plotting the full profile, one would use the state vector $\ket{u}$ that encodes the solution to estimate observables of the form 
\begin{align}
	(r,u) &:= \int_\Omega r(x) u(x) \,d\Omega 
	\,,
\end{align}
where $r \colon \Omega \to \bbR$ is some scalar field. Here, choosing $r$ to be the mass density would recover the total mass flowrate through the square duct.

\begin{figure}
    \centering
    \begin{subfigure}[t]{0.45\textwidth}
        \centering
        \includegraphics[width=1.0\textwidth]{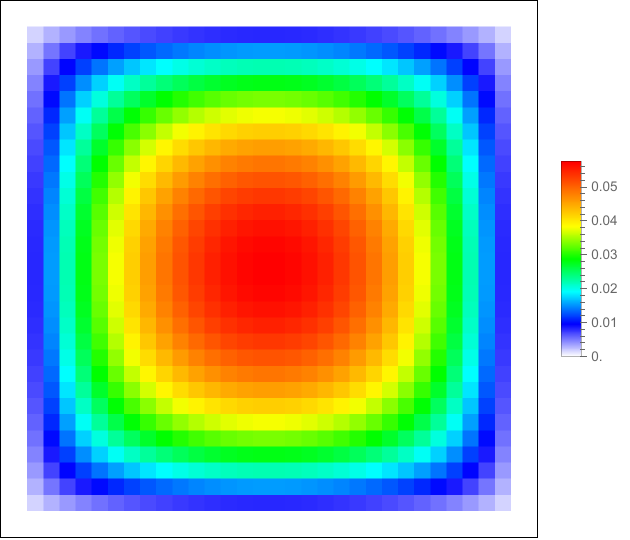}
        \caption{Solution vector $\ket{u}$}
    \end{subfigure}%
    ~ 
    \begin{subfigure}[t]{0.45\textwidth}
        \centering
        \includegraphics[width=1.0\textwidth]{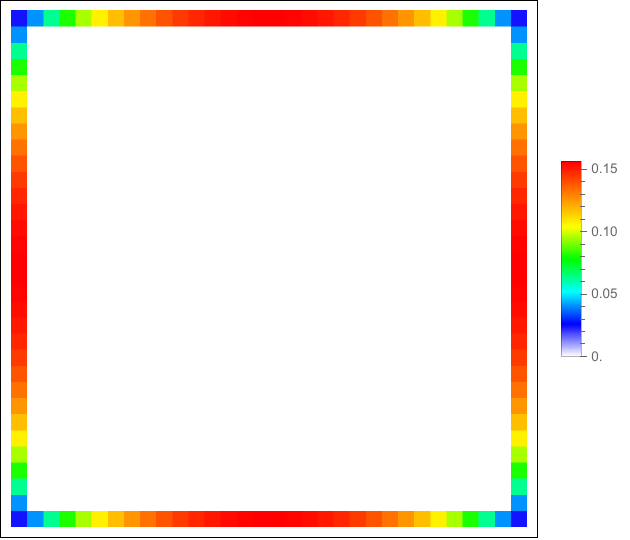}
        \caption{Lagrange multipliers $\ket{\lambda}$}
    \end{subfigure}
    \caption{Heatmaps for the solution vector and Lagrange multipliers produced by Qu-FEM for the unidirectional flow through a square duct of a fluid governed by the Navier Stokes equations. The solution vector gives the velocity of the fluid through the duct, while the Lagrange multipliers enforce a no-slip boundary condition (i.e., zero velocity) at the edges of the domain. 
    The domain is a $2^5 \times 2^5$ grid of rectangular elements.
    	}
    \label{fig:Wolfram-NS}
\end{figure}

\section{Discussion and Conclusions}
\label{sec:conclusions}

We introduced \textbf{Qu-FEM}, a fault-tolerant quantum algorithm that implements the classical finite element method on a quantum computer.  Below, we summarize the main contributions of this work, discuss their relation to other quantum PDE solvers, and outline directions for future research.

\paragraph{Summary of key contributions.}
\begin{enumerate}
	\item \textbf{Efficient assembly through the \emph{unit of interaction}.}  
	We defined the interaction matrices $\uoi_{jk}$ (Definition~\ref{defn:unit-of-interaction}) and showed that every global FEM array with constant elemental contributions (which corresponds to constant coefficients in Eq.~\eqref{eq:modified-Poisson-eqn}) can be written as a \emph{constant-size} linear combination of these units. 
	Together with state-preparation oracles for the elemental contributions, this yields efficient block-encodings of stiffness and mass matrices that utilize, respectively, $\Oc(d^2 n m p^2)$ and $\Oc(dn m p^2)$ Toffoli or simpler gates for tensor product elements $\Qc_p(\Omega^{e,d})$ (Theorems~\ref{thm:query-complexity-Assembly-1D-degree-1} and~\ref{thm:query-complexity-Assembly-dD-degree-p}).
	
	\item \textbf{Numerical integration on-the-fly.}  
	We defined the local-to-global node number indicator matrix $\uoi_j$ (Definition~\ref{defn:local-to-global-node-number-indicator-matrix}), 
	and extended Qu-FEM to problems with spatially varying coefficients, non-homogeneous Neumann data and general body forces (Theorems~\ref{thm:assembly-via-numerical-integration-1D} and~\ref{thm:assembly-via-numerical-integration-d-dims}).
	This construction relied on Gauss-Legendre quadrature of order $G$, together with QSP/MQET of a set of position operators. The gate cost for this block-encoding is $ \Oc\left( dn \cdot (D^d G^d + p^{2d} m) \right) $ for tensor product elements $\Qc_p(\Omega^{e,d})$, which again is logarithmic in the system size, but scales exponentially with dimensionality of the problem.
	
	\item \textbf{Projector-free treatment of constraints.}  
	Dirichlet conditions were incorporated through a block-encoded Lagrange-multiplier system (Section~\ref{sec:constraints}), allowing boundary values and other algebraic constraints to be handled without modifying previously assembled block-encodings.
	
	\item \textbf{Explicit block-encodings for Cartesian domains.}  
	We constructed explicit quantum circuits for assembly in Cartesian geometries, thus showing that the geometric capabilities of Qu-FEM are at least as general as other quantum PDE solvers.
	
	\item \textbf{Numerical demonstrations.}  
	Classical simulation of these quantum circuits (Section~\ref{sec:demo}) reproduced classical solutions for Poisson and steady Navier-Stokes uni-directional flow benchmarks on non-trivial two-dimensional geometries, thus validating the analytic circuits.
\end{enumerate}

\paragraph{Where Qu-FEM fits among other quantum PDE solvers.}\mbox{}

Existing quantum approaches~\cite{childs2021high,kharazi2024explicitblockencodingsboundary} chiefly target finite-difference or spectral discretizations, and therefore require rectangular domains or periodic grids.  Enforcing general boundary conditions in these other existing quantum PDE approaches also remains a challenge.
On the other hand, the Qu-FEM framework retains the general geometric flexibility of the classical FEM, as the oracles (introduced in Section~\ref{subsec:unit-of-interaction-and-indicator-matrix}) utilized during the assembly process are \textit{a priori} independent of the particular geometry of the problem. 
Thus, Qu-FEM has the potential to access the $hp$-adaptivity and rich element catalogue of the classical FEM, while preserving the logarithmic state preparation advantage of amplitude encoding for Cartesian domains.
Additionally, Qu-FEM retains the automated bookkeeping of the interactions between nodes from the classical FEM assembly procedure, avoiding errors such as that made Ref.~\cite{clader2013preconditioned} (see \cref{subsec:unit-of-interaction-and-indicator-matrix} for a more detailed discussion).
The current explicit implementation of Qu-FEM, however, faces several limitations and challenges.

One limitation of this algorithm is the exponential scaling with the dimension of the problem $ \Oc\left( dn \cdot (D^d G^d + p^{2d} m) \right) $ for the numerical integration procedure in Section~\ref{sec:numerical-integration} (for non-constant elemental arrays). 
This exponential scaling in dimension is a result of using a black-box numerical integrator; in the worst case, $\Omega(G^d)$ grid points are needed to resolve the variation of a function in each of the $d$ dimensions.
Once a particular problem has been chosen, however, more efficient numerical integration schemes can be made if additional structure is present in the problem. For example, Ref.~\cite{kharazi2024explicitblockencodingsboundary} utilizes the finite difference method (FDM), and  
considers Hamiltonians with a potential term consisting of pairwise interactions (with a radially symmetric inverse power law). Here, the gate cost for implementing the potential function scales polynomially with the dimensionality of the problem.
Similarly, alternative quadrature schemes that match the problem of interest can be incorporated into Qu-FEM.
When comparing our algorithm to the quantum implementation of the FDM in Ref.~\cite{kharazi2024explicitblockencodingsboundary}, we find that the asymptotic Toffoli count for Qu-FEM\footnote{The FDM for a uniform mesh should be compared to the FEM with constant elemental contributions.} of $\Oc(d^2 p^2 n m)$ matches the FDM's logarithmic scaling with the system size $2^n$, but is quadratically slower in the dimension $d$ and the order $p$. Future work, however, has the potential to improve this scaling to match the FDM by increasing the efficiency of the block-encoding for the unit of interaction (Proposition~\ref{prop:BE-unit-of-interaction-p}) and the quantum circuit for arithmetic (Eq.~\eqref{eq:quantum-division-unitary}).
Even without improvements to this scaling, however, advancing the Qu-FEM framework remains a promising direction due to its potential to simulate in more arbitrary geometries.

There are several directions for increasing the generality of the problem under consideration (beyond the modified Poisson equation of Section~\ref{subsec:weak-formulation-modified-Poisson}).
Extending connectivity oracles to unstructured simplicial meshes and space-filling elements remains an important next step. One promising route for this extension is to partition the mesh into a small set of sub-meshes with injective connectivity matrices (as described in Section~\ref{subsec:unit-of-interaction-and-indicator-matrix}). 
Additionally, our algorithm does not currently address the solution of nonlinear or time-dependent PDEs. 
Classically, nonlinearities are typically addressed by assembling the Jacobian matrix (sometimes referred to as the \textit{tangent matrix} in the FEM literature) and carrying out Newton-Raphson iterations in accordance with a time discretization scheme.
The no-cloning theorem, however, makes a direct implementation of this method on a quantum computer impossible, as one cannot both generate the Jacobian and update each Newton iteration using only one copy of the solution. 
Future work will thus require modifications to the classical Newton-Raphson method, or an alternate strategy for handling nonlinearities (such as Carlemann linearization~\cite{liu2024dense} or the quantum nonlinear Schr\"{o}dinger linearization technique~\cite{lloyd2020quantumalgorithmnonlineardifferential}).
Finally, we note that explicit bookkeeping for problems that have multiple degrees of freedom at each node also needs to be addressed, though this can likely be implemented as an extension of the assembly of block-linear systems introduced in Section~\ref{sec:constraints} and Appendix~\ref{sec:block-linear-systems}.

With respect to achieving a quantum advantage, we outline a different set of challenges.
First, we note that the cost of solving the QLSP is proportional to the condition number of the linear operator to be inverted.
While bounded‐condition number estimates exist for many elliptic problems, these condition numbers typically scale polynomially with the system size $N$~\cite{brenner2008mathematical,bagherimehrabFastQuantumAlgorithm2023}. Practical quantum advantage will require quantum-implementable preconditioners (e.g., multi-grid, BPX, and SPAI) compatible with block-encodings. 
This issue is well-known in the literature, and has recently gained the attention of several works~\cite{deiml2024quantum,lapworth2025preconditioned,jin2025quantumpreconditioningmethodlinear,bagherimehrabFastQuantumAlgorithm2023}.
With respect to the input/output problem~\cite{aaronson2015read}, identifying and extracting observables of interest (such as forces, fluxes, and energies) also remains a challenge.
For practical engineering problems, we expect the input state preparation to be efficient, as the input data is known from the problem description.
On the other hand, algorithms that learn global quantities without utilizing exponentially many copies of the full solution state are essential, and a key step in proving a quantum advantage that persists once the cost of input/output have been included. 
Amplitude estimation~\cite{suzuki2020amplitude}, shadow tomography~\cite{huang2020predicting}, and standard Monte Carlo (shot-based) sampling~\cite{cerezo2021variational} are natural starting points for addressing the output problem.

\paragraph{Outlook.}\mbox{}

In summary, Qu-FEM shows that the classical finite element method---one of the most versatile tools in computational science---can be implemented as a quantum algorithm in the fault-tolerant era of quantum computation. This quantum implementation shows promise for a quantum advantage in time complexity for high-dimensional problems (when elemental arrays are constant), and a quantum advantage in space complexity for low-dimensional problems (when the elemental arrays are non-constant). 
The modular nature of this implementation (input state preparation $\ket{f}$ $\rightarrow$ connectivity oracle implementation $\Oc_\IX$ $\rightarrow$ quantum assembly (Qu-FEM) $\rightarrow$ QSVT/QSP solve $\rightarrow$ observable output $(r,u)$) facilitates integration with future advances in quantum linear solvers, preconditioning, observable estimation, and block-encoding strategies.
Although significant work remains, future quantum algorithms for FEM have the potential to drive quantum simulations of Kolmogorov-scale turbulence~\cite{pope2000turbulent}, multidimensional Fokker–Planck dynamics~\cite{risken1996fokker}, and other frontier challenges that exceed classical exascale limits.

\section*{Acknowledgements}

We are grateful to Philipp Schleich and Jakob Huhn for helpful discussions, and to Prof. Lin Lin for valuable feedback. K.K.M. gratefully acknowledges discussions with Prof. Panayiotis Papadopoulos. 
This work was supported by the U.S. Department of Energy, Office of Science, Office of Advanced Scientific Computing Research under Award Number DE-SC0023273.
Part of this research was performed while A.M.A. and T.D.K. were visiting the Institute for Pure and Applied Mathematics (IPAM), which is supported by the National Science Foundation (Grant No. DMS-1925919).
The authors also acknowledge the support of the University of California, Berkeley.
A.M.A. is also supported by the Natural Sciences and Engineering Research Council of Canada Postgraduate Scholarships---Doctoral program.

\vspace{02pt}

\printbibliography

\nomenclature[A]
{$C(\Omega) = \{u \colon \Omega \to \bbR \mid u \text{ is continuous}\}$}
{}
\nomenclature[A]
{$C^k(\Omega) = \{u \colon \Omega \to \bbR \mid u \text{ is } k \text{-times continuously differentiable}\}$}
{}
\nomenclature[A]
{\(L^p(\Omega) = \{u \colon \Omega \to \bbR \mid u \text{ is Lebesgue measurable and } \abs{\abs{u}}_{L^p(\Omega)} < \infty\}\) ($1 \le p < \infty$)}
{}
\nomenclature[A]
{\(L^\infty(\Omega) = \{u \colon \Omega \to \bbR \mid u \text{ is Lebesgue measurable and } \abs{\abs{u}}_{L^\infty(\Omega)} < \infty\}\)}
{}
\nomenclature[A]
{\(W^{k,p}(\Omega), H^k(\Omega)\)}
{Sobolev Spaces}
\nomenclature[A]
{\( \Uc \text{ and } \Uc_0\)}
{ Space of admissible solutions and space of admissible variations, respectively}

\nomenclature[B]
{$\numel{}$}
{Number of elements}
\nomenclature[B]
{$\numnp{}$}
{Number of global nodal points}
\nomenclature[B]
{$\nen{}$}
{Number of element nodes}
\nomenclature[B]
{$K$}
{Stiffness matrix}
\nomenclature[B]
{$M$}
{Mass matrix}
\nomenclature[B]
{$K^e$}
{Elemental stiffness matrix}
\nomenclature[B]
{$M^e$}
{Elemental mass matrix}
\nomenclature[B]
{$\ket{N^e}$}
{(Unnormalized) vector of elemental basis functions}
\nomenclature[B]
{$\numbp{}$}
{Number of boundary nodal points}

\nomenclature[B]
{$\Qc_p(\Omega^e)$}
{Lagrange element of order $p$ on the domain $\Omega^e$}

\nomenclature[B]
{$\Omega$}
{Global Computational domain (a subset of $\bbR^d$)}
\nomenclature[B]
{$\Omega^e$}
{Element domain (a subset of $\Omega$)}

\nomenclature[C]
{$\norm{A} := \sup_{\braket{\psi \vert \psi} = 1} \norm{A\ket{\psi}}$}
{Spectral norm}
\nomenclature[C]
{$\norm{A}_\mathrm{max} := \max_{ij} \abs{A_{ij}}$}
{Max norm}
\nomenclature[C]
{$\norm{A}_1 := \sum_{ij} \abs{A_{ij}}$}
{Entrywise $\ell_1$-norm}
\nomenclature[C]
{$(\alpha,m)$-$\mathrm{BE}(A)$}
{Block-encoding of $A$ with subnormalization $\alpha$ and $m$ ancilla qubits} %
\nomenclature[C]
{$S := \ketbra{0}{0} + i \ketbra{1}{1}$}
{Phase gate}
\nomenclature[C]
{$S^1$}
{Shift operator}

\nomenclature[C]
{$A_{ij}$ or $A_{i,j}$}
{Components of the matrix $A$}

\nomenclature[C]
{$\{I,X,Y,Z\}$}
{Pauli matrices}
\nomenclature[C]
{$H$}
{Hadamard gate}

\nomenclature[D]
{$\delta_{jk}$}
{Kronecker delta}
\nomenclature[D]
{$\bbD$}
{Closed unit disk in $\bbC$}
\nomenclature[D]
{$\operatorname{SU}(2)$}
{Special unitary group of degree $2$; all matrices in $\bbC^{2 \times 2}$ with determinant equal to $1$}

\newpage
\printnomenclature
\newpage

\begin{appendices}

\renewcommand{\theequation}{\thesection.\arabic{equation}}
\setcounter{equation}{0}
\section{Analysis of the Finite Element Method}
\label{sec:analysis-of-FEM}

A study of the convergence of solutions from the Finite Element Method (FEM) requires identification of the function space that the solution lies in. In general, this is accomplished by assessing the continuity and differentiability requirements imposed by the weak formulation of the partial differential equation (PDE). Identification of the appropriate function space (along with a suitable norm) enables a theoretical analysis of the well-posedness of the problem, uniqueness of the solution, and convergence of the solution. 
Moreover, identifying the function space also informs element design, as one often chooses the finite element basis so that the function space it spans ``conforms to'' (i.e., is a subset of) the function space required by the problem. For example, the Lagrange elements that we consider in this work are $H^1$-conforming elements. In this section, we review some definitions and properties of function spaces used in the numerical analysis of the FEM. For a more detailed presentation of these results, we refer the reader to~\cite{brenner2008mathematical,chapelle2011finite,donea2003finite,papadopoulos2015280a}.

Let $\Omega \subset \bbR^d$ be a connected, bounded set. We will be interested in solving PDEs within $\Omega$, subject to some boundary conditions. 
The PDEs that we consider in this work also have real-valued solutions. Although quantum mechanics is set in complex space, this presents no technical issues, as we may assume without loss of generality that the amplitudes for any quantum state are real-valued for the algorithms that we consider here.
Solutions that emerge from FEM need only satisfy the PDE in a weak sense, and may not contain the same degree of smoothness or regularity as that required by the strong form of the PDE (i.e., solutions may not live in $C^k(\Omega)$). 
The appropriate setting for analyzing solutions to the weak (or variational) formulation of PDEs is Sobolev spaces.

\begin{defn}[$L^p$ Space]
    For any $1 \le p < \infty$, the $L^p$ space of functions with domain $\Omega$ are defined as 
    \begin{align}
        L^p(\Omega) = \{u \colon \Omega \to \bbR \mid u \text{ is Lebesgue measurable and } \abs{\abs{u}}_{L^p(\Omega)} < \infty\} \,,
    \end{align}
    where the norm $\|\bullet\|_{L^p(\Omega)}$ is defined as\footnote{An $L^p$ space can be defined for $p=\infty$ as well, but we will not require this space for our discussion of the FEM.} 
    \begin{align}
        \|u\|_{L^p(\Omega)} &:= \left(\int_\Omega \|u\|^p \,d\Omega \right)^{1/p} \,.
    \end{align}
\end{defn} \noindent
It can be shown that the $L^p$ spaces form Banach spaces (i.e., complete normed vector spaces)~\cite{evans2022partial}. In particular, we will be interested in the space of square-integrable functions $L^2(\Omega)$, which together with the inner product 
\begin{align}
    \langle u,v \rangle_{L^2(\Omega)} := \int_\Omega u v \,d\Omega 
    \quad ; \quad 
    u,v \in L^2(\Omega) \,,
\end{align}
forms a Hilbert space\footnote{By Hilbert space, we mean a complete inner product space
rather than ``the'' particular Hilbert space $\Hc$ of state vectors $\ket{\psi}$.}.

\begin{defn}[Sobolev Space]
    For any integer $m \ge 0$, the Sobolev space of order $m$ is defined as
    \begin{equation}
        H^m(\Omega) := \left\{u \colon \Omega \to \bbR \mid D^\alpha u \in L^2(\Omega) \,, \quad \forall |\alpha| \le m \right\} \,,
    \end{equation}
    where $\alpha$ is a $d$-dimensional multi-index, and 
    \begin{align}
        D^\alpha &:= \frac{\partial^{|\alpha|}}{\partial x_1^{\alpha_1} \dots \partial x_1^{\alpha_d}}  
        \quad , \quad 
        |\alpha| = \alpha_1 + \dots + \alpha_d \,.
        \label{eq:D-alpha}
    \end{align}
    In Eq.~\eqref{eq:D-alpha}, derivatives are taken in the weak sense (i.e., in the sense of distributions). We further define the norm
    \begin{align}
        \|u\|_{H^m(\Omega)} := \left(\sum_{|\alpha| \le m} \|D^\alpha u \|_{L^2(\Omega)}^2\right)^{1/2} \,,
        \label{eq:Sobolev-norm}
    \end{align}
    and the inner product 
    \begin{align}
        \langle u,v \rangle_{H^m(\Omega)} := \sum_{|\alpha| \le m} \langle D^\alpha u, D^\alpha v\rangle_{L^2(\Omega)} \,.
        \label{eq:Sobolev-inner-product}
    \end{align}
    Note that in Eqs.~\eqref{eq:Sobolev-norm} and~\eqref{eq:Sobolev-inner-product}, derivatives of different orders are summed together, which requires that the spatial coordinates be non-dimensionalized. When the coordinates do not correspond to dimensionless quantities, a conversion factor must be applied to each term to make the equation dimensionally consistent.
\end{defn}\noindent
With the norm and inner product defined above, the Sobolev spaces are Hilbert spaces~\cite{evans2022partial}.
In particular, note that $H^0(\Omega) = L^2(\Omega)$. More generally, we have the inclusions $H^m(\Omega) \subseteq L^2(\Omega) \subseteq L^1_\text{loc}(\Omega)$, where 
\begin{align}
    L^1_\text{loc}(\Omega) &:= \left\{ u \colon \Omega \to \bbR \mid u \in L^1(K) \text{ for all compact } K \subset \Omega \right\} \,.
\end{align}
The space $L^1_\text{loc}(\Omega)$ will allow for the definition of \textit{weak derivatives} over each element in our mesh (see~\cite{brenner2008mathematical} for a more detailed discussion).

The first step in formulating a problem in the FEM is to convert the strong formulation of a PDE into a weak formulation. In this work, we consider linear PDEs, for which the weak formulation takes a general form.
In what follows, let $\Uc$ be a Hilbert space with the inner product denoted by $\langle \bullet , \bullet \rangle_{\Uc}$ and the induced norm 
$\| \bullet \|_{\Uc} := \sqrt{\langle \bullet , \bullet \rangle_{\Uc}}$.
\begin{defn}[Abstract Weak Formulation of a Linear Problem]
\label{def:abstract-weak-formulation}
    Consider a (continuous) bilinear form $B \colon \Uc \times \Uc \to \bbR$ and a (continuous) linear form $F \colon \Uc \to \bbR$. The \textbf{abstract weak formulation of a linear problem} is to find $u \in \Uc$ such that 
    \begin{align}
        B[u,\delta u] = F[\delta u] 
        \label{eq:abstract-weak-formulation}
    \end{align}
    for all $\delta u \in \Uc$. 
\end{defn}\noindent
We refer to the linear form $F$ in Eq.~\eqref{eq:abstract-weak-formulation} as the \textit{data} in the problem, as the Riesz Representation Theorem (see \cite[Appendix D]{evans2022partial}) asserts the existence of an element $f \in U$ such that 
\begin{align}
    F[\delta u] = \langle f , \delta u \rangle_{\Uc}
    \label{eq:data-F-Riesz-representation}
\end{align}
for all $\delta u \in \Uc$.
Definition~\ref{def:abstract-weak-formulation} is also known as the abstract \textit{linear variational formulation}~\cite{chapelle2011finite}, as it can be recast as an optimization problem whenever the bilinear form $B$ in Eq.~\eqref{eq:abstract-weak-formulation} is symmetric. That is, every symmetric bilinear form is derivable from a potential functional (a result sometimes referred to as \textit{Vainberg's Theorem}~\cite{papadopoulos2015280a}), which up to a constant can be taken to be 
\begin{align}
    J[u] := \frac{1}{2} B[u,u] - F[u] \,,
\end{align}
whence extremization of its variation $\delta J[u;\delta u] = 0$ is equivalent to Eq.~\eqref{eq:abstract-weak-formulation}. Symmetry, however, is not necessary for an abstract weak formulation. 

We now turn our attention to the existence, uniqueness, and stability (i.e., the well-posedness) of solutions to the abstract weak formulation. A fundamental result that establishes this for linear problems in the FEM is the Lax-Milgram theorem. 
\begin{prop}[Lax-Milgram Theorem]
    Suppose that a (continuous) bilinear form $B \colon \Uc \times \Uc \to \bbR$ is coercive (also known as ``elliptic''), i.e., there exists a constant $\gamma > 0$ such that 
    \begin{align}
        B[u,u] \ge \gamma \| u \|_{\Uc}^2 \,,
        \label{eq:coercivity-defn}
    \end{align}
    for every $u \in \Uc$.
    Recall that continuity of the bilinear form means that there exists a constant $M > 0$ such that 
    \begin{align}
        \left| B[u,\delta u] \right| &\le M \| u \|_{\Uc} \| \delta u \|_{\Uc} \,,
        \label{eq:continuity-defn}
    \end{align}
    for all $u,\delta u \in \Uc$. Then there exists a unique element $u \in \Uc$ that solves the abstract weak formulation in Eq.~\eqref{eq:abstract-weak-formulation}. Furthermore, the data $F$ yields the two-sided stability estimate 
    \begin{align}
        \gamma \| f \|_{\Uc} &\le \| u \|_{\Uc} \le \frac{1}{\gamma} \| f \|_{\Uc} \,,
        \label{eq:Lax-Milgram-solution-bounds}
    \end{align}
    where $f \in U$ is related to the data $F$ by Eq.~\eqref{eq:data-F-Riesz-representation}.
\end{prop}

\begin{proof}
	See \cite[Section~6.2]{evans2022partial} and \cite[Section~3.2]{chapelle2011finite}.
\end{proof}

\noindent
Equation~\eqref{eq:Lax-Milgram-solution-bounds} demonstrates that the constant $\gamma$ from coercivity controls the well-posedness of the problem; the upper bound ensures stability by requiring that small forcing terms $F$ yield small solutions $u$, while the lower bound ensures that the solution is not arbitrarily small for a given nonzero $F$, which indicates that the problem is well-conditioned. In general, the smaller that $\gamma$ is, the more ill-conditioned the problem becomes, as the bounds in Eq.~\eqref{eq:Lax-Milgram-solution-bounds} indicate a large sensitivity to small perturbations in the data $F$.

Up until now, the Hilbert spaces $\Uc$ we have been dealing with are infinite dimensional. The Finite Element Method seeks to solve Eq.~\eqref{eq:abstract-weak-formulation} in some finite dimensional subspace $\Uc_h \subset \Uc$ that we call the \textit{finite element space}. The index ``$h$'' in the finite dimensional subspace $\Uc_h$ is used to indicate the dependence of the finite dimensional space on a refinement parameter $h$ that controls the mesh size. We can formulate the finite dimensional version of Definition~\ref{def:abstract-weak-formulation} as follows.
\begin{defn}[Classical FEM Problem]
\label{def:classical-FEM-problem}
    Let $B \colon \Uc \times \Uc \to \bbR$ be a (continuous) bilinear form, and $F \colon \Uc \to \bbR$ a (continuous) linear form. Let $\Uc_h \subset \Uc$ be a finite dimensional subspace. The \textbf{classical FEM problem} is to find $u_h \in \Uc_h$ such that 
    \begin{align}
        B[u_h,\delta u] = F[\delta u] 
        \label{eq:classical-FEM-problem}
    \end{align}
    for all $\delta u \in \Uc_h$. 
\end{defn}\noindent
The existence of a solution to the classical FEM problem (Definition~\ref{def:classical-FEM-problem}), as well as its relation to the (exact) solution to the abstract weak formulation (Definition~\ref{def:abstract-weak-formulation}) is established by C\'{e}a's Lemma.

\begin{prop}[C\'{e}a's Lemma]
    Suppose that the bilinear form $B$ in Eq.~\eqref{eq:classical-FEM-problem} is coercive. Then there exists a unique solution $u_h \in \Uc_h$ to the classical FEM problem. Moreover, this solution satisfies
    \begin{align}
        \| u - u_h \|_{\Uc} &\le \frac{M}{\gamma} \inf_{v_h \in \Uc_h} \| u - v_h \|_{\Uc} \,,
        \label{eq:Cea-best-approximation-inequality}
    \end{align}
    where $u \in \Uc$ is the (exact) solution to the abstract weak formulation (Eq.~\eqref{eq:abstract-weak-formulation}), $\gamma$ is the coercivity constant from Eq.~\eqref{eq:coercivity-defn}, and $M$ is the continuity constant from Eq.~\eqref{eq:continuity-defn}. Observe that the constant $\frac{M}{\gamma}$ depends only on the bilinear form $B$ and is independent of the mesh parameter $h$.
\end{prop}\noindent
In the inequality~\eqref{eq:Cea-best-approximation-inequality}, the term $\inf_{v_h \in \Uc_h} \| u - v_h \|_{\Uc}$ represents the \textit{best approximation error} in the finite element space $\Uc_h$. C\'{e}a's Lemma guarantees that the solution to the classical FEM problem is ``close'' to the optimal solution to the abstract weak formulation constrained on the subspace $\Uc_h$. This inequality is a fundamental tool for many theoretical convergence rate and error estimates in the FEM.

\section{The Block-Encoding Framework}
\label{sec:block-encoding-framework}

In this section, we collect many well-established results regarding the block-encoding framework \cite{gilyenQuantumSingularValue2019a,Nguyen_2022,lin2022lecture}. In particular, we discuss the product, linear combination, and tensor product of block-encodings and their associated costs, and show how this can be extended to rectangular matrices. We also give proofs for some of the lemmas and theorems presented in the main text. As a reminder, recall the definition of a block-encoding:
\begin{defn}[$(\alpha,m,\epsilon)$-Block-Encoding \cite{gilyenQuantumSingularValue2019a,lin2022lecture}]
    For an $n$-qubit matrix $A$, we say that the $(n+m)$-qubit unitary matrix $U_A$ is an $(\alpha,m,\epsilon)$-block-encoding of $A$ if
    \begin{equation}
        ||A - \alpha (\bra{0}^{\otimes m}\otimes I_n) U_A (\ket{0}^{\otimes m}\otimes I_n)|| \leq \epsilon \,,
    \end{equation}
    in which case we write $U_A \in (\alpha,m,\epsilon)\mathrm{-BE}(A)$. 
    The parameter $\alpha > 0$ is called the \textbf{subnormalization factor}, while the integer $m \ge 0$ signifies the number of ancilla needed to implement the block-encoding.
    If we can implement the above unitary $U_A$ exactly, then we call it an $(\alpha,m)$-block-encoding of $A$, and write $U_A \in (\alpha,m)\mathrm{-BE}(A)$.
    \label{def:block-encoding-repeat}
\end{defn}
Note that for exact block-encodings, $U_A \in (\alpha,m)\mathrm{-BE}(A)$ if and only if component-wise we have: $A_{ij} = \alpha\bra{0}^{\otimes m}\bra{i} U_A \ket{0}^{\otimes m} \ket{j}$. 
If we wish to apply $A$ to some vector $\ket{b}$, observe that for $U_A \in (\alpha,m)\mathrm{-BE}(A)$:
\begin{align}
    U_A\ket{0}^{\otimes m}\ket{b}_n = \frac{1}{\alpha}\ket{0}^{\otimes m} (A\ket{b}_n) + \ket{\mathrm{rubbish}}_{m+n}
    \,,
\end{align}
where $\ket{\mathrm{rubbish}}_{m+n}$ is a ``garbage'' state that is orthogonal to the subspace $\ket{0}^{\otimes m}$ (i.e., $(\bra{0}^{\otimes m} \otimes I_n) \ket{\mathrm{rubbish}}_{m+n} = 0$). Thus, the probability of successfully applying the block-encoding is $||\frac{1}{\alpha} A \ket{b}_n||^2 = \frac{1}{\alpha^2} \bra{b} A^\dag A \ket{b}$, which is upper bounded by the spectral norm $(||A||/\alpha)^2$. It follows that we must have $\alpha \ge ||A||$ in order to be able to embed the matrix $A$ inside of a unitary. A block-encoding is said to be \textit{optimal} if $\alpha = \Theta(||A||)$.

If the matrix $A$ is rectangular, say $A \in \bbC^{2^s \times 2^t}$, then we can pad $A$ with zeros to render it square, after which all rectangular matrix arithmetic will become equivalent to performing (square) matrix arithmetic with the padded matrix. Let $n = \max\{s,t\}$. Then we define the zero-padded $2^n \times 2^n$ matrix as 
\begin{align}
    \zpad(A)_{ij} &= \begin{cases}
        A_{ij} & \text{if } i < 2^s \text{ and } j < 2^t\\
        0 & \text{otherwise}
    \end{cases}
    \,.
    \label{eq:zpad-defn}
\end{align}
We then proceed by block-encoding the padded matrix.

\begin{defn}[$(\alpha,m,\epsilon)$-(Rectangular)-Block-Encoding]
    Let $A$ be a $2^s \times 2^t$ matrix, and let $n := \max\{s,t\}$. We say that an $n$-qubit unitary $U_A$ is an $(\alpha,m,\epsilon)$-block-encoding of the rectangular matrix $A$ if it is an $(\alpha,m,\epsilon)$-block-encoding of $\zpad(A)$. If we can implement the unitary exactly (i.e., if $\epsilon = 0$), then we call it an $(\alpha,m)$-block-encoding of $A$, and write $U_A \in (\alpha,m)\mathrm{-BE}(A)$.
    \label{def:rectangular-block-encoding}
\end{defn}
The rectangular matrices that we consider in this work are both sparse and low-rank. Block-encodings of such matrices are often easier to explicitly define in terms of unitaries (sometimes called ``oracles'') that describe the structure of the matrix, as well as an oracle that gives its entries.

\subsection{Linear Combinations and Products of Block-Encoded Matrices}
Performing standard matrix arithmetic such as addition and multiplication using the block-encoding framework requires that we demonstrate a block-encoding of the product/sum utilizing only the original block-encodings and additional ancilla qubits (if required). 

In the main text, we outline the Linear Combination of Unitaries (LCU) routine to take the sum/difference of block-encodings. Here, we provide a proof for it.

\begin{defn}[State preparation pair, adapted from
\cite{gilyenQuantumSingularValue2019a,Nguyen_2022}]
\label{defn:state-prep-pair}
    Let $y \in \bbC^m$ with $||y||_1 \le \beta$. The pair of unitaries $(P_L,P_R)$ is called a $(\beta,b,\epsilon)$-state-preparation-pair if $P_L\ket{0}^{\otimes b} = \sum_{j \in [2^b]} c_j \ket{j}$ and $P_R\ket{0}^{\otimes b} = \sum_{j \in [2^b]} d_j \ket{j}$ where $\sum_{j \in [m]} |\beta c_j^* d_j - y_j| \le \epsilon$ and $c_j^* d_j = 0$ for all $j \in \{m,\dots,2^b - 1\}$.
\end{defn}
Observe that the ``prepare'' oracles $(\widetilde{\textsc{prep}}^\dag , \textsc{prep})$ (Eqs.~\eqref{eq:prep-L-dag} and~\eqref{eq:prep-R}) used in the LCU routine form a state preparation pair for the vector of LCU coefficients. Additionally, the definition of a state preparation pair can be extended to accommodate cases where $P_L$ and $P_R$ are non-unitary~\cite{Nguyen_2022} (and thus are themselves block-encodings). 
By forming $\LCU_\epsilon\left((U_{A_j})_{j \in [m]}, \beta \right)$, we immediately obtain a method for block-encoding the sum of block-encoded matrices.

\begin{lem}[Linear combination of block-encoded matrices, adapted from {\cite[Lemma~52]{gilyenQuantumSingularValue2019a}}]
\label{lem:linear-combination-of-block-encodings}
    Let $A := \sum_{j \in [m]} y_j A_j$ be a linear combination of $s$-qubit matrices. Let $\beta \ge ||y||_1$, and suppose $(P_L,P_R)$ is a $(\beta,b,\epsilon_p)$-state-preparation-pair for $y$. For all $j \in [m]$, let $U_j$ be an $(\alpha,a,\epsilon)$-block-encoding of $A_j$, and for $j \in \{m,\dots,2^{b}-1\}$ let $U_j = I_s$. Assume access to the select oracle $W := \sum_{j \in [2^b]} \ketbra{j}{j} \otimes U_j$. Then the unitary $\widetilde{W} := (P_L^\dag \otimes I_a \otimes I_s) W (P_R \otimes I_a \otimes I_s)$ is an $(\alpha\beta,a+b,\alpha\epsilon_p + \beta\epsilon)$-block-encoding of $A$. 
\end{lem}

\begin{proof}
    Observe that 
    \begin{align*}
        &\norm{ A - \alpha\beta (\bra{0}^{\otimes b} \bra{0}^{\otimes a} \otimes I_s) \widetilde{W} (\ket{0}^{\otimes b} \ket{0}^{\otimes a} \otimes I_s) }\\
        &= \norm{ A - \alpha \sum_{j \in [m]} \beta (c_j^*d_j) (\bra{0}^{\otimes a} \otimes I_s) U_j (\ket{0}^{\otimes a} \otimes I_s) }\\
        &\le \alpha \epsilon_p + \norm{ A - \alpha \sum_{j \in [m]} y_j (\bra{0}^{\otimes a} \otimes I_s) U_j (\ket{0}^{\otimes a} \otimes I_s) }\\
        &\le \alpha \epsilon_p + \sum_{j \in [m]} |y_j| \norm{ A_j - \alpha (\bra{0}^{\otimes a} \otimes I_s) U_j (\ket{0}^{\otimes a} \otimes I_s) }\\
        &\le \alpha \epsilon_p + \sum_{j \in [m]} |y_j| \epsilon\\
        &\le \alpha \epsilon_p + \beta \epsilon
    \end{align*}
    so that $\widetilde{W} \in (\alpha\beta,a+b,\alpha\epsilon_p + \beta\epsilon)\mathrm{-BE}(A)$, as desired.
\end{proof}

Note that if the subnormalization factors of the block-encodings $U_j$ of the $A_j$ in the above lemma are different, say $U_j \in (\alpha_j,a)\mathrm{-BE}(A_j)$, then we can embed $\frac{\alpha_j}{\alpha} \le 1$ into the LCU coefficients $y_j$ to obtain the desired linear combination of block-encoded matrices. A block-encoding of the product of two block-encoded matrices can be obtained by simply composing the block-encodings, but without sharing the ancilla qubits:

\begin{lem}[Product of block-encoded matrices, adapted from {\cite[Lemma~53]{gilyenQuantumSingularValue2019a}}]
\label{lem:product-of-block-encodings}
    If $U$ is an $(\alpha, a, \delta)$-block-encoding of an $s$-qubit matrix $A$, and $V$ is a $(\beta, b, \epsilon)$-block-encoding of an $s$-qubit matrix $B$, then $(I_b \otimes U)(I_a \otimes V)$ is an $(\alpha\beta, a+b, \alpha\epsilon + \beta\delta)$-block-encoding of $AB$. Note that here, $I_a$ (respectively, $I_b$) acts on the ancilla qubits of $U$ (respectively, $V$).
\end{lem}

\begin{proof}
\begin{align*}
    &\| AB - \alpha\beta (\bra{0}^{\otimes(a+b)} \otimes I_s)(I_b \otimes U)(I_a \otimes V)( \ket{0}^{\otimes(a+b)} \otimes I_s)\| \\
    &= \| AB - 
    \underbrace{\alpha (\bra{0}^{\otimes a} \otimes I_s)U( \ket{0}^{\otimes a} \otimes I_s)}_{=:\Tilde{A}}
    \underbrace{\beta( \bra{0}^{\otimes b} \otimes I_s)V( \ket{0}^{\otimes b} \otimes I_s)}_{=:\Tilde{B}} \| \\
    &= \| AB - \Tilde{A}B + \Tilde{A}B - \Tilde{A}\Tilde{B} \|\\
    &= \| (A - \Tilde{A})B + \Tilde{A}(B - \Tilde{B}) \|\\
    &\leq \| A - \Tilde{A} \|\beta + \alpha\| B - \Tilde{B} \|\\
    &\leq \alpha\epsilon + \beta\delta
    \,.
\end{align*}
\end{proof}
In terms of quantum circuitry, we can represent the product of block-encodings as:
\tikzexternalenable
\begin{align}
\label{eq:product-of-block-encodings}
\begin{quantikz}
    \lstick{$\ket{0}^{\otimes a}$} & \qwbundle{a} & &  \swap{1} & & \swap{1} & \\
    \lstick{$\ket{0}^{\otimes b}$} & \qwbundle{b} & \gate[2]{V} & \targX{} & \gate[2]{U} & \targX{} &  \\
    & \qwbundle{s} & & & & & 
\end{quantikz}
\quad .
\end{align}
\tikzexternaldisable
Note that since a successful application of the block-encoding requires that we measure all ancilla to be in the $\ket{0}$ state, the last swap operation in Eq.~\eqref{eq:product-of-block-encodings} can be left out. Finally, we note that the state-preparation pair $(P_L,P_R)$ in Lemma~\ref{lem:linear-combination-of-block-encodings} can actually be a pair of block-encodings instead of unitaries, in which case the product of unitaries that forms $\widetilde{W}$ is actually a product of block-encodings, which we can use Lemma~\ref{lem:product-of-block-encodings} for.

\begin{cor}[Tensor product of block-encoded matrices]
\label{cor:product-of-block-encodings}
    If $U_A$ is an $(\alpha, a, \delta)$-block-encoding of an $n$-qubit matrix $A$, and $U_B$ is a $(\beta, b, \epsilon)$-block-encoding of an $m$-qubit matrix $B$, then $U_A \otimes U_B$ is an $(\alpha\beta, a+b, \alpha\epsilon + \beta\delta)$-block-encoding of $A \otimes B$.
\end{cor}

\begin{proof}
    Write $U_A \otimes U_B = (U_A \otimes I^{\otimes (b + m)} ) (I^{\otimes (a + n)} \otimes U_B )$, then apply Lemma~\ref{lem:product-of-block-encodings}. If desired, swap gates can be added to move the ancilla qubits to the top of the circuit (though this does not change the subnormalization or number of ancilla of the block-encoding).
\end{proof}

\subsection{Block-Encoding of Oracle Access Models}
The block-encoding framework is extremely powerful, and allows one to implement almost any matrix function 
on the quantum computer via frameworks such as Quantum Signal Processing (QSP) \cite{Low_2017}, the Quantum Singular Value Transform (QSVT), Quantum EigenValue Transformation (QEVT) \cite{low2024quantumeigenvalueprocessing}, LCU, and many others. Much work has been done to provide efficient block-encodings given certain ``oracle'' access. In general, demonstrating explicit implementations for these oracles can be just as difficult. For linear problems in structured domains, however, we are able to implement these oracles in the main text.

\begin{defn}[Amplitude oracle of a matrix, adapted from \cite{lin2022lecture}]
\label{defn:amplitude-oracle}
    Let $A \in \bbC^{2^n \times 2^n}$, and suppose that $\norm{A}_\mathrm{max} \le 1$ (if $\norm{A}_\mathrm{max} \le 1$ then rescale the matrix as $A/\alpha$ for some $\alpha \ge \norm{A}_\mathrm{max}$). A unitary $O_A$ is called an amplitude oracle for $A$ if it satisfies 
    \begin{align}
        O_A \ket{0}\ket{i}\ket{j} &= \left( A_{ij}\ket{0} + \sqrt{1 - |A_{ij}|^2} \ket{1} \right)\ket{i}\ket{j} \,,
    \end{align}
    where the implementation of $O_A$ may require some additional ancilla registers that are not shown here.
\end{defn}

In practice, the amplitude oracle may only be implemented to some precision $\epsilon$, either because the controlled rotations required to encode the matrix components in the amplitudes can only be done to some precision, or because the coefficients themselves are only known up to some precision. For simplicity, we will assume here that all amplitude oracles are exact.

The matrices that append the Lagrange multipliers to the system of equations are all $1$-sparse, so we begin with these block-encodings.

\begin{lem}[Block-encoding of $1$-sparse oracle-access matrices, adapted from \cite{lin2022lecture}]
\label{lem:BE-1-sparse-oracle-access}
    Let $A \in \bbC^{2^n \times 2^n}$ be a $1$-sparse matrix. Then for each $j \in [2^n]$, there exists a unique $c(j) \in [2^n]$ such that $A_{c(j),j} \ne 0$. The mapping $c$ is a permutation of the rows of the identity matrix, so there exists a unitary $O_c$ (which may require some additional ancilla) such that 
    \begin{align}
        O_c \ket{j} &= \ket{c(j)} \,.
        \label{eq:1-sparse-row-oracle}
    \end{align}
    Suppose that we have access to an amplitude oracle of the form 
    \begin{align}
        O_A\ket{0}\ket{j} &= \left( A_{c(j),j}\ket{0} + \sqrt{1 - |A_{c(j),j}|^2} \ket{1} \right)\ket{j}
        \,.
        \label{eq:1-sparse-amplitude-oracle}
    \end{align}
    Then $U_A := (I \otimes O_c)O_A$ is a $(1,1)$-block-encoding of $A$.
\end{lem}

\begin{proof}
    We can verify that:
    \begin{align*}
        \bra{0}\bra{i} U_A \ket{0}\ket{j}
        &= \bra{0}\bra{i} (I \otimes O_c) \left( A_{c(j),j}\ket{0} + \sqrt{1 - |A_{c(j),j}|^2} \ket{1} \right)\ket{j}\\
        &= \bra{0}\bra{i}\left( A_{c(j),j}\ket{0} + \sqrt{1 - |A_{c(j),j}|^2} \ket{1} \right)\ket{c(j)}\\
        &= A_{c(j),j} \delta_{i,c(j)} = A_{ij} \,,
    \end{align*}
    which shows that $U_A \in (1,1)\mathrm{-BE}(A)$.
\end{proof}

\subsection{Quantum Signal Processing}
\label{sec:QSP}

Quantum Signal Processing (QSP)~\cite{Low_2017,low2019hamiltonian,gilyenQuantumSingularValue2019a}
allows us to enact a polynomial transformation on the eigenvalues of an $N \times N$ Hermitian matrix $A$ (a process known as the Quantum Eigenvalue Transformation (QET)). That is, writing the eigendecomposition $A = U \Lambda U^\dag$ where $U$ is unitary and $\Lambda = \operatorname{diag}(\lambda_0, \dots, \lambda_{N-1})$ is a diagonal matrix containing the eigenvalues of $A$, a scalar function $f \colon \bbR \to \bbC$ can be implemented as a matrix function as 
\begin{align}
    f(A) &:= U \left( \sum_{i \in [N]}f(\lambda_i)\ket{i}\bra{i}\right) U^\dagger
    \,.
    \label{eq:matrix-function}
\end{align}
When the matrix $A$ is not Hermitian, the matrix function $f(A)$ can be implemented using the Quantum Singular Value Transform (QSVT)~\cite{gilyenQuantumSingularValue2019a}, except that the singular values of $A$ are transformed rather than its eigenvalues. More recent~\cite{low2024quantumeigenvalueprocessing} work allows for the Quantum EigenValue Transformation (QEVT) of diagonalizable (but not necessarily Hermitian) matrices.
The QSP protocol meshes well with the Block-Encoding Framework (which we recap in Section~\ref{sec:block-encoding-framework}) in that we can enact a QET of $A$ using a block-encoding $U_A \in (\alpha,m)\mathrm{-BE}(A)$ with a cost that scales with the degree $D$ of the polynomial approximation as $O(D)$. 

Due to qubitization~\cite{gilyenQuantumSingularValue2019a}, it is sufficient to consider $\operatorname{SU}(2)$ matrices. Firstly, consider the Hermitian block-encoding of the scalar $x \in [-1,1]$ 
\begin{align}
    U(x) &:= \begin{bmatrix}
        x & \sqrt{1-x^2}\\
        \sqrt{1-x^2} & -x
    \end{bmatrix}\,.
\end{align}
We can map this into a rotation matrix $O(x)$ (i.e., an element of $\operatorname{SU}(2)$) by applying a Pauli $Z$ gate 
\begin{align}
    O(x) &:=  U(x) Z = \begin{bmatrix}
        x & -\sqrt{1-x^2}\\
        \sqrt{1-x^2} & x
    \end{bmatrix}\,.
\end{align}
QSP then guarantees the existence of a transformation that implements a block-encoding of the polynomial $p(x)$.

\begin{thm}[Quantum Signal Processing, adapted from {\cite[Theorem~7.20]{lin2022lecture}}]
    There exists a vector $\boldsymbol{\Phi} := (\phi_0, \ldots, \phi_{D}) \in \mathbb{R}^{D+1}$, where $\phi_j$ are called the ``phase factors'', such that
    \begin{equation}
    \begin{aligned}
        U_\Phi(x) &:= e^{i \phi_0 Z}\prod_{j=1}^D[O(x)e^{i\phi_j Z}]\\
        &=
        \begin{pmatrix}
            P(x) & - Q(x)\sqrt{1-x^2}\\
            Q^*(x)\sqrt{1-x^2} & P^*(x)
        \end{pmatrix},
    \end{aligned}
    \end{equation}
    if and only if $P,Q \in \mathbb{C}[x]$ satisfy
    \begin{enumerate}
        \item $\deg(P) \leq d, \deg(Q) \leq d-1$,
        \item $P$ has parity $d \pmod{2}$ and $Q$ has parity $d-1 \pmod{2}$, and
        \item $|P(x)|^2 + (1-x^2)|Q(x)|^2 =1 \hspace{.2cm}\forall x \in[-1,1]$.
    \end{enumerate}
    In this context, a polynomial with $\deg Q = -1$ means $Q \equiv 0$.
\end{thm}
Using qubitization, we can implement the matrix polynomial $p(A)$ of a Hermitian matrix $A$ using a Hermitian block-encoding $U_A$ as 
\tikzexternalenable
\begin{align}
\label{eq:QSP-circuit}
\begin{quantikz}[column sep=6pt]
    \lstick{$\ket{0}$} & 
    \targ{} & \gate{e^{-i\phi_D Z}} & \targ{} & & 
    \targ{} & \gate{e^{-i\phi_{D-1} Z}} & \targ{} & &
    \ \cdots\  & 
    & \targ{} & \gate{e^{-i\phi_0 Z}} & \targ{} &\\
    \lstick{$\ket{0}^{\otimes m}$} & 
    \octrl{-1} & & \octrl{-1} & \gate[2]{O} & 
    \octrl{-1} & & \octrl{-1} & \gate[2]{O} &
    \ \cdots\  & 
    \gate[2]{O} & \octrl{-1} & & \octrl{-1} & \\
    \lstick{$\ket{\psi}_n$} & 
    &&&& 
    &&&& 
    \ \cdots\  & 
    &&&& 
\end{quantikz}
\end{align}
\tikzexternaldisable
where 
\tikzexternalenable
\begin{align}
\label{eq:O-rotation-circuit}
\begin{quantikz}[column sep=6pt]
    \lstick{$\ket{0}^{\otimes m}$} & \gate[2]{O} & \\
    \lstick{$\ket{\psi}_n$} & & 
\end{quantikz}
\equiv
\begin{quantikz}[column sep=6pt]
    \lstick{$\ket{0}$} & 
    \targ{} & \targ{} & \gate{Z} & \targ{} & \targ{} & \\
    \lstick{$\ket{0}^{\otimes m}$} & & 
    \octrl{-1} & & \octrl{-1} & \gate[2]{U_A} & \\
    \lstick{$\ket{\psi}_n$} & 
    & &&&&
\end{quantikz}
\,.
\end{align}
\tikzexternaldisable
This circuit requires only one additional ancilla qubit, which can be thrown away after the computation is over.

Both determination of the polynomial approximation $p$ to $f$ and its corresponding phase factors are performed classically. Currently, there exist algorithms that perform this task for polynomials of degree up to $D = O(10^7)$~\cite{dong2021efficient,motlagh2024generalized}. Since both ancilla qubits in Eqs.~\eqref{eq:QSP-circuit} and~\eqref{eq:O-rotation-circuit} can be thrown away at the end of the computation, QSP gives an $(||p||_{L^\infty\left( (-\alpha,\alpha) \right)},m)\mathrm{-BE}(p(A))$.

\begin{rem}[QSP on a Unitary]
    Given a Hermitian unitary $U$, we can $p(U)$ by performing QSP on the ``$(1,1)$-block-encoding'' $\ketbra{0}{0} \otimes U + \ketbra{1}{1} \otimes I$, which is just the controlled application of $U$.
\end{rem}

\subsection{Multivariate Quantum Eigenvalue Transformation}
\label{sec:MQET}

Given a function $f \colon \Omega \to \bbR$ defined on a $d$-dimensional mesh, and a degree $D$ (multivariate) polynomial approximation (determined through classical means) $p \colon \bbR^d \to \bbR$, we outline how to use the Multivariate Quantum Eigenvalue Transformation (MQET) \cite{borns2023MQET} in conjunction with the position operators $X^{(i)}$ to access a general mesh function $f$ to some desired accuracy $\epsilon > 0$. First recall the definition of the MQET problem.

\begin{defn}[Multivariate Quantum Eigenvalue Transformation (MQET), adapted from \cite{borns2023MQET}]
    Let $\Ab := \{A^{(\ell)}\}_{\ell \in [d]}$ be a family of pairwise commuting diagonalizable $N \times N$ matrices. Since we will access the family $\Ab$ through the block-encoding framework, assume without loss of generality that $\|A_\ell\| \le 1$ for all $\ell \in [d]$. Let $\{\bm{v}_k\}_{k \in [N]}$ be a mutual eigenbasis of $\Ab$, and let $\lambdab_k := (\lambda_{\ell,k})_{\ell \in [d]} \in \bbD^{d}$ denote the vector of eigenvalues corresponding to the eigenvector $\bm{v}_k$ for each operator $A_\ell$ (here, $\bbD$ is the closed unit disk in $\bbC$). Given a function $f \colon \bbC^{d} \to \bbC$, a state $\ket{\bm{v}} := \sum_{k \in [N]} c_k \ket{\bm{v}_k}$, and $\epsilon > 0$, the MQET problem is to output the state $f(\Ab)\ket{\bm{v}} := \sum_{k \in [N]} c_k f(\lambdab_k)\ket{\bm{v}_k}$ (up to a global scaling factor) to precision $\epsilon$ . 
\end{defn}

In our case, we have the family of commuting position operators $\Xb := \{X^{(i)}\}_{i \in [d]}$, which have the position basis vectors $\ket{j}$ for $\jb \in [N]^d$ as eigenvectors, and the lattice points $\frac{\jb}{N-1}$ as the corresponding eigenvalues. Ref.~\cite{borns2023MQET} gives an algorithm that solves the MQET problem with a subnormalization factor that scales as $D^d$. We will summarize the MQET algorithm here, with a slight modification that utilizes the compression gadget~\cite{fang2023time} to reduce the number of ancilla required to apply the block-encodings of the position operators.

The MQET algorithm uses QSP as a subroutine, which can only transform eigenvalues that lie in the interval $[-1,1]$ since we are accessing matrices through the block-encoding framework. Thus, we will consider the domain $\Omega = [-1,1]^d$, and a continuous function $f \colon [-1,1]^d \to \bbD$. For more general continuous functions $f \colon \Omega \to \bbC$ on some compact $d$-dimensional domain $\Omega \subseteq \bbR^d$, we assume that we can homeomorphically map $\Omega$ to some subset of $\hat{\Omega} \subseteq [-1,1]^d$ through $h \colon \Omega \to \hat{\Omega}$. We further assume that $(f \circ h^{-1})/\|f\|_{L^\infty(\Omega)} \colon \hat{\Omega} \to \bbD$ has a continuous extension $\Tilde{f}$ onto $[-1,1]^d$, and relabel this function as $f$. This pull-back of $f$ onto $[-1,1]^d$ is closely related to the concept of isoparametric mapping, which we leave for future work.

Let $g$ be a polynomial approximation to $f$ of degree less than $D$ in each variable, which may be computed for example through the techniques of \cite[Appendix~A]{borns2023MQET}. Decompose $g$ as 
\begin{align}
    g(\xb) &= \sum_{\sbm \in [D]^{d-1}} Q_{\sbm}(x^{d-1})  \prod_{k \in [d-1]} P_{\sbm}^{(k)}(x^k) \,,
    \label{eq:g-MQET-decomposition}
\end{align}
where $\xb = (x^0,\dots,x^{d-1}) \in \bbC^d$, $P_{\sbm}^{(k)},Q_{\sbm} \in \bbC[z]$ are polynomials with complex coefficients, and $\sbm = (s_0,\dots,s_{d-2})$ is a $D$-ary string of length $d-1$. The polynomials $P_{\sbm}^{(k)}$ are chosen to be the $s_k$th Chebyshev polynomials of the first kind $T_{s_k}$, while $Q_{\sbm}$ are defined using the orthogonal projection 
\begin{align}
    Q_{\sbm}(x^{d-1}) &:= \left( \prod_{\ell \in [d-1]} \frac{2 - \delta_{s_\ell,0}}{\pi}\right) 
    \int_{[-1,1]^{d-1}} g(\xb) \left( \prod_{k \in [d-1]} T_{s_k}(x^k) \cdot \frac{1}{\sqrt{1 - (x^k)^2}} \right) \,d\xb 
    \,.
\end{align}
Ref.~\cite{borns2023MQET} shows that this choice of polynomials guarantees that $\beta_{\sbm} := \|P_{\sbm}^{(k)}\|_{L^\infty([-1,1])} \|Q_{\sbm}\|_{L^\infty([-1,1])} \le 2$, so that the LCU that forms Eq.~\eqref{eq:g-MQET-decomposition} has a subnormalization factor that's bounded by $(D+2)^{d-1} = O(D^{d-1})$. Now, choose block-encodings\footnote{In principle, these block-encodings could belong to $(\alpha_k,m_k,\epsilon_k)\mathrm{-BE}(A^{(k)})$, but we choose exact block-encodings with the same sub-normalization factors and number of ancilla here for convenience.} $U^{(k)} \in (\alpha,m)\mathrm{-BE}(A^{(k)})$ for a family of pairwise commuting operators $\Ab := \{A^{(\ell)}\}_{\ell \in [d]}$.
For convenience, define 
\begin{align}
    T(\sbm,\Ab) &:= Q_{\sbm}(A^{(d-1)}) \prod_{k \in [d-1]} P_{\sbm}^{(k)}(A^{(k)}) \,,
\end{align}
so that $g(\Ab) = \sum_{\sbm \in [D]^{d-1}} T(\sbm,\Ab)$. 
Compute the QSP phase factors $\Phib_{\sbm}^{(k)},\Psib_{\sbm}$ that implement the normalized polynomials $\Tilde{P}_{\sbm}^{(k)} := \dfrac{P_{\sbm}^{(k)}}{\|P_{\sbm}^{(k)}\|_{L^\infty([-1,1])}}$ and $\Tilde{Q}_{\sbm} := \dfrac{Q_{\sbm}}{\|Q_{\sbm}\|_{L^\infty([-1,1])}}$ to obtain the block-encodings $U_{\Tilde{P}_{\sbm}^{(k)}(A^{(k)})} \in (\alpha,m)\mathrm{-BE}\left(\Tilde{P}_{\sbm}^{(k)}(A^{(k)})\right)$ and $U_{\Tilde{Q}_{\sbm}(A^{(d-1)})} \in (\alpha,m)\mathrm{-BE}\left(\Tilde{Q}_{\sbm}(A^{(d-1)})\right)$. 
By taking the product of block-encodings (see Lemma~\ref{lem:product-of-block-encodings}), form the block-encoding $U_{T(\sbm,\Ab)} \in (\alpha^d,dm)\mathrm{-BE}\left(T(\sbm,\Ab)\right)$. Finally, taking the LCU over all $D$-ary bit strings $\sbm \in [D]^{d-1}$ yields 
\begin{align}
    U_{g(\Ab)} = \LCU\left( (T(\sbm,\Ab))_{\sbm \in [D]^{d-1}}, \betab \right) \in (\alpha^d \|\betab\|_1, dm + \log(D))\mathrm{-BE}(g(\Ab)) \,,
    \label{eq:MQET-LCU}
\end{align}
where $\betab \in \bbR_{\ge 0}^{[D]^{r-1}}$ is the vector of the LCU coefficients $\beta_{\sbm}$. Ref.~\cite{borns2023MQET} shows that by choosing $P_{\sbm}^{(k)} = T_{s_k}$, the one-norm of the LCU coefficients is bounded by $\|\betab\|_1 \le (D + 2)^{d-1}$.

\section{Assembly of Block Linear Systems}
\label{sec:block-linear-systems}

Let $A \in \bbC^{2^n \times 2^n}$, $B \in \bbC^{2^n \times 2^b}$ 
with $b < n$ for simplicity. Identify the matrix $B$ with the $n$-qubit matrix $\zpad(B)$ (see Eq.~\eqref{eq:zpad-defn}). 
For the purposes of imposing Dirichlet boundary conditions, we want to demonstrate a block-encoding for the block matrix 
\begin{align}
    M &:= \begin{bmatrix}
    A & B\\
    B^\dag & 0
\end{bmatrix} \,.
\label{eq:block-system-Hermitian}
\end{align}
More generally, we will also demonstrate how to obtain a block-encoding of the block matrix 
\begin{align}
    \Ab &:= \begin{bmatrix}
    A_{(00)} & A_{(10)}\\
    A_{(10)} & A_{(11)}
\end{bmatrix} \,,
\label{eq:block-system-general}
\end{align}
where $A_{(ij)}$ are $n$-qubit matrices. Additionally, given two $n$-qubit vectors $\ket{f_0}$ and $\ket{f_1}$, we want to prepare the vector 
\begin{align}
    \ket{\fb} &:= \frac{1}{\sqrt{2}}\begin{bmatrix}
        \ket{f_0}\\
        \ket{f_1}
    \end{bmatrix}
    \,.
\end{align}
Altogether, we will then have prepared the Quantum Linear Systems Problem (QLSP) $\Ab \ket{\xb} = \ket{\fb}$. The case where we have more than $4$ blocks in the partitioned matrix $\Ab$ is a straightforward generalization of the methods in Sections~\ref{subsec:BE-partitioned-matrices} and~\ref{subsec:partitioned-vectors}.

\subsection{Block-Encoding Partitioned Matrices}
\label{subsec:BE-partitioned-matrices}

We will begin with a block-encoding of the more general Eq.~\eqref{eq:block-system-general}, and then demonstrate how it can be modified to obtain a slightly more efficient block-encoding of Eq.~\eqref{eq:block-system-Hermitian}.

\begin{lem}
\label{lem:Mij-block-encodings}
    Let $A_{(ij)} \in \bbC^{2^n \times 2^n}$ for $i,j \in [2]$, and let $U_{A_{(ij)}} \in (\alpha_{ij},m,\epsilon_{ij})\mathrm{-BE}(A_{(ij)})$. Then 
    \begin{align}
        U_{M_{ij}} &:= 
        \begin{bmatrix}
            0 & I_n & 0 & 0\\
            I_n & 0 & 0 & 0\\
            0 & 0 & I_n & 0\\
            0 & 0 & 0 & I_n
        \end{bmatrix}^i
        \begin{bmatrix}
            U_{A_{(ij)}} & 0 & 0 & 0\\
            0 & I_n & 0 & 0\\
            0 & 0 & I_n & 0\\
            0 & 0 & 0 & I_n
        \end{bmatrix}
        \begin{bmatrix}
            I_n & 0 & 0 & 0\\
            0 & 0 & I_n & 0\\
            0 & I_n & 0 & 0\\
            0 & 0 & 0 & I_n
        \end{bmatrix}
        \begin{bmatrix}
            0 & I_n & 0 & 0\\
            I_n & 0 & 0 & 0\\
            0 & 0 & I_n & 0\\
            0 & 0 & 0 & I_n
        \end{bmatrix}^j\\
        &= (c_{0}\text{-}X \otimes I_{m+n})^i(c_{00}\text{-}U_{A_{(ij)}})(\mathrm{SWAP} \otimes I_{m+n})(c_{0}\text{-}X \otimes I_{m+n})^j
    \end{align}
    is an $(\alpha_{ij},m+1,\epsilon_{ij})\mathrm{-BE}(\ketbra{i}{j} \otimes A_{(ij)})$. In terms of quantum circuitry, we can represent this as 
    \tikzexternalenable
    \begin{align}
        \begin{quantikz}
            \lstick{$\ket{0}$} && \octrl{1} 
            \gategroup[2,steps=1,style={dashed,rounded
            corners,fill=blue!0, inner
            xsep=2pt},label style={anchor=north west,xshift=0.3cm,yshift=-0.cm},background]{j}
            && \swap{1} & \octrl{2} & \octrl{1}
            \gategroup[2,steps=1,style={dashed,rounded
            corners,fill=blue!0, inner
            xsep=2pt},label style={anchor=north west,xshift=0.3cm,yshift=-0.cm},background]{i}
            & \\
            && \targ{} && \targX{} & \octrl{0} & \targ{} & \\
            \lstick{$\ket{0}^{\otimes m}$}& \qwbundle{m} & && & \gate[2]{U_{A_{(ij)}}} & & \\
            & \qwbundle{n} & && & & &
        \end{quantikz}
        \,.
    \end{align}
    \tikzexternaldisable
\end{lem}

\begin{proof}
    Observe that for any $k,\ell \in [2]$, 
    \begin{align*}
        &(\bra{0} \otimes \bra{k} \otimes \bra{0}^{\otimes m} \otimes I_n) U_{M_{ij}} (\ket{0} \otimes \ket{\ell} \otimes \ket{0}^{\otimes m} \otimes I_n)\\
        \begin{split}
        &=(\bra{0} \otimes \bra{k} \otimes \bra{0}^{\otimes m} \otimes I_n)(c_{0}\text{-}X \otimes I_{m+n})^i(c_{00}\text{-}U_{A_{(ij)}})\\
        &\qquad\quad
        (\mathrm{SWAP} \otimes I_{m+n})(c_{0}\text{-}X \otimes I_{m+n})^j
        (\ket{0} \otimes \ket{\ell} \otimes \ket{0}^{\otimes m} \otimes I_n)
        \end{split}\\
        &=(\bra{0} \otimes \bra{i \oplus k} \otimes \bra{0}^{\otimes m} \otimes I_n)(c_{00}\text{-}U_{A_{(ij)}})
        (\ket{j \oplus \ell} \otimes \ket{0} \otimes \ket{0}^{\otimes m} \otimes I_n)\\
        &= \braket{0 | j \oplus \ell} \braket{i \oplus k | 0} (\bra{0}^{\otimes m} \otimes I_n) U_{A_{(ij)}} (\ket{0}^{\otimes m} \otimes I_n)\\
        &= \delta_{i,k} \delta_{j,\ell} \tilde{A}_{(ij)} / \alpha_{ij} \,,
    \end{align*}
    where ``$\oplus$'' denotes addition modulo $2$.
    Since the final result is $\epsilon$-close to $\delta_{i,k} \delta_{j,\ell} A_{(ij)} / \alpha_{ij}$ in spectral norm, we are done. 
\end{proof}

By taking an LCU of the block-encodings $U_{M_{ij}}$ that appear in Lemma~\ref{lem:Mij-block-encodings}, we can form a block-encoding for the system in Eq.~\eqref{eq:block-system-general}.

\begin{lem}[Block-encoding of a general partitioned matrix (Eq.~\eqref{eq:block-system-general})]
\label{lem:BE-of-partitioned-matrix}
    Let $A_{(ij)} \in \bbC^{2^n \times 2^n}$ for $i,j \in [2]$, and $U_{M_{ij}} \in (\alpha_{ij},m+1,\epsilon_{ij})\mathrm{-BE}(\ketbra{i}{j} \otimes A_{(ij)})$ as in Lemma~\ref{lem:Mij-block-encodings}. Define $\betab := \begin{bmatrix}
        \alpha_{00} & \alpha_{01} & \alpha_{10} & \alpha_{11}
    \end{bmatrix}^\dag$ and suppose that we have a $(\sum_{i,j \in [2]} \alpha_{ij},2,\epsilon_p)$-state-preparation-pair for $\betab$. Then 
    \begin{align}
        \LCU_{\epsilon_p}\left((U_{M_{ij}})_{i,j \in [2]}, \betab \right) \in 
        \left(\sum_{i,j \in [2]} \alpha_{ij}, m + 3, \epsilon_p + \sum_{i,j \in [2]} \epsilon_{ij} \right)\mathrm{-}BE\left(\Ab\right)
        \,.
    \end{align}
\end{lem}

\begin{proof}
    Use Lemma.~\ref{lem:linear-combination-of-block-encodings} together with Lemma.~\ref{lem:Mij-block-encodings}.
\end{proof}

We now turn our attention to block-encoding $M := \begin{bmatrix}
    A & B\\
    B^\dag & 0
\end{bmatrix}$ given access to $U_A \in (\alpha,a,\epsilon_a)\mathrm{-BE}(A)$ and $U_B\in (\beta,b,\epsilon_b)\mathrm{-BE}(B)$.
Using an ancilla qubit, we will block-encode an enlarged system $\tilde{M}$ using the LCU 
\begin{align}
    \tilde{M} &:= \begin{bmatrix}
        A & B & 0 & 0\\
        B^\dag & 0 & I_n & 0\\
        0 & I_n & 0 & I_n\\
        0 & 0 & I_n & I_n
    \end{bmatrix}\\
    &= 
    \begin{bmatrix}
        A & 0 & 0 & 0\\
        0 & I_n & 0 & 0\\
        0 & 0 & I_n & 0\\
        0 & 0 & 0 & I_n
    \end{bmatrix}
    \begin{bmatrix}
        I_n & 0 & 0 & 0\\
        0 & 0 & I_n & 0\\
        0 & I_n & 0 & 0\\
        0 & 0 & 0 & I_n
    \end{bmatrix}
    + 
    \begin{bmatrix}
        0 & I_n & 0 & 0\\
        I_n & 0 & 0 & 0\\
        0 & 0 & 0 & I_n\\
        0 & 0 & I_n & 0
    \end{bmatrix}
    \begin{bmatrix}
        B^\dag & 0 & 0 & 0\\
        0 & B & 0 & 0\\
        0 & 0 & I_n & 0\\
        0 & 0 & 0 & I_n
    \end{bmatrix}\\
    &= \underbrace{ (c_{00}\text{-}A) (\mathrm{SWAP} \otimes I_n) }_{\Dc} + \underbrace{ (I \otimes X \otimes I_n)(c_{00}\text{-}B^\dag)(c_{01}\text{-}B) }_{\Oc}
\end{align}
where we've omitted the subnormalization factors. For the first term $\Dc$, Lemma~\ref{lem:Mij-block-encodings} gives a $U_{\Dc} \in (\alpha,a+1,\epsilon_a)\mathrm{-BE}(\Dc)$. For the second term $\Oc$, by modifying the proof of the Lemma to account for the product of block-encodings $(c_{00}\text{-}U_B^\dag)(c_{01}\text{-}U_B)$, we have the $(\beta,b+1,\epsilon_b)$-block-encoding $U_\Oc$ given by 
\tikzexternalenable
\begin{align}
    \begin{quantikz}
        \lstick{$\ket{0}$} && \octrl{1} & \octrl{1} & & \\
        && \octrl{1} & \ctrl{1} & \targ{} & \\
        \lstick{$\ket{0}^{\otimes b}$} & \qwbundle{b} & \gate[2]{U_B^\dag} & \gate[2]{U_B} & &\\
        & \qwbundle{n} & & & &
    \end{quantikz}
\end{align}
\tikzexternaldisable
By letting $m := \max\{a, b\}$ and appending ancilla qubits to the top of the circuits that form $U_\Dc$ or $U_\Oc$, we have $U_\Dc \in (\alpha,m+1,\epsilon_a)\mathrm{-BE}(\Dc)$ and $U_\Oc \in (\beta,m+1,\epsilon_b)\mathrm{-BE}(\Oc)$. Then following Lemma~\ref{lem:BE-of-partitioned-matrix}, taking the LCU using an $(\alpha + \beta, 1, \epsilon_p)$-state-preparation-pair for $\betab := \begin{bmatrix}
    \alpha & \beta
\end{bmatrix}^\dag$ yields 
\begin{align}
    \LCU_{\epsilon_p}\left( (U_\Dc, U_\Oc) , \betab \right) \in 
    \left(\alpha + \beta, m + 2, \epsilon_p + \alpha + \beta \right)\mathrm{-}BE\left(\tilde{M}\right)
    \,,
\end{align}
which is a slightly more efficient result than the $\left(\alpha + 2\beta, m + 3, \epsilon_p + \alpha + 2\beta \right)$-block-encoding given in the lemma.

\subsection{Preparing Partitioned Vectors}
\label{subsec:partitioned-vectors}

Suppose that we have access to two state preparation unitaries $U_0,U_1$ that satisfy 
\begin{align}
    U_0 \ket{0}^{\otimes n} &= \ket{f_0}
    \quad \text{and} \quad
    U_1 \ket{0}^{\otimes n} = \ket{f_1} \,.
\end{align}
Then the vector $\ket{\fb} := \frac{1}{\sqrt{2}}\left(\ket{0} \otimes \ket{f_0} + \ket{1} \otimes \ket{f_1} \right)$ can be prepared using the LCU circuit, but without any post-selection, i.e., 
\tikzexternalenable
\begin{align}
\begin{quantikz}
    \lstick{$\ket{0}$} && \gate{H} & \octrl{1} & \ctrl{1} & \rstick[2]{$\ket{\fb}$} \\
    \lstick{$\ket{0}^{\otimes n}$} & \qwbundle{n} && \gate{U_0} & \gate{U_1} &
\end{quantikz}
\,.
\end{align}
\tikzexternaldisable

If we wish to access the superposition of the vectors $\ket{f_0}$ and $\ket{f_1}$, we can use Lemma~\ref{lem:BE-of-partitioned-matrix} to obtain a $(2,1)$-block-encoding of $\begin{bmatrix}
    U_0 & U_1\\
    0 & 0
\end{bmatrix}$, which yields the state $\frac{1}{\sqrt{2}}\left(\ket{0} + \ket{1}\right) \otimes \frac{1}{\sqrt{2}} \left( \ket{f_0} + \ket{f_1} \right)$ when successfully applied to the state $\frac{1}{\sqrt{2}}\left(\ket{0} + \ket{1}\right)\otimes \ket{0}^{\otimes n}$ (after which we may throw away the first qubit). Normally, the addition of quantum states from different Hilbert spaces is forbidden \cite{alvarez2015forbidden}. Using LCU, however, we can access the ``forbidden quantum adder'', albeit with a success probability of $\frac{1}{2}$.

\section{Incorporating Dirichlet Boundary Conditions Using Projectors}
\label{sec:projector-bcs}

In this section, we give an alternative method for enforcing Dirichlet boundary conditions to the algorithm detailed in Section~\ref{subsec:assembling-Lagrange-multipliers}. The method herein incorporates the boundary conditions directly into the system of equations by modifying the linear operator, and thus more closely resembles the classical method for enforcing boundary conditions~\cite{ciarlet2023finite}.

In Section~\ref{subsec:weak-formulation-modified-Poisson}, we show that at the end of the assembly procedure, we have a linear system of the form 
\begin{align}
    \Lc \ub &= \bb
    \,,
    \label{eq:assembled-linear-system}
\end{align}
where $\Lc$ is a linear operator.
In classical implementations of the finite element method, Dirichlet boundary conditions are typically enforced strongly by deleting rows in the final assembled system of equations (Eq.~\eqref{eq:assembled-linear-system}) and replacing each of them with the known value of the solution. That is, to enforce the condition $x_i = \bar{x}_i$, the $i$-th row of $\Lc$ is replaced with the row vector $\eb_i^\dagger$ and the $i$-th entry of $\bb$ is replaced with the known value $\bar{x}_i$. 
This same procedure may be done on a quantum computer by introducing a projection operator $\bbP_\mathrm{int}$ that indicates the interior of the domain (defined in Eq.~\eqref{eq:interior-dof-projector}). 
The operator $\bbP_\mathrm{bd} := (I - \bbP_\mathrm{int})$ is then a projection onto the boundary degrees of freedom, and the known values can be incorporated into any vector $\bar{\ub}$ that satisfies 
\begin{align}
    (I - \bbP_\mathrm{int}) (\ub - \bar{\ub}) = \zerob
    \,.
    \label{eq:Dirichlet-bc-projection-condition}
\end{align}
Left-multiplying Eq.~\eqref{eq:assembled-linear-system} by $\bbP_\mathrm{int}$ and incorporating Eq.~\eqref{eq:Dirichlet-bc-projection-condition}, we can rearrange the assembled system as
\begin{align}
    \underbrace{ 
    \left[\bbP_\mathrm{int} \Lc \bbP_\mathrm{int} + (I - \bbP_\mathrm{int}) \right]
    }_{\Lc_\mathrm{Dirich.}} \ub &= \underbrace{ 
    \bbP_\mathrm{int}\bb - (\Lc - I)(I - \bbP_\mathrm{int})\bar{\ub}
    }_{\bb_\mathrm{Dirich.}} 
    \,.
    \label{eq:Dirichlet-incorporated-assembled-system}
\end{align}
The operator $\Lc_\mathrm{Dirich.}$ defined above is modified so that the fixed degrees of freedom no longer interact with the system, as the forcing vector $\bb_\mathrm{Dirich.}$ on the right-hand side has been modified to include the effects of the fixed boundary conditions. Note that although Eq.~\eqref{eq:Dirichlet-incorporated-assembled-system} can be simplified (using Eq.~\eqref{eq:Dirichlet-bc-projection-condition}) to 
$ (\bbP_\mathrm{int} \Lc \bbP_\mathrm{int}) \ub = \bbP_\mathrm{int}\bb - \Lc (I - \bbP_\mathrm{int})\bar{\ub}$, this renders the operator on the left-hand side of the equation rank-deficient. Block-encoding the projection operator $\bbP_\mathrm{int}$, and then taking the product and LCU of block-encodings to form the system in Eq.~\eqref{eq:Dirichlet-incorporated-assembled-system} enforces any Dirichlet boundary conditions in the FEM problem. 
Finally, we note that Eqs.~\eqref{eq:Dirichlet-incorporated-assembled-system} and~\eqref{eq:FEM-block-system} are equivalent.

\end{appendices}

\end{document}